\documentclass[11pt,letterpaper]{article}
\usepackage{authblk}

\usepackage[margin=1in]{geometry}
\usepackage{amsmath,amssymb,amsthm,mathtools}
\usepackage{qcircuit}
\usepackage{xcolor}
\usepackage{graphicx}
\usepackage{placeins}
\usepackage{array,booktabs,longtable,enumitem}
\usepackage{microtype}

\makeatletter
\@ifpackagelater{microtype}{2025/02/12}{}{%
  \IfFormatAtLeastTF{2025-06-01}{%
    \microtypesetup{expansion=false}%
    \setlength{\emergencystretch}{1em}%
  }{}%
}
\makeatother
\usepackage{tikz}
\usepackage[colorlinks=true,linkcolor=blue!55!black,
  citecolor=red,urlcolor=blue!55!black]{hyperref}

\newtheorem{theorem}{Theorem}[section]
\newtheorem{lemma}[theorem]{Lemma}
\newtheorem{proposition}[theorem]{Proposition}
\newtheorem{corollary}[theorem]{Corollary}

\newtheorem{definition}[theorem]{Definition}

\providecommand{\F}{\mathbb F}

\providecommand{\one}{\mathbf 1}
\providecommand{\Tr}{\operatorname{Tr}}

\providecommand{\wt}{\operatorname{wt}}

\providecommand{\ket}[1]{\left|#1\right\rangle}

\usetikzlibrary{arrows.meta,positioning,calc}
\title{From Simple Sources to Quantum Advantage: Homomorphic Polynomial Transduction via Relative Decoding}

\author[1]{Zhong-Xia Shang%
  \thanks{Email: \texttt{zhongxia.shang@math.ku.dk}}}

\author[1]{Daniel Stilck Fran\c{c}a%
  \thanks{Email: \texttt{dsfranca@math.ku.dk}}}

\affil[1]{Department of Mathematical Sciences, University of Copenhagen, Denmark}

\date{}

\usepackage{etoolbox}
\makeatletter
\patchcmd{\@maketitle}
  {\vskip 2em}{\vskip .5em}
  {}{\PackageWarning{layout}{Could not adjust title top spacing}}
\makeatother

\usepackage{titlesec}
\titlespacing*{\section}
  {0pt}{1.8ex plus .4ex minus .2ex}{.9ex plus .2ex}
\titlespacing*{\subsection}
  {0pt}{1.5ex plus .3ex minus .2ex}{.7ex plus .2ex}
\titlespacing*{\subsubsection}
  {0pt}{1.2ex plus .2ex minus .1ex}{.6ex plus .1ex}
\titlespacing*{\paragraph}
  {0pt}{1ex plus .2ex minus .1ex}{1em}

\setlist{topsep=6pt plus 2pt minus 1pt,
         itemsep=3pt plus 1pt minus 1pt,
         parsep=2pt plus 1pt,partopsep=1pt}

\AtBeginDocument{%
  \apptocmd{\normalsize}{%
    \setlength{\abovedisplayskip}{6pt plus 2pt minus 1pt}%
    \setlength{\belowdisplayskip}{6pt plus 2pt minus 1pt}%
    \setlength{\abovedisplayshortskip}{0pt plus 2pt}%
    \setlength{\belowdisplayshortskip}{4pt plus 2pt minus 1pt}%
  }{}{\PackageWarning{layout}{Could not adjust display spacing}}%
  \normalsize
}

\begin{document}

\maketitle
\begin{abstract}
Decoded quantum interferometry (DQI) and its Hamiltonian extension
(HDQI) prepare states whose amplitudes are low-degree polynomials of an
objective, using Fourier transforms and coherent decoding. We recast and generalize this approach as quantum state transduction through algebra
homomorphisms. We transfer an efficiently preparable operator state
of a degree-$D$ polynomial in a source Hamiltonian $H_A$ to the
corresponding polynomial state of a target Hamiltonian $H_B=\pi(H_A)$,
where $\pi$ is a unital $*$-homomorphism between the finite-dimensional
source and target $C^*$-algebras. The transduction uses operator Fourier transforms and relative decoding
to recover source operators rather than individual term-selection labels. The relative distance $d_{\mathrm{rel}}$ is the first degree at which
source and target traces disagree. We prove that $2D<d_{\mathrm{rel}}$
is equivalent to preserving all inner products between degree-$D$
polynomials. Moreover, $d_{\mathrm{rel}}\ge d_{\mathrm{ord}}$, where
$d_{\mathrm{ord}}$ is the ordinary distance used in (H)DQI, and we give
families with $d_{\mathrm{ord}}=O(1)$ but
$d_{\mathrm{rel}}=\Theta(n)$ and efficient decoders. Our framework replaces the complicated pilot state preparation by the more modular task
of preparing a source polynomial state. DQI and HDQI arise as relation-free special cases. The framework
accommodates nonuniform coefficients and noncommuting interactions
and extends to fermionic and qudit systems, with applications
to approximate optimization and Gibbs sampling. As evidence of advantages, for a nonlinear pairwise variant of optimal polynomial intersection
where constant $d_{\mathrm{ord}}$ limits DQI-style preparations,
relative decoding yields an ideal quantum score of 0.643 versus
0.606 for the best tested classical heuristic, a gap exceeding
3 percentage points.

\end{abstract}
\clearpage
\tableofcontents
\clearpage
\section{Overview}
\label{sec:introduction}

\subsection{Introduction and previous works}

Decoded quantum interferometry (DQI)~\cite{JordanEtAlDQI}
connects approximate optimization to error-correcting codes
through a quantum reduction related to Regev's
construction~\cite{Regev2009}. It prepares a weighted
superposition of term selections, computes their syndromes,
and coherently decodes to erase the selection register.
A Fourier transform then produces amplitudes proportional
to a low-degree polynomial of the objective, biasing
measurement toward high-scoring solutions. The attainable
polynomial degree and score are therefore tied to decoding
performance. Optimal polynomial intersection (OPI) provides
a structured application based on Reed--Solomon decoding with apparent exponential quantum speedup.
Hamiltonian DQI (HDQI)~\cite{HamiltonianDQI} extends this
mechanism to polynomial filters of quantum Hamiltonians,
with applications to optimization and Gibbs-state preparation. 

For extensions, Bu, Gu, and Li extend HDQI to general Pauli Hamiltonians
with nonuniform coefficients, using a more involved
pilot-state construction~\cite{GeneralHDQI2026}.
Their multivariate DQI framework also handles weighted
objectives~\cite{WeightedDQI2026}. Other extensions broaden the decoding primitives and
underlying code families. Chailloux and Tillich improve
OPI variants using soft Reed--Solomon decoding and a
reduction from syndrome decoding to coset
sampling~\cite{ChaillouxTillich2024SoftDecoders}.
Gu and Jordan introduce Hermitian OPI using
algebraic-geometry codes~\cite{GuJordan2025AGDQI},
while COFI-DQI studies further curve-based intersection
problems~\cite{MatthewsShapiro2026COFI}.
Rank-metric and translation-association schemes extend
DQI beyond Hamming geometry~\cite{RankDQI2026}.
Quantum-decoding advances include efficient optimal
decoders for codes with efficient trellis
representations~\cite{PiveteauRenes2025BPQM},
code-aware affine filtering~\cite{MandalEtAl2026AffineFiltering},
and locally quantum decoders~\cite{ShuttyEtAl2026LocallyQuantum}. Implementation studies provide explicit coherent-decoding
circuits~\cite{PatamawisutEtAl2025DQICircuits} and nearly
linear-time OPI algorithms under suitable input-access
assumptions~\cite{opi-optimized,Rosmanis2026NearlyLinearDQI}.
DQI-Kit automates transformations from constrained
optimization to Max-LINSAT and analyzes their
overhead~\cite{ThelenMauerer2026DQIKit}.
Further work studies sensitivity to local depolarizing
noise~\cite{BuEtAl2026DQINoise} and proposes kernel-based
spectral modifications to improve
robustness~\cite{Wang2025KernelizedDQI}. The scope of an advantage also depends on classical
competition. Complexity-theoretic
analyses~\cite{MarwahaDQIComplexity}, spin-glass
obstructions under specified decoding-threshold
assumptions~\cite{AnschuetzDQIStructure}, and
minimum-distance-based MaxCut
limitations~\cite{ParekhMaxCut} identify distinct
restrictions on DQI. Guide-weight Monte Carlo methods
attain DQI's predicted scores on substantial tested
instances~\cite{GilFusterDQIMCMC}.

In this work, we unify and generalize these constructions by taking an
\emph{operator-algebraic point of view}. The central idea is to prepare a
polynomial in an accessible source system and coherently transport the
same polynomial to a target system, which we call the \emph{polynomial transduction}. Our framework has three components:
a surjective \emph{unital $*$-homomorphism} between represented
finite-dimensional $C^*$-algebras~\cite[Chs.~2--3]{Murphy1990},
operator--Fourier transforms that expose operator coefficients as quantum amplitudes, and coherent
\emph{relative decoding} that uncomputes the source information.

The homomorphism
$\pi:\mathcal A\twoheadrightarrow\mathcal B$ maps marked source
generators $a_i$ to target generators $b_i$, so that
$\pi(q(a))=q(b)$ for every noncommutative $*$-polynomial $q$.
The corresponding quantum primitive, which we call \emph{polynomial
transduction}, is the state conversion
\begin{equation}
 \frac{|q(a)\rangle\!\rangle_A}{\|q(a)\|_{2,A}}
 \;\longmapsto\;
 \frac{|q(b)\rangle\!\rangle_B}{\|q(b)\|_{2,B}}.
 \label{eq:intro-transduction}
\end{equation}
Here $|X\rangle\!\rangle_L=(X\otimes I)|\Phi_L\rangle$ is the vectorized operator state of $X$.

Quotient codes and relative decoding arise naturally in this
framework, paralleling degenerate decoding
in quantum error correction, where errors differing by a stabilizer
need not be distinguished~\cite{Poulin2005,IyerPoulin2015}. Here, this principle
is used for coherent polynomial-state transduction rather than
recovery from physical noise. Algebraically, relative decoding corresponds to inverting the
homomorphism $\pi$ on the low-degree source subspace: given
$Y=\pi(X)$, it recovers the unique source operator $X$ in that
subspace. Since the homomorphism preserves every source relation,
the source relation code $K_A$ is contained in the target relation code
$K_B$. The quotient
$K_B/K_A$ therefore captures the additional target relations
that must be handled by decoding. The resulting relative distance
$d_{\mathrm{rel}}$ sets the polynomial-degree bound for coherent
transduction. We prove the sharp
criterion that $2D<d_{\mathrm{rel}}$ is equivalent to preserving
\emph{all} operator inner products on the degree-$D$ source subspace.
With compatible operator bases and efficient source preparation, operator--Fourier transforms, label computation, and relative decoding,
this condition gives deterministic transduction of
every nonzero degree-$D$ source polynomial state.

This relative decoding
task is different from the ordinary decoding used in DQI and HDQI construction: ordinary
decoding recovers the term-selection label, whereas relative decoding
recovers only the source operator represented by that label. Therefore, inherited source relations do not consume the relative-decoding budget. Consequently, $d_{\mathrm{rel}}\ge d_{\mathrm{ord}}$ for the corresponding
ordinary marking. When the inequality is strict, relative decoding
relaxes the degree bound from $2D<d_{\mathrm{ord}}$ to
$2D<d_{\mathrm{rel}}$, potentially allowing polynomial filters beyond
the ordinary decoding limit. The improvement is particularly
significant when $d_{\mathrm{ord}}$ remains constant while
$d_{\mathrm{rel}}$ grows linearly with system size. Additionally, by reducing the number of labels to be recovered, relative decoding can relax and go beyond the classical-decoding DQI Shannon bound. 

The operator-algebraic viewpoint turns the complicated,
case-dependent task of pilot-state preparation into the more transparent
task of source polynomial-state preparation. To get a target
Hamiltonian polynomial state $p(H_B)$ with $b_i$ marked as the terms in $H_B$, one prepares the source Hamiltonian polynomial state of $p(H_A)$ with $a_i$ marked as the terms in $H_A$,
where $\pi(H_A)=H_B$. Its operator--Fourier image is the required pilot state. Rather than deriving pilot amplitudes through
explicit enumeration of term selections, multiplicities, and commutation
phases~\cite{HamiltonianDQI,Wocjan2026Pilot}, one can therefore apply existing state-preparation methods
whenever they prepare the required source polynomial state~\cite{SchonEtAl2005,Schollwock2011}. Nonuniform
Hamiltonians fit naturally into this picture: for
$H_B=\sum_i\lambda_i b_i$, one simply uses the corresponding nonuniform
source $H_A=\sum_i\lambda_i a_i$, in contrast to the explicit pilot-state constructions in \cite{GeneralHDQI2026,WeightedDQI2026}.

In our framework, DQI and HDQI are special cases of Hamiltonian
polynomial transduction from a relation-free Pauli source.
The source relation code is $K_A=\{0\}$, so $K_B/K_A=K_B$ and
relative decoding reduces to ordinary term-selection decoding.
Because the construction depends on algebraic relations, it applies
beyond qubits: we treat Majorana operator algebras for fermions and
finite-dimensional Weyl operator algebras for qudits and truncated
bosonic systems. Across these settings, the $*$-homomorphic structure governs
polynomial transfer, while source preparation, operator--Fourier
transforms, and relative decoding are adapted to the corresponding
operator algebra.

We generalize the applications introduced in DQI and HDQI: approximate optimization and
thermofield-double (TFD) Gibbs preparation. For optimization, we give a formula
for the best average objective value achievable with polynomial filters
of a given degree. We also identify a limitation shared by decoded
algorithms, including DQI and HDQI: low-degree quantities such as the average target energy can already be evaluated
using the source state, without preparing the target. The potential
advantage lies instead in high-degree quantities such as generating target samples, with
polynomial filters favoring high-scoring solutions or desired energies.
For Gibbs-state preparation, we transfer a source TFD into an
approximate target TFD. The source TFD can be
prepared using compatible purified Gibbs-state algorithms, including the adiabatic preparation from recent Lindbladian-based methods~\cite{ChenKastoryanoBrandaoGilyen2025ThermalSimulation,RouzeFrancaAlhambra2026}. We relate the
attainable temperature and accuracy to the available polynomial degree.
For Hamiltonians with extensive spectral width,
a linear relative distance suffices for constant inverse temperature.

Our numerical example imposes nonlinear tests on neighbouring evaluations
of a low-degree polynomial. Dropping the polynomial constraint gives an
accessible source with independent residuals; relative decoding restores
the constraint after applying a score filter. At $p=2053$, with $205$
polynomial coefficients and $2051$ tests, a degree-$50$ guide has exact
mean satisfaction $0.64313$ (rounded), versus a median best classical
score of $0.60666$ across ten independently sampled residual-coefficient
vectors.
The tested classical methods include interpolation, local search,
annealing, tempering, and inference followed by projection to the code,
with approximately $5\times10^6$ candidates per deep run. For the same
Fourier-term marking, short source relations limit the ordinary-distance
criterion to degree at most one, while the relative decoder supports
this higher-degree guide.

\noindent\textbf{Concurrent works.}
We here briefly summarize the DQI-related works after our arXiv submission. Hastings and Schmidhuber identify the HDQI pilot with
the TFD of a redundancy-free Hamiltonian ~\cite{HastingsSchmidhuber2026RedundancyFree} and explore the complexity of this Hamiltonian class, which is equivalent to the relation-free source Hamiltonian of HDQI we discovered in this paper. Quantum-decoding Gibbs preparation improves the temperatures accessible
through classical and quantum
decoders~\cite{PiveteauQuekFrancoGarrido2026Gibbs}. For classical spin glasses, DQI and a classical Prange sampler can both
sample beyond a shattering barrier that obstructs stable
algorithms~\cite{ZhouEtAl2026ShatteredGibbs}. Cycle-code analyses distinguish imperfect and typical-error decoding
from worst-case minimum-distance
guarantees~\cite{ApteEtAl2026CycleCodes}. Mazumder and Niroula derive candidate-screening criteria and a sharp
DQI--Prange margin under worst-case Hamming decoding, and study further
optimization problems~\cite{MazumderNiroula2026NewDQIProblems}.
Folded OPI admits a proved approximation separation with membership-oracle
access~\cite{KramerEtAl2026OracleDQI}, while streaming formulations yield
quantum space separations~\cite{WuYang2026StreamingDQI,BoulandDingJain2026Streaming}.
Our companion Gibbs paper~\cite{ShangWangFranca2026Gibbs} develops the
Pauli specialization and explicit nonlocal target families using the
general transduction framework established here.

\subsection{Homomorphic polynomial transduction and relative decoding}
\label{intro:sec:main-framework}
Let $H_A=\sum_{i=1}^s \lambda_i a_i$ and $H_B=\sum_{i=1}^s \lambda_i b_i$,
where $a_i$ and $b_i$ are the matched marked interaction terms of the
source and target Hamiltonians, respectively. Our goal is to prepare the normalized operator state of a target
polynomial $q(b_1,\ldots,b_s)$ from an efficiently preparable state of
the same polynomial in the source operators $a_1,\ldots,a_s$.  For
Hamiltonian filters, this specializes to preparing the state of
$p(H_B)$ from that of $p(H_A)$, without applying the polynomial
directly to the target.  The polynomial $p$ is typically chosen or
optimized for a target objective, for example to bias the resulting
state toward low-energy, high-scoring spectral sectors, and to approximate the exponential function for Gibbs sampling.

\subsubsection{Polynomial states, the $*$-homomorphism, and the relative decoding}
\label{intro:subsec:mf-goal-algebra}
We use basic properties of finite-dimensional $C^*$-algebras and unital
$*$-homomorphisms between them~\cite{Murphy1990}.
Let $\mathcal A=C^*(I_A,a_1,\ldots,a_s)\subseteq\mathcal B(\mathcal H_A)$
and $\mathcal B=C^*(I_B,b_1,\ldots,b_s)\subseteq\mathcal B(\mathcal H_B)$
be finite-dimensional represented unital $C^*$-algebras. Given a $D$-degree noncommutative $*$-polynomial $q$, abbreviate its evaluations by
$X_A=q(a)=q(a_1,\ldots,a_s)$ and
$X_B=q(b)=q(b_1,\ldots,b_s)$.
Write $N_L=\dim\mathcal H_L$ for $L=A,B$. The representation
fixes the normalized trace and operator inner product $\tau_L(X)=\Tr_{\mathcal H_L}X/N_L$, $\langle X,Y\rangle_L=\tau_L(X^*Y)$, $\|X\|_{2,L}=\sqrt{\tau_L(X^*X)}$.

We define the vectorized operators~\cite{Watrous2018}
\begin{equation}
 |\Phi_L\rangle=N_L^{-1/2}\sum_j|j,j\rangle,
 \qquad |X\rangle\!\rangle_L=(X\otimes I)|\Phi_L\rangle.
 \label{intro:eq:mf-vectorization}
\end{equation}
Thus $\langle\!\langle X|Y\rangle\!\rangle_L=\tau_L(X^*Y)$. The transduction is from an efficiently preparable source polynomial state $|\psi_A^q\rangle$ to the target polynomial state $|\psi_B^q\rangle$  
\begin{equation}
 |\psi_A^q\rangle=\frac{|q(a)\rangle\!\rangle_A}{\|q(a)\|_{2,A}}\rightarrow|\psi_B^q\rangle=\frac{|q(b)\rangle\!\rangle_B}{\|q(b)\|_{2,B}}.
 \label{intro:eq:mf-output}
\end{equation}

The transduction requires a unital $*$-homomorphism
$\pi:\mathcal A\twoheadrightarrow\mathcal B$ satisfying
$\pi(a_i)=b_i$, and therefore $\pi(H_A)=H_B$. A \emph{marked word} of length $\ell$ is an ordered product of $\ell$
marked generators or their adjoints. Define the degree-at-most-$t$ spaces as
\begin{equation}
 \mathcal F_t^A=\operatorname{span}\{I_A,
 a_{i_1}^{\epsilon_1}\cdots a_{i_\ell}^{\epsilon_\ell}:
 \ell\le t,\ \epsilon_j\in\{1,*\}\},\qquad
 \mathcal F_t^B=\pi(\mathcal F_t^A).
 \label{intro:eq:mf-filtration}
\end{equation}
It follows from the $*$-homomorphism property that $X_A\in\mathcal F_D^A$ and $X_B=\pi(X_A)$.
In particular, a degree-$D$ polynomial $p(H_A)$ has marked degree at most
$D$ for the Hamiltonian-term marking, and $\pi(p(H_A))=p(H_B)$.

The homomorphism identity alone, however, does not
preserve the physical operator inner products. Relative distance
specifies the degree through which it does.

\begin{definition}[Intrinsic relative distance]
\label{intro:def:mf-intrinsic-distance}
The relative distance of the marked, represented homomorphism is
\begin{equation}
 d_{\mathrm{rel}}=
 \min\{t\in\mathbb N_0:\exists Y\in\mathcal F_t^A,
          \ \tau_B(\pi(Y))\ne\tau_A(Y)\},
 \label{intro:eq:mf-abstract-relative-distance}
\end{equation}
with the empty minimum equal to $+\infty$.
\end{definition}
\noindent Equivalently, for every integer $t\ge0$, $\tau_B\circ\pi=\tau_A\text{ on }\mathcal F_t^A
 \quad\Longleftrightarrow\quad t<d_{\mathrm{rel}}$.

\begin{proposition}
\label{intro:thm:mf-degree-distance}
For an integer $D\ge0$, the condition $ 2D<d_{\mathrm{rel}}$ is equivalent to $\tau_B(\pi(X)^*\pi(Y))=\tau_A(X^*Y)$ for $X,Y\in\mathcal F_D^A$. In particular $\pi$ is an isometry from $\mathcal F_D^A$ onto
$\mathcal F_D^B$ in their operator inner products.
\end{proposition}

Fix $D$ with $2D<d_{\mathrm{rel}}$. Choose orthogonal operator bases
$\{E_\alpha^L:\alpha\in\Omega_L\}$, with $E_0^L=I_L$.
For the chosen basis, we ask for two compatibility conditions:
\begin{description}[leftmargin=0pt,labelsep=0.5em]
\item[(C1) Generator--basis compatibility.]
A label set $S_D$ such that $ \mathcal F_D^A=\operatorname{span}\{E_\alpha^A:\alpha\in S_D\}$.
\item[(C2) Single-image basis compatibility.]
For every $\alpha\in\Omega_A$, $\pi(E_\alpha^A)=\chi_\alpha E_{M(\alpha)}^B$.
\end{description}
Under (C1), (C2), and $2D<d_{\mathrm{rel}}$, one has
$|\chi_\alpha|=1$ for every $\alpha\in S_D$, and
$M|_{S_D}$ is injective. 

The relative distance Definition~\ref{intro:def:mf-intrinsic-distance} has the equivalent expression via operator bases. Define $w_A(\alpha)=\min\{t:E_\alpha^A\in\mathcal F_t^A\}$. If, besides C1 and C2, an operator basis is additionally adapted to the word filtration
at every degree: $\mathcal F_t^A
 =
 \operatorname{span}\{E_\alpha^A : w_A(\alpha)\le t\}$ for every $t\geq0$, the relative distance is
\begin{equation}
 d_{\mathrm{rel}}=\min\{w_A(\alpha):\alpha\ne0,
 \chi_\alpha\ne0,\ M(\alpha)=0\}.
\end{equation}

\subsubsection{The four-step transduction workflow and the polynomial transduction
\label{intro:subsec:mf-four-step-circuit}}
In the degree-$D$ setting with $2D<d_{\mathrm{rel}}$, if $ \frac{q(a)}{\|q(a)\|_{2,A}}
   =\sum_{\alpha\in S_D}x_\alpha E_\alpha^A$ with $\sum_{\alpha\in S_D}|x_\alpha|^2=1$,
then under C1 and C2, we also have $\frac{q(b)}{\|q(b)\|_{2,B}}
   =\sum_{\alpha\in S_D}\chi_\alpha x_\alpha E_{M(\alpha)}^B$. In the following, the registers are denoted by $\mathsf S$ (source) and $\mathsf T$ (target)

\paragraph{Step 1. Operator--Fourier transform of the efficient source.}
For $L=A,B$, the operator--Fourier transform is defined by $\mathsf F_L|E_\alpha^L\rangle\!\rangle_L=|\alpha\rangle$. Starting from the polynomial source
guide, $F_A$ produces $|\psi_A^q\rangle|0\rangle_{\mathsf T}
 \xrightarrow{\ \mathsf F_A\ }
 \sum_{\alpha\in S_D}x_\alpha
     |\alpha\rangle_{\mathsf S}|0\rangle_{\mathsf T}$.

\paragraph{Step 2. Compute the target label and apply its phase.}
Define two unitaries, $\mathsf U_M:
 |\alpha\rangle_{\mathsf S}|0\rangle_{\mathsf T}
 \longmapsto
 |\alpha\rangle_{\mathsf S}|M(\alpha)\rangle_{\mathsf T}$, and $\mathsf U_\chi|\alpha\rangle=\chi_\alpha|\alpha\rangle$. Then the implementaton of these two unitaries gives $\sum_{\alpha\in S_D}\chi_\alpha x_\alpha
 |\alpha\rangle_{\mathsf S}|M(\alpha)\rangle_{\mathsf T}$.

\paragraph{Step 3. Uncompute the source label by relative decoding.}
We require a relative decoder on $M(S)$:
\begin{equation}
 \mathsf{Dec}_{A\leftarrow B}:
 |M(\alpha)\rangle_{\mathsf T}|0\rangle_{\mathsf R}
 \longmapsto
 |M(\alpha)\rangle_{\mathsf T}|\alpha\rangle_{\mathsf R}
 \qquad(\alpha\in S).
 \label{intro:eq:mf-basis-decoder}
\end{equation}
Injectivity proves that this answer is unique and makes the uncomputation valid. We obtain $|0\rangle_S\sum_{\alpha\in S_D}\chi_\alpha x_\alpha|M(\alpha)\rangle_T$.

\paragraph{Step 4. Inverse operator--Fourier transform of the target.} Ignoring the source register, an inverse target operator--Fourier transform finishes the transduction $ \sum_{\alpha\in S}\chi_\alpha x_\alpha|M(\alpha)\rangle_{\mathsf T}
 \xrightarrow{\ \mathsf F_B^{-1}\ }
 |\psi_B^q\rangle$.

\begin{theorem}
Given efficiently preparable source $|\psi_A^q\rangle$, the target $|\psi_B^q\rangle$ can be efficiently prepared if $F_A$ and $F_B$ are efficiently implementable, $M(\alpha)$ and $\chi(\alpha)$ are efficiently calculable, and the relative decoder is efficient.
\end{theorem}

\subsubsection{Polynomial transduction in spin systems}
\label{intro:subsec:mf-relative-codes}
We now specialize to Pauli Hamiltonian transduction (Fig.~\ref{intro:fig:relative-decoding-framework}). This case makes the
transduction framework and the relative decoding concrete.

\begin{figure}[htbp]
\centering
\begin{tikzpicture}[scale=0.9, transform shape,
  x=1cm,y=0.9cm,>=Latex,
  font=\fontsize{9}{11}\selectfont,
  flow/.style={-Latex,semithick},
  panel/.style={draw,rounded corners=7pt,semithick,align=center,
                inner xsep=5pt,inner ysep=7pt},
  coeff/.style={panel,fill=orange!12,draw=orange!60!black,
                text width=5.35cm},
  state/.style={align=center,inner sep=2pt},
  operation/.style={align=center,text width=4.25cm,
                     font=\fontsize{9}{11}\selectfont}
]

  \node[panel,fill=blue!8,draw=blue!55!black,text width=7.2cm,
        minimum height=2.1cm] (source) at (12.25,7.25)
    {\textbf{Source polynomial guide}\\[2mm]
     $\displaystyle |\psi_A^p\rangle=
       \frac{|p(H_A)\rangle\!\rangle_A}{\|p(H_A)\|_{2,A}}$\\[2mm]
     purification of $\rho_A^p=p(H_A)^2/\Tr p(H_A)^2$};

  \node[coeff,minimum height=1.7cm] (ca) at (2,7.25)
    {\textbf{Source coefficients}\\[1mm]
     $\displaystyle\sum_{u\in S_D(A)}x_u|u\rangle$\\[1mm]
     DQI/HDQI: $K_A=\{0\}$};

  \node[state] (s0) at (1.7,5.25)
    {$|u\rangle|0\rangle$};
  \node[state] (s1) at (1.7,3.75)
    {$\chi(u)|u\rangle|Mu\rangle$};
  \node[state] (s2) at (1.7,2.25)
    {$\chi(u)|0\rangle|Mu\rangle$};

  \draw[flow] (ca.south -| s0.north) -- (s0.north);
  \draw[flow] (s0.south) -- (s1.north);
  \node[operation] at (5.0,4.5)
    {\textbf{2.} Compute the target label $Mu$\\
     and the phase $\chi(u)$};

  \draw[flow] (s1.south) -- (s2.north);
  \node[operation,text=red!65!black] at (5.0,3.0)
    {\textbf{3.} Relative decoder $Bc\mapsto Ac$\\
     erases the source label\\
     $2D<d_{\mathrm{rel}}$};

  \node[panel,draw=black!25,fill=black!2,text width=7.2cm,
        minimum height=2.35cm] at (12.25,4.05)
    {$\pi(a_i)=b_i,\qquad 2D<d_{\mathrm{rel}}$\\[1mm]
     $u=Ac,\quad Mu=Bc,\quad\ker M\simeq K_B/K_A$};

  \node[coeff,minimum height=1.35cm] (cb) at (2,0.75)
    {\textbf{Target coefficients}\\[1mm]
     $\displaystyle\sum_{u\in S_D(A)}
       \chi(u)x_u|Mu\rangle$};

  \node[panel,fill=green!8,draw=green!45!black,text width=7.2cm,
        minimum height=2.1cm] (target) at (12.25,0.75)
    {\textbf{Target polynomial output}\\[2mm]
     $\displaystyle |\psi_B^p\rangle=
       \frac{|p(H_B)\rangle\!\rangle_B}{\|p(H_B)\|_{2,B}}$\\[2mm]
     purification of $\rho_B^p=p(H_B)^2/\Tr p(H_B)^2$};

  \draw[flow] (s2.south) -- (s2.south |- cb.north);

  \draw[flow] (source.west) -- (ca.east)
    node[midway,above=2pt] {$1.\;\mathsf F_A$}
    node[midway,below=2pt,font=\fontsize{8}{9.5}\selectfont]
      {$|x,z\rangle_A\leftarrow|P(x,z)\rangle_A$};

  \draw[flow] (cb.east) -- (target.west)
    node[midway,above=2pt] {$4.\;\mathsf F_B^{-1}$}
    node[midway,below=2pt,font=\fontsize{8}{9.5}\selectfont]
      {$|x,z\rangle_B\rightarrow|P(x,z)\rangle_B$};
\end{tikzpicture}
\caption{Pauli Hamiltonian polynomial transduction as an illustration.}
\label{intro:fig:relative-decoding-framework}
\end{figure} 

Consider a
source Hamiltonian and a target Hamiltonian
\begin{equation}
 H_A=\sum_{i=1}^s \lambda_i a_i,
 \qquad
 H_B=\sum_{i=1}^s \lambda_i b_i,
 \qquad
 \lambda_i>0,
 \label{intro:eq:mf-pauli-hamiltonians}
\end{equation}
whose indexed signed Pauli terms are $a_i=\eta_i^A P_{n_A}(w_i)$ and $b_i=\eta_i^B P_{n_B}(v_i)$, with $P_n(x,z):=i^{x\cdot z}\bigotimes_{j=1}^n X_j^{x_j}Z_j^{z_j}$. The individual Hamiltonian
terms $a_i$ and $b_i$ are the marked generators of the source and
target algebras,
\begin{equation}
 \mathcal A=C^*(I,a_1,\ldots,a_s),
 \qquad
 \mathcal B=C^*(I,b_1,\ldots,b_s).
\end{equation}
To characterize when a homomorphism $\pi(a_i)=b_i$ exists, collect the binary
Pauli labels of the Hamiltonian terms into~\cite{DehaeneDeMoor2003}
\begin{equation}
 A=[w_1\ \cdots\ w_s],
 \qquad
 B=[v_1\ \cdots\ v_s],
 \qquad
 K_A:=\ker A,
 \qquad
 K_B:=\ker B.
 \label{intro:eq:mf-pauli-codes}
\end{equation}
All kernels and matrix operations in the label space are over
$\mathbb F_2$.  A vector $c\in K_L$ specifies a product of indexed
Hamiltonian terms whose projective Pauli label vanishes.  Accordingly,
for $c\in\mathbb F_2^s$ define the ordered products $ W_A(c):=a_1^{c_1}\cdots a_s^{c_s}$ and $W_B(c):=b_1^{c_1}\cdots b_s^{c_s}$. If $c\in K_L$ with $L=A,B$, then this product is scalar, and we write
\begin{equation}
 W_L(c)=\zeta_L(c)I.
\end{equation}

\begin{theorem}\label{intro:theorem:pauli}
The termwise assignment $a_i\mapsto b_i$ extends to a unital
$*$-homomorphism $\pi:\mathcal A\to\mathcal B$ exactly when
\begin{equation}
 \underbrace{
 A^{\mathsf T}J_{n_A}A
 =
 B^{\mathsf T}J_{n_B}B
 }_{\text{anticommutation}},
 \qquad
 \underbrace{
 K_A\subseteq K_B
 }_{\text{kernel}},
 \qquad
 \underbrace{
 \zeta_A(c)=\zeta_B(c)
 \quad(c\in K_A)}_{\text{phase}},
\end{equation}
where
$J_n:=\left(\begin{smallmatrix}
0&I_n\\
I_n&0
\end{smallmatrix}\right)$.
\end{theorem}
The relative distance here becomes
\begin{equation}
d_{\mathrm{rel}}=d_{A\to B}:=\min\{|c|:c\in K_B\setminus K_A\}.
 \label{intro:eq:mf-relative-distance}
\end{equation}
Here $|c|$ is Hamming
weight. For operator basis, let $S_D(A):=\{u\in\operatorname{im}A:w_A(u)\le D\}$ with $w_A(Ac):=\min_{k\in K_A}|c+k|$ and $\mathcal F_D^A
 =\operatorname{span}\{P_{n_A}(u):u\in S_D(A)\}$. The homomorphism acts by
\begin{equation}
 \pi(P_{n_A}(u))=\chi(u)P_{n_B}(M(u)),\qquad \chi(u)\in\{\pm1\}.
\end{equation}
The sign is obtained by choosing any $c$ with $Ac=u$ and comparing
the ordered source and target word phases. C1 and C2 are satisfied.
Here, $M$ maps the quotient of label spaces
\begin{equation}
 M:\mathbb F_2^s/K_A\longrightarrow\mathbb F_2^s/K_B,
 \qquad [c]_{K_A}\longmapsto[c]_{K_B},
 \label{intro:eq:mf-code-quotient}
\end{equation}
or, in Pauli labels, $M(Ac)=Bc$. 

Let $p$ be a $D$-degree polynomial, under the homomorphism, the transduction is
\begin{equation}
|p(H_A)\rangle=\sum_{u\in S_D(A)}x_u|P_{n_A}(u)\rangle\rightarrow |p(H_B)\rangle=\sum_{u\in S_D(A)}\chi(u)x_u|P_{n_B}(Mu)\rangle.
\end{equation}
The operator Fourier transform used here is the generalized Bell transform $|P_n(x,z)\rangle\rightarrow|x,z\rangle$. For the relative decoding, the task is $|Mu\rangle\rightarrow |u\rangle$ for $u\in S_D(A)$, which is equivalent to
\begin{equation}
 \mathsf{Dec}_{A\leftarrow B}:|Bc\rangle|0\rangle
 \longmapsto |Bc\rangle|Ac\rangle
 \qquad(|c|\le D).
 \label{intro:eq:mf-relative-decoder}
\end{equation}
When $2D<d_{\mathrm{rel}}$, the decoder's output is unique on
the degree-$D$ support and thus the uncomputation is valid.

\paragraph{Relative versus relation-free ordinary distance.}
The ordinary distance of the target relation code is
\begin{equation}
 d_B:=\min\{|c|:0\ne c\in K_B\}.
 \label{intro:eq:mf-ordinary-target-distance}
\end{equation}
It governs recovery of vector $c$ from $Bc$, rather than recovery
only of its source class $Ac$.

\begin{theorem}[Relative decoding has higher distance]
\label{intro:prop:mf-relative-vs-ordinary}
\begin{equation}
 d_{\mathrm{rel}}=d_{A\to B}\ge d_B,
 \qquad K_A=\{0\}\ \Longrightarrow\ d_{A\to B}=d_B.
 \label{intro:eq:mf-ordinary-relative-comparison}
\end{equation}
If $d_B<\infty$, strict improvement holds exactly when every
minimum-weight target relation is already a source relation
\begin{equation}
 d_{A\to B}>d_B
 \quad\Longleftrightarrow\quad
 \{c\in K_B:|c|=d_B\}\subseteq K_A.
 \label{intro:eq:mf-strict-distance-improvement}
\end{equation}
More generally, for compatible source choices with kernels
$K_A\subseteq K_{A'}\subseteq K_B$,
\begin{equation}
 d_{A\to B}\le d_{A'\to B}.
 \label{intro:eq:mf-source-kernel-monotonicity}
\end{equation}
\end{theorem}
\noindent Therefore, by absorbing relations into source preparation, we can enlarge the relative distance and therefore enable higher-degree polynomial transduction. Theorem~\ref{intro:prop:mf-relative-vs-ordinary} is not restricted to Paulis but is the intrinsic property of our homomorphic transduction framework (See Section~\ref{sec:main-framework}).

\subsection{DQI and HDQI as relation-free case}
\label{intro:subsubsec:spin-dqi}
We here explain how DQI and HDQI serve as special relation-free examples of our transduction framework. We call a source Hamiltonian relation-free when $ K_A=\{0\}$. For a relation-free source, the relative decoding reduces to the ordinary decoding with $d_{A\to B}=d_B$. Both DQI and HDQI can be viewed as relation-free source-to-target Hamiltonian polynomial transduction.

For the max-XORSAT problem considered in DQI~\cite{JordanEtAlDQI}, the final target state can be viewed as the polynomial state of the target Hamiltonian $H_B=\sum_{i=1}^s b_i$ with $b_i=(-1)^{v_i}Z^{\beta_i}$ on $n_B$ target qubits, where
$\beta_i\in\mathbb F_2^{n_B}$. For the pilot state, the superposition of Dicke states on $s$ source qubits, after doing the Fourier transform, it is the polynomial state of the simple relation-free source Hamiltonian $H_A=\sum_{i=1}^s a_i$ with $a_i=Z_i$.

For HDQI~\cite{HamiltonianDQI}, suppose the uniform target Hamiltonian is $H_B=\sum_{i=1}^s b_i=\sum_{i=1}^s \eta_i^B P_{n_B}(v_i)$ with $\eta_i^B\in\{\pm1\}$, then the complicated pilot state of HDQI, after the Fourier transform, is the polynomial state of the relation-free $s$-qubit source Hamiltonian
\begin{equation}
 H_A:=\sum_{i=1}^s a_i,
 \qquad a_i:=X_i\prod_{j<i}Z_j^{\Sigma_{ij}},\qquad  \Sigma_{ij}:=v_i^{\mathsf T}J_{n_B}v_j.
 \label{intro:eq:mf-hdqi-lift}
\end{equation}

\subsection{Relative decoding relaxes the DQI Shannon bound}
Source relations plus relative decoding can relax the Shannon bound for DQI with classical decoding.
For binary max-XORSAT with full-row-rank \(B\in\mathbb F_2^{n_B\times s}\) and \(n_B/s\to R\), the asymptotic ($s\to\infty$) DQI bound from semi-circle law is \(\mu_{\mathrm{DQI}}\le\overline\mu_{\mathrm{DQI}}(R):=\frac12+\sqrt{H_2^{-1}(R)[1-H_2^{-1}(R)]}\), where \(H_2\) is binary entropy~\cite{JordanEtAlDQI}.
For a compatible source \(K_A\subseteq K_B\), relative decoding recovers only the labels in \(S_D(A):=\{Ac:|c|\le D\}\), whose number can be much smaller than the \(\sum_{j=0}^{D}\binom{s}{j}\) selection vectors.
Let \(D_{\max}\) be the largest integer \(0\le D\le s/2\) satisfying the necessary counting condition \(\log_2|S_D(A)|\le n_B\) for exact recovery, and set \(\delta_{\mathrm{rel}}:=\limsup_{s\to\infty}D_{\max}/s\).
Under the same semicircle law, the asymptotic relative-decoding score satisfies \(\mu_{\mathrm{rel}}\le\overline\mu_{\mathrm{rel}}:=\frac12+\sqrt{\delta_{\mathrm{rel}}(1-\delta_{\mathrm{rel}})}\).
Since source relations reduce the number of labels, \(\delta_{\mathrm{rel}}\ge H_2^{-1}(R)\), giving \(\overline\mu_{\mathrm{rel}}\ge\overline\mu_{\mathrm{DQI}}(R)\).

\subsection{Low-degree vs high-degree observables}
\label{subsec:low-degree-observables}

Let $0\ne X_A\in\mathcal F_D^A$, $X_B=\pi(X_A)$, and therefore the reduced density matrix
$\rho_L:=X_LX_L^*/\operatorname{Tr}(X_LX_L^*)$ for $L=A,B$.
Suppose a target observable $O_B$ has a known self-adjoint lift
$O_A\in\mathcal F_r^A$ with $O_B=\pi(O_A)$, then, we have
\begin{proposition}\label{prop:ldp}
Whenever $2D+r<d_{\mathrm{rel}}$, $\operatorname{Tr}(O_B\rho_B)
 =
 \frac{\tau_A(X_A^*O_AX_A)}
      {\tau_A(X_A^*X_A)}
 =
 \operatorname{Tr}(O_A\rho_A)$.
\end{proposition}

Thus, if $O_A$ is efficiently accessible and measurable, its
expectation on the source guide estimates the target expectation
without preparing the target or implementing relative decoding.
The source estimation procedure may itself be quantum or classical depending on whether the source is classically efficient. In particular, the energy observable $H_B=\pi(H_A)$ has a
degree-one lift, so its mean is obtainable from the source whenever
$2D+1<d_{\mathrm{rel}}$. This explains why in DQI~\cite{JordanEtAlDQI} and in our later optimization applications, the expected score can be calculated without quantum computers.

On the other hand, for high-degree observables with $2D+r>d_{\mathrm{rel}}$, the target preparation is required. For example, sampling involves the target outcome projectors as observables, which are high-degree. This provides the supporting evidence of the advantage of DQI and our framework: the measured samples cannot be obtained from the source side even though the expected average score can be calculated. Sampling is not the only example where the target preparation is required; any observable that lies outside the low-degree span of the generators is also 
feasible.

\subsection{Generalizations: Fermionic and qudit systems}
\label{intro:sec:generalizations}
The power of the homomorphic transduction framework is that it doesn't rely on specific operator algebras. We show how the framework can go beyond Pauli operators to also Majorana operators for Fermionic systems and to Weyl operators for qudit systems. In both realizations, put
$A=[w_1\ \cdots\ w_s]$, $B=[v_1\ \cdots\ v_s]$,
$K_A=\ker A$, and $K_B=\ker B$.

\paragraph{Fermions.}
For fermionic annihilation operators $f_{L,j}$ on $n_L$ modes,
define $\gamma_{L,2j-1}:=f_{L,j}+f_{L,j}^*$ and
$\gamma_{L,2j}:=i(f_{L,j}^*-f_{L,j})$.
The Hermitian Majorana monomials are
$\Gamma_L(u):=i^{|u|(|u|-1)/2}
\prod_{j=1}^{2n_L}\gamma_{L,j}^{u_j}$,
where $u\in\mathbb F_2^{2n_L}$, $|u|$ is the Hamming weight,
and factors are ordered by increasing index.
For nonzero even-weight labels $w_i,v_i$, the source and the target Hamiltonians can be
\begin{equation}
 \begin{aligned}
 H_A=\sum_{i=1}^s\lambda_i a_i,~
 a_i:=\eta_i^A\Gamma_A(w_i);~
 H_B=\sum_{i=1}^s\lambda_i b_i,~
 b_i:=\eta_i^B\Gamma_B(v_i),
 \end{aligned}
 \label{eq:mf-fermion-hamiltonians}
\end{equation}
with $\lambda_i>0$, and
$\eta_i^L\in\{\pm1\}$.
For $c\in K_A$, write
$a_1^{c_1}\cdots a_s^{c_s}=\zeta_A(c)I_A$,
and define $\zeta_B$ analogously.
The assignment $a_i\mapsto b_i$ extends to a unital
$*$-homomorphism $\pi$ between the generated algebras,
with $\pi(H_A)=H_B$, if and only if, over $\mathbb F_2$,
\begin{equation}
 A^{\mathsf T}A=B^{\mathsf T}B,\qquad
 K_A\subseteq K_B,\qquad
 \zeta_A(c)=\zeta_B(c)\quad(c\in K_A).
 \label{eq:mf-fermion-conditions}
\end{equation}

We provide two examples in Section~\ref{subsec:families-fermion}. One is a Gaussian source transduced to an
interacting noncommuting target, and the other is an interacting density-chain
source mapped to a target with extensive code constraints. In both examples, the ordinary distance is constant while the relative distance can be linear.

\paragraph{Qudits.}
For an odd prime $p$, let $\omega=e^{2\pi i/p}$,
$X_p|r\rangle:=|r+1\bmod p\rangle$, and
$Z_p|r\rangle:=\omega^r|r\rangle$.
Define the Weyl operators by
$W_m(x,z):=\omega^{x\cdot z/2}
\bigotimes_{j=1}^mX_p^{x_j}Z_p^{z_j}$,
where $x,z\in\mathbb F_p^m$.
For nonzero directions $w_i\in\mathbb F_p^{2m_A}$,
$v_i\in\mathbb F_p^{2m_B}$ and phases
$r_i^A,r_i^B\in\mathbb F_p$, put
\begin{equation}
 a_i(\xi):=\omega^{\xi r_i^A}W_{m_A}(\xi w_i),
 \qquad
 b_i(\xi):=\omega^{\xi r_i^B}W_{m_B}(\xi v_i),
 \qquad \xi\in\mathbb F_p.
\end{equation}
Thus $a_i(\xi)a_i(\eta)=a_i(\xi+\eta)$ and
$a_i(\xi)^*=a_i(-\xi)$, and similarly for $b_i$.
For nonempty symmetric symbol sets
$\Lambda_i=-\Lambda_i\subseteq\mathbb F_p^\times$,
the source and the target Hamiltonians can be
\begin{equation}
 \begin{aligned}
 H_A^{(p)}=
   \sum_{i=1}^s\sum_{\xi\in\Lambda_i}
      \lambda_{i,\xi}a_i(\xi),~
 H_B^{(p)}=
   \sum_{i=1}^s\sum_{\xi\in\Lambda_i}
      \lambda_{i,\xi}b_i(\xi),
 \end{aligned}
 \qquad
 \lambda_{i,-\xi}=\lambda_{i,\xi}^*.
 \label{eq:mf-qudit-hamiltonians}
\end{equation}
For $c\in K_A$, write
$a_1(c_1)\cdots a_s(c_s)=\zeta_A(c)I_A$,
and similarly for $B$.
With
$\Omega_m=\left(\begin{smallmatrix}0&-I_m\\I_m&0\end{smallmatrix}\right)$,
the assignment $a_i(\xi)\mapsto b_i(\xi)$ extends to a unital
$*$-homomorphism $\pi_p$, satisfying
$\pi_p(H_A^{(p)})=H_B^{(p)}$, if and only if,
over $\mathbb F_p$,
\begin{equation}
 A^{\mathsf T}\Omega_{m_A}A
 =B^{\mathsf T}\Omega_{m_B}B,\qquad
 K_A\subseteq K_B,\qquad
 \zeta_A(c)=\zeta_B(c)\quad(c\in K_A).
 \label{eq:mf-qudit-conditions}
\end{equation}

The original OPI problem studied in DQI~\cite{JordanEtAlDQI} corresponds to a relation-free case here. In Section~\ref{subsec:families-boson}, we also show how this Weyl formulation can be applied to truncated Bosonic systems.

\subsection{Application: Optimization}
\label{subsec:polynomial-optimization}

Fix $D$
with $2D+1<d_{\mathrm{rel}}$.
We optimize the mean target energy over $|\psi_B^p\rangle$, where $p(x)=\sum_{j=0}^{D}c_jx^j$.
Define the source moments $m_k:=\tau_A(H_A^k)$ and the
$(D+1)\times(D+1)$ matrices
\begin{equation}
 (G_D^{\mathrm{poly}})_{ij}:=m_{i+j},
 \qquad
 (K_D^{\mathrm{poly}})_{ij}:=m_{i+j+1},
 \qquad 0\le i,j\le D.
 \label{eq:optimization-moment-matrices}
\end{equation}
Proposition~\ref{prop:ldp} implies $\langle H_B\rangle
 =
 \frac{\tau_A\!\left(p(H_A)^*H_Ap(H_A)\right)}
      {\tau_A\!\left(p(H_A)^*p(H_A)\right)}
 =
 \frac{c^*K_D^{\mathrm{poly}}c}
      {c^*G_D^{\mathrm{poly}}c}$.
Thus finding the optimal polynomial reduces to the generalized
eigenvalue problem
\begin{equation}
 K_D^{\mathrm{poly}}c_\star
 =\lambda_\star G_D^{\mathrm{poly}}c_\star,
 \qquad
 p_\star(x)=\sum_{j=0}^{D}(c_\star)_j x^j.
 \label{eq:optimization-eigenproblem}
\end{equation}
The largest generalized eigenvalue gives the maximum attainable
mean energy within this polynomial family, and its eigenvector
gives the optimal coefficients; use the smallest eigenpair for
energy minimization.
If $G_D^{\mathrm{poly}}$ is singular, restrict both matrices to
$\operatorname{ran}G_D^{\mathrm{poly}}$, removing coefficient
vectors that represent the zero operator. Finding the optimal polynomial is classically
efficient whenever the involved source moments are efficiently computable.

\subsection{Application: Thermofield-double (TFD) transduction}
\label{intro:sec:purified-gibbs-guides}
Another application is to transfer existing fast-mixing TFD classes to broader regimes. While HDQI~\cite{HamiltonianDQI} also admits Gibbs-state applications, our framework
can combine with existing Gibbs-preparation methods such as the adiabatic preparation along the parent path of the recent Lindbladian approach~\cite{ChenKastoryanoBrandaoGilyen2025ThermalSimulation} that efficiently
produce a canonical source TFD. For $L\in\{A,B\}$, the canonical thermofield-double
state is $|\mathrm{TFD}_L(\beta)\rangle
 :=
 \frac{(e^{-\beta H_L/2}\otimes I)|\Phi_L\rangle}
      {\sqrt{\tau_L(e^{-\beta H_L})}}$. Its reduced density matrix is the Gibbs state $e^{-\beta H_L}/Z_L(\beta)$, with
$Z_L(\beta):=\operatorname{Tr}(e^{-\beta H_L})$.
Assume the source TFD can be prepared efficiently to trace-distance error $\delta_{\mathrm{src}}$. Examples include local Hamiltonians at
sufficiently high temperature~\cite{RouzeFrancaAlhambra2026} and $1D$ systems~\cite{BergamaschiChen1D}. Let $U_D$ be a fixed unitary realization of the efficient polynomial
transduction circuit at degree $D$, with $2D<d_{\mathrm{rel}}$.
Suppose a real polynomial $Q$ of degree at most $D$ satisfies
\begin{equation}
 \sup_{x\in I}\bigl|e^{\beta x/2}Q(x)-1\bigr|
 \le\eta<1,
 \qquad I\supseteq\operatorname{spec}(H_A).
 \label{eq:tfd-polynomial-certificate}
\end{equation}
Because of $\pi(H_A)=H_B$, we have $\operatorname{spec}(H_B)\subseteq\operatorname{spec}(H_A)$, makes this a common approximation bound for both systems.
Now both source and target TFD is within trace distance $\eta$ of its normalized
$Q(H_L)$ operator state.
Consequently, applying $U_D$ directly to the prepared source TFD gives an output reduced density matrix that is within $\delta_{\mathrm{src}}+2\eta$ trace distance from the target Gibbs state.

For a source spectral interval of width $W>0$, fix a target
trace-distance error $0<\varepsilon<1$ with
$0\le\delta_{\mathrm{src}}<\varepsilon$, and set
$\eta:=(\varepsilon-\delta_{\mathrm{src}})/2$ and
$\ell:=\ln[2/(\varepsilon-\delta_{\mathrm{src}})]$.
If $D+1\ge\ell/3$, a sufficient inverse-temperature window is
\begin{equation}
 0\le\beta\le
 \frac{2}{W}
 \left(
   \sqrt{D+1+\frac{\ell}{6}}
   -\sqrt{\frac{\ell}{2}}
 \right)^2.
 \label{eq:tfd-temperature-window}
\end{equation}
For every $\beta$ in this window, there exists a degree-$D$
polynomial certificate satisfying
Eq.~\eqref{eq:tfd-polynomial-certificate}, so the target Gibbs state error is at most
$\varepsilon$. For an extensive spectral width $W=\Theta(m)$, reaching a constant
inverse temperature $\beta>0$ within this guaranteed window
requires $D=\Omega(m)$ and hence a linear relative distance
$d_{\mathrm{rel}}=\Omega(m)$. For non-trivial examples, we refer to the constructions in our companion work~\cite{ShangWangFranca2026Gibbs}, where we show how combining the recent 1D fast-mixing results~\cite{BergamaschiChen1D} as sources and the relative decoding can help the TFD preparation of geometrically nonlocal Hamiltonians, which, to our knowledge, is a new class of efficiently preparable TFD.

\subsection{Advantage example with relative decoding: nonlinear optimal polynomial intersection}
\label{intro:sec:applications}
\label{intro:app:chain-example}
\paragraph{The optimization problem.}
Optimal polynomial intersection (OPI) asks for a low-degree polynomial
whose evaluations belong to prescribed sets at as many points as
possible~\cite{JordanEtAlDQI}. We consider a nonlinear variant in which
each test couples two neighbouring evaluations. Let $p$ be an odd prime,
let $\alpha_1,\ldots,\alpha_L$ list the nonzero elements of $\F_p$, and
put $L=p-1$ and $M=L-1$. Given $2\le n<L$, nonzero coefficients
$a_i\in\F_p$, and explicitly specified sets $F_i\subset\F_p$ of size
$(p-1)/2$, the task is
\begin{equation}
 \max_{\deg P<n}\operatorname{Sat}(P),\qquad
 \operatorname{Sat}(P):=\frac1M\sum_{i=1}^{M}
 \one\!\left\{P(\alpha_{i+1})+a_iP(\alpha_i)^2\in F_i\right\}.
 \label{intro:eq:nonlinear-opi-score}
\end{equation}
Thus we choose the $n$ coefficients of $P$ to satisfy as many of the $M$
nonlinear tests as possible. The difficulty is the global polynomial
constraint: any $n$ evaluations determine all the others. Without this
constraint, one could choose passing residuals $r_i\in F_i$ and set
$x_{i+1}=r_i-a_ix_i^2$ successively, satisfying every test. Write $T(x)$
for the number of tests passed by an arbitrary value vector $x\in\F_p^L$.

\paragraph{Why ordinary decoding is obstructed, but relative decoding works.}
The Fourier expansion of a test contains characters
$\chi_{s,t}(x_i,x_{i+1})=\omega^{sx_i+tx_{i+1}}$, where
$\omega=e^{2\pi i/p}$. All pairs with $t\ne0$ occur with nonzero
coefficients. These terms already satisfy short relations before the
polynomial constraint is imposed; for example,
\[
 \chi_{0,1}\chi_{2,1}=\chi_{1,1}^{\,2}.
\]
Consequently, different products of two Fourier terms can produce the
same label. Ordinary term-selection decoding attempts to distinguish
these products, and therefore encounters an ambiguity. For this same
Fourier-term marking, $d_{\mathrm{ord}}\le4$, so its usual condition
$2D<d_{\mathrm{ord}}$ cannot certify filters of degree $D\ge2$.

Relative decoding instead uses the unconstrained chain as its source,
retaining these identities. Source preparation and Fourier transformation
combine all products representing the same character into a single
amplitude. The decoder recovers only the resulting aggregate label
$u\in\F_p^L$, not the individual terms that produced it. The local
ambiguities are therefore absorbed into the source rather than treated
as decoding errors.

\paragraph{The code, attainable degree, and optimized filter.}
Let $V\in\F_p^{L\times n}$ be the evaluation matrix
$V_{ij}=\alpha_i^{\,j}$, with $0\le j<n$. The feasible value vectors form
the Reed--Solomon code $C=\operatorname{im}V$: a coefficient vector
$z\in\F_p^n$ represents $P_z(t)=\sum_{j=0}^{n-1}z_jt^j$ and has
evaluations $Vz$. Restriction to this code is the homomorphism
$(\pi F)(z)=F(Vz)$, which sends a source Fourier label $u$ to the target
label $V^{\mathsf T}u$.

The decoder uses the dual code $C^\perp=\ker V^{\mathsf T}$, whose
minimum Hamming distance is $n+1$. Each Fourier term touches at most two
vertices, so a degree-$D$ score filter has Fourier support only on labels
with $\operatorname{wt}_H(u)\le2D$. Reed--Solomon syndrome decoding
recovers these labels efficiently whenever $4D\le n$, and can be made
coherent. Moreover, uniform feasible polynomials have independent uniform
evaluations on any $n$ vertices. This also transfers the filtered mean
score from source to target when $4D+2\le n$. Thus
\begin{equation}
 d_{\mathrm{rel}}\ge\left\lceil\frac{n+1}{2}\right\rceil,
 \qquad
 D=\left\lfloor\frac{n-2}{4}\right\rfloor
 \label{intro:eq:nonlinear-opi-degree}
\end{equation}
is an efficiently usable filter degree with a source-computable mean,
despite the constant ordinary distance. Here $D$ is the degree of the
\emph{score filter}, not of the output polynomial $P$.

Source preparation is efficient because
$x\mapsto(x_1,r_1,\ldots,r_M)$ is a bijection. Therefore, the source preparation can be reduced to the standard OPI preparation
followed by a reversible change of coordinates (See Lemma~\ref{app:residual-source-lemma}). 
A coherent pass/fail pattern weighted by a filter $G(T)$, followed by
uniform residual preparation in $F_i$ or its complement and inversion
of the residual map, prepares the source guide. Relative decoding then
transfers it to
\[
 |\psi_G\rangle\propto
 \sum_{z\in\F_p^n}G(T(Vz))|z\rangle.
\]
Measuring returns a feasible polynomial, with probability proportional
to the squared filter value at its score.

We optimize $G$, of degree at most $D$, to maximize the expected measured
score. By source--target moment matching, this is the classical problem
\begin{equation}
 \max_{\substack{G\ne0\\\deg G\le D}}
 \frac{\mathbb E_{T\sim\operatorname{Bin}(M,\rho)}
                  [T|G(T)|^2]}
      {M\,\mathbb E_{T\sim\operatorname{Bin}(M,\rho)}[|G(T)|^2]}.
 \label{intro:eq:nonlinear-opi-optimization}
\end{equation}
In the orthonormal polynomial basis for this binomial law, the
optimization reduces to the largest eigenpair of a
$(D+1)\times(D+1)$ symmetric tridiagonal matrix:
\begin{equation}
 \begin{aligned}
 (J_{\rho,D})_{kk}&=M\rho+(1-2\rho)k,
       &&0\le k\le D,\\
 (J_{\rho,D})_{k,k+1}&=(J_{\rho,D})_{k+1,k}
       =\sqrt{\rho(1-\rho)(k+1)(M-k)},
       &&0\le k<D.
 \end{aligned}
 \label{intro:eq:nonlinear-opi-jacobi}
\end{equation}
All other entries vanish. A unit top eigenvector gives the optimal
filter coefficients in that basis.

\begin{theorem}[Nonlinear OPI, informal]
\label{intro:thm:nonlinear-opi}
For every instance above, relative decoding gives a polynomial-time
quantum algorithm using filter degree $D=\lfloor(n-2)/4\rfloor$ whose
ideal output is a polynomial $P$ of degree less than $n$ satisfying
\begin{equation}
 \mathbb E[\operatorname{Sat}(P)]
 =\mu_{p,n}:=\frac{\lambda_{\max}(J_{\rho,D})}{M}.
 \label{intro:eq:nonlinear-opi-guarantee}
\end{equation}
This is the best expected satisfaction among all score-polynomial
filters of degree at most $D$. In particular, it is an expected
$\mu_{p,n}$-approximation: if
$\mathrm{OPT}:=\max_{\deg P<n}\operatorname{Sat}(P)$, then
$\mathbb E[\operatorname{Sat}(P)]\ge\mu_{p,n}\mathrm{OPT}$, since
$\mathrm{OPT}\le1$.
As $p\to\infty$ with $n/L\to R\in(0,1)$,
\begin{equation}
 \mu_{p,n}\longrightarrow
 \frac12+\sqrt{\frac R4\left(1-\frac R4\right)}.
 \label{intro:eq:nonlinear-opi-limit}
\end{equation}
\end{theorem}
The constructive source preparation, decoding, and performance proofs are given
in Section~\ref{sec:applications}, particularly
Theorem~\ref{app:chain-performance}. 

\paragraph{Numerical evidence of quantum advantage.}
For $p=2053$, $n=205$, and $M=2051$, the allowed filter degree is $D=50$.
The resulting $51\times51$ Jacobi matrix gives the ideal quantum mean
$\mu_{p,n}\simeq0.643134$. We compared this score with various classical algorithms, such as information-set
interpolation, local search, simulated annealing, parallel tempering,
guide-weight sampling methods, and chain inference followed by projection
to a feasible codeword. The test set contains ten independent
residual-coefficient vectors, with two independent choices of the sets
$F_i$ per vector; deep runs evaluate approximately $5\times10^6$
candidate polynomials each.

For each coefficient vector, we average the best classical scores across
its two predicate instances. The median of these ten averages is
$0.60666$, leaving a median gap of approximately $3.65$ percentage points and the gap is positive for all ten averages. Thus the tested classical
benchmarks have not yet matched the ideal quantum value in this
comparison. This suggests potential \textit{exponential} quantum advantage, but does not
establish a classical lower bound.
Full methods, budgets, and reproducibility details appear in
Appendix~\ref{app:chain-benchmark} and Ref.~\cite{RDQINumerics2026}.

\subsection{Discussion and outlook}
\label{sec:discussion}

This work develops a general framework for quantum polynomial transduction
by relative decoding.  The basic idea is to prepare the polynomial state in a simpler source system, transport the polynomial
through a unital $*$-homomorphism, and use a relative decoder for coherent uncomputation. This viewpoint provides several advantages
over the original DQI paradigm.  First, the complicated and case-by-case pilot preparation is reduced to
a unified task of source polynomial-state preparation, to which existing state-preparation methods apply.  Second, the relevant decoding radius is controlled by the
relative distance: relations already present in the source do not
count as decoding obstructions, so the relative distance can be much larger
than the ordinary target distance and can therefore support substantially
higher-degree polynomial filters.  Third, the $*$-homomorphism language naturally enables the construction to accommodate nonuniform Hamiltonian coefficients and extends beyond Pauli
systems to fermionic Majorana algebras, prime-dimensional Weyl systems,
and bosonic Hamiltonians. We show the target expectation value estimation from source for low-degree observables.  We apply the framework to both optimization and Gibbs
state preparation. In the nonlinear OPI example, relative decoding
supports polynomial filters beyond the degree window certified by
ordinary term-selection decoding with the same marking. The exact
quantum mean score also exceeds the scores reached by the tested
classical methods at the reported budgets. 

Several open questions remain.  First, fully exploiting the larger
degree window offered by relative decoding requires efficiently
preparable structured source states whose inherited relations are
rich enough to increase the relative distance without making the source
itself difficult to prepare.  Finding constructions of such
sources or identifying what class of sources can be efficiently preparable is therefore central to the applicable regimes of the framework.  Second, a large relative
distance is only useful algorithmically when the corresponding
relative decoder can be implemented efficiently and coherently. Developing such relative decoder families or applying the degenerate decoding of quantum codes to the framework
may substantially enlarge the applicable scope. It would be especially remarkable if ordinary decoding were hard while relative decoding were easy.  Third, it is
important to identify more problem classes for which the additional
freedom of source design and relative decoding leads to end-to-end advantages. And it would be especially attractive to explore the advantages in solving quantum chemistry problems via relative decoding.

\clearpage
\section{Data and code availability}
The nonlinear residual-chain implementation, reference instances,
polynomial witnesses, and reproduction scripts are publicly available
in \texttt{RDQINumerics}, release v1.0.0~\cite{RDQINumerics2026}, at
\url{https://github.com/dsfranca/RDQINumerics}.
Appendix~\ref{app:chain-benchmark} describes the instances in more detail, reproduction and archive
coverage.

\section{Acknowledgments and AI use statement}
\noindent\textbf{Acknowledgments. }The authors would like to thank Zoë Holmes, Stephen Jordan, Yihui Quek, and Yufei Wang for relevant discussions. Z.S. and D.S.F. acknowledge support by the ERC grant GIFNEQ 101163938. D.S.F. acknowledges financial support from the Novo Nordisk Foundation (Grant No. NNF20OC0059939 Quantum for Life).

\noindent\textbf{AI use statement.} The core ideas of viewing the pilot state of DQI and HDQI as the polynomial state of the source Hamiltonian, generalizing the relation-free sources to general sources, and observing the broader $*$-homomorphism structure are from the human authors. ChatGPT 5.5 Pro, ChatGPT 5.6 Sol, ChatGPT 6 Astra, and Claude Fable 5.1 were used to help formalize the whole framework and prove results, especially identifying that the quotient code and the relative decoding are what is needed in the general source-target transduction. The same models were also used for proofreading. We used ChatGPT 6 Astra to find the nonlinear-chain example in Section~\ref{sec:applications} and ChatGPT 5.6 Sol to code the benchmarks.

\vspace{7em}
\noindent\textbf{Organization.}
Section~\ref{sec:main-framework} develops the detailed transduction
framework, relative decoding, the spin-system realization, and (H)DQI as relation-free special cases. Section~\ref{sec:performance} treats optimization (also including discussions on low vs high-degree observables) and thermofield-double
transduction in detail, while Section~\ref{sec:generalizations} extends the framework
to fermionic, truncated bosonic, and qudit systems with examples that have constant ordinary distances while linear relative distances.
Section~\ref{sec:applications} presents the non-linear OPI examples in more details.
Appendices~\ref{app:cstar-background} and~\ref{sec:constructing-sources}
provide the operator-algebraic background and source constructions.
Appendices~\ref{app:optimization-details} and~\ref{app:gibbs-details}
give the optimization and thermofield-double in further detail.
Appendix~\ref{app:homomorphism-criteria} proves the homomorphism criteria for different systems,
and Appendix~\ref{app:chain-benchmark} provides supplementary numerical details to the non-linear OPI.
\par

\clearpage
\section{Main framework}
\label{sec:main-framework}
Let $H_A=\sum_{i=1}^s \lambda_i a_i$ and $H_B=\sum_{i=1}^s \lambda_i b_i$,
where $a_i$ and $b_i$ are the matched marked interaction terms of the
source and target Hamiltonians, respectively. Our goal is to prepare the normalized operator state of a target
polynomial $q(b_1,\ldots,b_s)$ from an efficiently preparable state of
the same polynomial in the source operators $a_1,\ldots,a_s$.  For
Hamiltonian filters, this specializes to preparing the state of
$p(H_B)$ from that of $p(H_A)$, without applying the polynomial
directly to the target.  The polynomial $p$ is typically chosen or
optimized for a target objective, for example to bias the resulting
state toward low-energy, high-scoring spectral sectors, and to approximate the exponential function for Gibbs sampling.

A unital $*$-homomorphism transports the polynomial algebraically,
while the intrinsic relative distance certifies preservation of the
normalized operator-state geometry within a precise degree window.
Two operator-basis compatibility conditions then provide the
coordinates needed for the coherent circuit implementation.

We introduce these conditions first and give four explicit steps:
source operator--Fourier transform, computation of the target label and
phase, uncomputation by relative decoding, and inverse target
operator--Fourier transform. We then apply the framework for the spin systems and explain how DQI and HDQI serve as special relation-free examples of our transduction framework.

\subsection{Polynomial states, the $*$-homomorphism, and the relative decoding}
\label{subsec:mf-goal-algebra}
We use basic properties of finite-dimensional $C^*$-algebras and unital
$*$-homomorphisms between them~\cite{Murphy1990}, reviewed in
Appendix~\ref{app:cstar-background}.
Let $\mathcal A=C^*(I_A,a_1,\ldots,a_s)\subseteq\mathcal B(\mathcal H_A)$
and $\mathcal B=C^*(I_B,b_1,\ldots,b_s)\subseteq\mathcal B(\mathcal H_B)$
be finite-dimensional represented unital $C^*$-algebras. Given a
noncommutative $*$-polynomial $q$, abbreviate its evaluations by
$X_A=q(a)=q(a_1,\ldots,a_s)$ and
$X_B=q(b)=q(b_1,\ldots,b_s)$.
Write $N_L=\dim\mathcal H_L$ for $L=A,B$. The physical representation
fixes the normalized trace and operator inner product:
\begin{equation}
 \tau_L(X)=\frac{\Tr_{\mathcal H_L}X}{N_L},\qquad
 \langle X,Y\rangle_L=\tau_L(X^*Y),\qquad
 \|X\|_{2,L}=\sqrt{\tau_L(X^*X)}.
 \label{eq:mf-traces}
\end{equation}
In fixed computational bases, we define the vectorized operators~\cite{Watrous2018}
\begin{equation}
 |\Phi_L\rangle=N_L^{-1/2}\sum_j|j,j\rangle,
 \qquad |X\rangle\!\rangle_L=(X\otimes I)|\Phi_L\rangle.
 \label{eq:mf-vectorization}
\end{equation}
Thus $\langle\!\langle X|Y\rangle\!\rangle_L=\tau_L(X^*Y)$.

The input to transduction is the efficiently preparable normalized
operator state
\begin{equation}
 |\psi_A^q\rangle=\frac{|q(a)\rangle\!\rangle_A}{\|q(a)\|_{2,A}},
 \label{eq:mf-guide}
\end{equation}
which we call the \emph{source guiding state}, or \emph{guide}. The
desired output is
\begin{equation}
 |\psi_B^q\rangle=\frac{|q(b)\rangle\!\rangle_B}{\|q(b)\|_{2,B}},
 \label{eq:mf-output}
\end{equation}
whose reduced density matrix is
\begin{equation}
 \rho_B^q=\frac{X_BX_B^*}{\Tr(X_BX_B^*)};
 \qquad \rho_B^p=\frac{p(H_B)^2}{\Tr p(H_B)^2}
 \quad\text{for real }p\text{ and }H_B=H_B^*.
 \label{eq:mf-marginal}
\end{equation}
We discuss how to choose the polynomials when optimizing low-degree expectation values in
Section~\ref{sec:performance}.

\paragraph{The algebraic map and the marked terms.}
The transduction requires a unital $*$-homomorphism
$\pi:\mathcal A\twoheadrightarrow\mathcal B$ satisfying
$\pi(a_i)=b_i$. For Hamiltonian applications, a fixed term
decomposition $H_B=h_0I_B+\sum_i h_i b_i$ provides the marked target
operators. We keep the identity shift $h_0$ and coefficients $h_i$
separate from the generators and use the same coefficients in the source,
$H_A=h_0I_A+\sum_i h_i a_i$.
For non-Hermitian terms the decomposition includes their adjoints with
conjugate coefficients. In every case $\pi(H_A)=H_B$.

A \emph{marked word} of length $\ell$ is an ordered product of $\ell$
marked generators or their adjoints. The empty word is the identity
and has length zero; each generator or adjoint counts as one letter.
Define the degree-at-most-$t$ spaces as
\begin{equation}
 \mathcal F_t^A=\operatorname{span}\{I_A,
 a_{i_1}^{\epsilon_1}\cdots a_{i_\ell}^{\epsilon_\ell}:
 \ell\le t,\ \epsilon_j\in\{1,*\}\},\qquad
 \mathcal F_t^B=\pi(\mathcal F_t^A).
 \label{eq:mf-filtration}
\end{equation}
It follows from the $*$-homomorphism property that:
\begin{equation}
 X_A=q(a_1,\ldots,a_s)\in\mathcal F_D^A,
 \qquad X_B=q(b_1,\ldots,b_s)=\pi(X_A).
 \label{eq:mf-polynomial-identity}
\end{equation}
In particular, a degree-$D$ polynomial $p(H_A)$ has marked degree at most
$D$ for the Hamiltonian-term marking, and $\pi(p(H_A))=p(H_B)$.

The homomorphism identity alone, however, does not
preserve the physical operator inner products. Relative distance
specifies the degree through which it does.

\begin{definition}[Intrinsic relative distance]
\label{def:mf-intrinsic-distance}
The relative distance of the marked, represented homomorphism is
\begin{equation}
 d_{\mathrm{rel}}=
 \min\{t\in\mathbb N_0:\exists Y\in\mathcal F_t^A,
          \ \tau_B(\pi(Y))\ne\tau_A(Y)\},
 \label{eq:mf-abstract-relative-distance}
\end{equation}
with the empty minimum equal to $+\infty$.
\end{definition}
\noindent Equivalently, for every integer $t\ge0$,
\begin{equation}
 \tau_B\circ\pi=\tau_A\text{ on }\mathcal F_t^A
 \quad\Longleftrightarrow\quad t<d_{\mathrm{rel}}.
 \label{eq:mf-trace-distance-characterization}
\end{equation}

\begin{proposition}[Sharp degree--distance criterion]
\label{thm:mf-degree-distance}
For an integer $D\ge0$, the condition
\begin{equation}
 2D<d_{\mathrm{rel}}
 \label{eq:mf-half-distance}
\end{equation}
is equivalent to
\begin{equation}
 \tau_B(\pi(X)^*\pi(Y))=\tau_A(X^*Y)
 \qquad(X,Y\in\mathcal F_D^A).
 \label{eq:mf-degree-balance}
\end{equation}
In particular $\pi$ is an isometry from $\mathcal F_D^A$ onto
$\mathcal F_D^B$ in their physical operator inner products.
\end{proposition}
\begin{proof}
For the forward direction, $X^*Y\in\mathcal F_{2D}^A$, and
$\pi(X)^*\pi(Y)=\pi(X^*Y)$. Apply Eq.~\eqref{eq:mf-trace-distance-characterization}.
Conversely,
$\mathcal F_{2D}^A=\operatorname{span}\{X^*Y:X,Y\in\mathcal F_D^A\}$:
split each word of length at most $2D$ into two subwords, and absorb
adjoints in the first. Eq.~\eqref{eq:mf-degree-balance} therefore
implies trace matching on $\mathcal F_{2D}^A$, which is equivalent to
Eq.~\eqref{eq:mf-half-distance}.
\end{proof}

In particular, $\pi$ is injective on $\mathcal F_D^A$, which allows
the source label to be uncomputed in
Section~\ref{subsec:mf-four-step-circuit}.

\paragraph{Ordinary decoding vs relative decoding.}
Ordinary decoding in DQI and HDQI requires uniqueness with respect to
all low-weight target relations. Relative decoding instead absorbs
the relations already present in the source, so that uniqueness
depends on the first additional target relation. When the shortest
target relations are inherited, the larger relative distance can
enlarge the transduction window $2D<d_{\mathrm{rel}}$.

To define the ordinary distance, we specify the source before the
additional relations used by relative decoding are imposed. Let
$\mathcal A_0=C^*(I_0,r_1,\ldots,r_s)$ be a faithful, finite-dimensional
realization of a fixed generator presentation, with physical normalized
trace $\tau_0$ and marked word spaces $\mathcal F_t^0$.
Here \emph{relation free} means that $\mathcal A_0$ satisfies no
relations beyond those of the chosen presentation. In algebraic terms,
if $\mathcal P_s$ is the free complex $*$-polynomial algebra and
$\mathcal J_0$ is the ideal of stipulated defining relations (such as commutation relations), then
$\mathcal A_0\cong\mathcal P_s/\mathcal J_0$: evaluating the formal
polynomials in the $r_i$ has kernel exactly $\mathcal J_0$.
This retains generator orders and prescribed commutation relations.

The chosen source and target are successive quotients,
\begin{equation}
 \vartheta_A:\mathcal A_0\twoheadrightarrow\mathcal A,
 \qquad \pi:\mathcal A\twoheadrightarrow\mathcal B,
 \qquad \vartheta_B:=\pi\circ\vartheta_A,
 \qquad \vartheta_A(r_i)=a_i,\quad\vartheta_B(r_i)=b_i.
 \label{eq:mf-common-reference}
\end{equation}
The ideals $\ker\vartheta_A\subseteq\ker\vartheta_B$ record the
additional polynomial relations. Ordinary transduction uses
$\mathcal A_0\to\mathcal B$ directly; relative transduction uses
$\mathcal A\to\mathcal B$ after the source relations have been absorbed.
Accordingly, define the ordinary distance
\begin{equation}
 \begin{aligned}
 d_{\mathrm{ord}}(B\mid\mathcal A_0)
 &:=d_{\mathrm{rel}}(\mathcal A_0\xrightarrow{\vartheta_B}\mathcal B)\\
 &=\min\{t\in\mathbb N_0:\exists Z\in\mathcal F_t^0,
      \tau_B(\vartheta_B(Z))\ne\tau_0(Z)\}.
 \end{aligned}
 \label{eq:mf-ordinary-reference-distance}
\end{equation}
For every marked word $W_0$ in the $r_i$ and their adjoints,
write $W_A=\vartheta_A(W_0)$ and $W_B=\vartheta_B(W_0)$, and assume
\begin{equation}
 W_L\notin\mathbb C I_L
 \quad\Longrightarrow\quad
 \tau_L(W_L)=0,
 \qquad L=0,A,B,
 \label{eq:mf-scalar-or-traceless-words}
\end{equation}
which is true for all the cases considered in this work.

For the following statement, a \emph{target scalar-word relation
relative to $\mathcal A_0$} means an identity
$W_B=\lambda I_B$ with $\lambda\ne0$ that is not already satisfied
in the $\mathcal A_0$, namely $W_0\ne\lambda I_0$.
Such a relation is \emph{inherited from $\mathcal A$} if
$W_A=\lambda I_A$.

\begin{theorem}[Ordinary source versus a quotient source]
\label{the:mf-general-distance-comparison}
Under Eq.~\eqref{eq:mf-scalar-or-traceless-words},
\begin{equation}
 d_{\mathrm{rel}}(A\to B)
 \ge d_{\mathrm{ord}}(B\mid\mathcal A_0).
\end{equation}
If $d_{\mathrm{ord}}<\infty$, the inequality is strict if every minimum-length
target scalar-word relation relative to $\mathcal A_0$ is inherited
from $A$. When $\vartheta_A=\mathrm{id}$,
the ordinary and relative distances coincide.
\end{theorem}
\begin{proof}
Homomorphisms preserve scalars. Thus, under
Eq.~\eqref{eq:mf-scalar-or-traceless-words}, for $L=0,A$,
\[
 \tau_B(W_B)\ne\tau_L(W_L)
 \quad\Longleftrightarrow\quad
 W_B=\lambda I_B,\ \lambda\ne0,
 \quad W_L\notin\mathbb C I_L.
\]
Since marked words span the filtered spaces, the ordinary distance
is the minimum length of a target scalar-word relation relative to
$\mathcal A_0$, while the relative distance is the minimum length
among those not inherited from $A$. The latter minimum is taken over
a subset, proving the inequality. For finite ordinary distance,
the inequality is strict exactly when every minimum-length target
relation is inherited. Empty minima give $+\infty$. When
$\mathcal A=\mathcal A_0$ and $\vartheta_A=\mathrm{id}$, the two
sets coincide, so the distances agree.
\end{proof}

In Appendix~\ref{subsec:source-relations}, we discuss how to design
the source so as to increase the relative distance.

\subsection{Operator bases and the two compatibility conditions}
\label{subsec:mf-cutoff-coordinates}
Fix $D$ with $2D<d_{\mathrm{rel}}$. Choose operator bases
$\{E_\alpha^L:\alpha\in\Omega_L\}$ orthonormal in the physical inner
products of Eq.~\eqref{eq:mf-traces}, with $E_0^L=I_L$.
The degree criterion is already basis independent; the following
conditions connect it to the coordinates used by the circuit.
It suffices to specify the bases on the degree-$D$ subspaces and complete
them on the remaining operator spaces.

We need two compatibility conditions:
\begin{description}[leftmargin=0pt,labelsep=0.5em]
\item[(C1) Generator--basis compatibility at $D$.]
There is a label set $S_D$ such that
\begin{equation}
 \mathcal F_D^A=\operatorname{span}\{E_\alpha^A:\alpha\in S_D\},
 \qquad E_0^A=I_A.
 \label{eq:mf-generator-basis}
\end{equation}
Thus these basis operators span the degree-$D$ space.
\item[(C2) Single-image basis compatibility.]
For orthonormal target coordinates containing $E_0^B=I_B$,
\begin{equation}
 \pi(E_\alpha^A)=c_\alpha E_{M(\alpha)}^B,\qquad \text{for every }\alpha\in\Omega_A.
 \label{eq:mf-basis-transport}
\end{equation}
\end{description}

\begin{proposition}[Compatible coordinates of the balanced space]
\label{prop:mf-compatible-coordinates}
Under (C1), (C2), and $2D<d_{\mathrm{rel}}$, one has
$|c_\alpha|=1$ for every $\alpha\in S_D$, and
$M|_{S_D}$ is injective.
\end{proposition}
\begin{proof}
Proposition~\ref{thm:mf-degree-distance} gives
\[
 \delta_{\alpha\beta}
 =\overline{c_\alpha}c_\beta
      \delta_{M(\alpha),M(\beta)}.
\]
Diagonal entries prove unit modulus; off-diagonal entries prove
injectivity.
\end{proof}
The target generator--basis compatibility is a consequence, rather than
an independent assumption. With the inherited target marking $b_i=\pi(a_i)$,
set $S_D^B=M(S_D)$. Then
\begin{equation}
 \mathcal F_D^B=\pi(\mathcal F_D^A)
 =\operatorname{span}\{E_\beta^B:\beta\in S_D^B\},\qquad
 M:S_D\longrightarrow S_D^B\text{ is bijective}.
 \label{eq:mf-target-cutoff-compatibility}
\end{equation}
Thus both degree-$D$ spaces have the required orthonormal coordinates.

\begin{proposition}[Relative distance in basis coordinates]
\label{prop:relative coordinates}
Define $w_A(\alpha)=\min\{t:E_\alpha^A\in\mathcal F_t^A\}$. If an operator basis is additionally adapted to the word filtration
at every degree
\begin{equation}
 \mathcal F_t^A
 =
 \operatorname{span}\{E_\alpha^A : w_A(\alpha)\le t\},
 \qquad \text{for every}~t\geq0,
 \label{eq:filtration-adapted-basis}
\end{equation}
the relative distance Definition~\ref{def:mf-intrinsic-distance} has the equivalent expression
\begin{equation}
 d_{\mathrm{rel}}=\min\{w_A(\alpha):\alpha\ne0,
 c_\alpha\ne0,\ M(\alpha)=0\}.
\end{equation}
\end{proposition}
\begin{proof}
Orthonormality and $E_0^L=I_L$ give
\[
 \tau_A(E_\alpha^A)=\delta_{\alpha,0},
 \qquad
 \tau_B(\pi(E_\alpha^A))
 =c_\alpha\delta_{M(\alpha),0}.
\]
Unitality implies $c_0=1$ and $M(0)=0$, so the two traces agree
on the identity. Since the source basis is adapted to the word
filtration, every $Y\in\mathcal F_t^A$ has an expansion
\[
 Y=\sum_{w_A(\alpha)\le t}y_\alpha E_\alpha^A.
\]
Consequently,
\[
 \tau_B(\pi(Y))-\tau_A(Y)
 =
 \sum_{\substack{\alpha\ne0,\ w_A(\alpha)\le t\\
                  M(\alpha)=0}}
 y_\alpha c_\alpha.
\]
Thus the two traces agree on $\mathcal F_t^A$ if no nonidentity
label $\alpha$ of weight at most $t$ satisfies
$c_\alpha\ne0$ and $M(\alpha)=0$.
Conversely, if such a label exists, choosing $Y=E_\alpha^A$
gives
\[
 \tau_A(Y)=0,
 \qquad
 \tau_B(\pi(Y))=c_\alpha\ne0,
\]
so trace matching fails on $\mathcal F_t^A$.
The first degree at which the traces disagree is therefore
the minimum in the displayed characterization.
\end{proof}

\subsection{The four-step transduction circuit and the polynomial transduction
\label{subsec:mf-four-step-circuit}}
In the degree-$D$ setting with $2D<d_{\mathrm{rel}}$, take $S=S_D$ and
$\chi_\alpha=c_\alpha$. Proposition~\ref{prop:mf-compatible-coordinates}
gives $|\chi_\alpha|=1$ and makes $M$ injective on $S$.
Let the nonzero source guide in Eq.~\eqref{eq:mf-guide} be efficiently
preparable. Its normalized operator coefficients are
\begin{equation}
 x_\alpha:=\frac{\langle E_\alpha^A,q(a)\rangle_A}{\|q(a)\|_{2,A}},
 \qquad
 \frac{q(a)}{\|q(a)\|_{2,A}}
   =\sum_{\alpha\in S}x_\alpha E_\alpha^A,
 \qquad\sum_{\alpha\in S}|x_\alpha|^2=1.
 \label{eq:mf-normalized-coefficients}
\end{equation}
We treat preparation of this state as an input to the algorithm.

The computational label registers are denoted by $\mathsf S$ (source),
$\mathsf T$ (target), and $\mathsf R$ (temporary recovery). Labels are
encoded as bit strings. We suppress the fixed ancilla registers used
inside the operator--Fourier transforms for register reordering,
valid-label encodings, and multiplicity or spectator spaces. The four steps act coherently, without
measuring an intermediate label.

\paragraph{Step 1. Operator--Fourier transform of the efficient source.}
For $L=A,B$, the operator--Fourier transform is defined by
\begin{equation}
 \mathsf F_L|E_\alpha^L\rangle\!\rangle_L=|\alpha\rangle,
 \qquad
 \mathsf F_L|X\rangle\!\rangle_L
 =\sum_\alpha\langle E_\alpha^L,X\rangle_L|\alpha\rangle.
 \label{eq:mf-oft}
\end{equation}
This changes coordinates within one operator space. We require an
efficient implementation of $\mathsf F_A$ on the source subspace.
Starting from the polynomial source
guide, it produces
\begin{equation}
 |\psi_A^q\rangle|0\rangle_{\mathsf T}|0\rangle_{\mathsf R}
 \xrightarrow{\ \mathsf F_A\ }
 \sum_{\alpha\in S}x_\alpha
     |\alpha\rangle_{\mathsf S}|0\rangle_{\mathsf T}|0\rangle_{\mathsf R}.
 \label{eq:mf-source-analysis-step}
\end{equation}
A direct preparation of the same coefficient state can replace the
source preparation and analysis together.

\paragraph{Step 2. Compute the target label and apply its phase.}
An efficient reversible evaluation of $M$ implements
\begin{equation}
 \mathsf U_M:
 |\alpha\rangle_{\mathsf S}|0\rangle_{\mathsf T}
 \longmapsto
 |\alpha\rangle_{\mathsf S}|M(\alpha)\rangle_{\mathsf T}.
 \label{eq:mf-target-label-computation}
\end{equation}
Compute $\chi_\alpha$, apply it as the diagonal phase
$\mathsf U_\chi|\alpha\rangle=\chi_\alpha|\alpha\rangle$, and
uncompute the phase-evaluation workspace. The state is now
\begin{equation}
 \sum_{\alpha\in S}\chi_\alpha x_\alpha
 |\alpha\rangle_{\mathsf S}|M(\alpha)\rangle_{\mathsf T}
 |0\rangle_{\mathsf R}.
\end{equation}
At this point the source label is still present. Discarding it would
generally decohere the target superposition, so Step 3 uncomputes it.

\paragraph{Step 3. Uncompute the source label by relative decoding.}
We require an efficient clean reversible decoder on $M(S)$~\cite{Bennett1973}:
\begin{equation}
 \mathsf{Dec}_{A\leftarrow B}:
 |M(\alpha)\rangle_{\mathsf T}|0\rangle_{\mathsf R}
 \longmapsto
 |M(\alpha)\rangle_{\mathsf T}|\alpha\rangle_{\mathsf R}
 \qquad(\alpha\in S).
 \label{eq:mf-basis-decoder}
\end{equation}
Injectivity proves that this answer is unique. On each basis component, decoding, bitwise XOR from the
recovery register to the source register, and inverse decoding give
\begin{equation}
 \begin{aligned}
 |\alpha\rangle_{\mathsf S}|M(\alpha)\rangle_{\mathsf T}|0\rangle_{\mathsf R}
 &\xrightarrow{\ \mathsf{Dec}_{A\leftarrow B}\ }
 |\alpha\rangle_{\mathsf S}|M(\alpha)\rangle_{\mathsf T}|\alpha\rangle_{\mathsf R}\\
 &\xrightarrow{\ \mathrm{XOR}_{\mathsf R\to\mathsf S}\ }
 |0\rangle_{\mathsf S}|M(\alpha)\rangle_{\mathsf T}|\alpha\rangle_{\mathsf R}\\
 &\xrightarrow{\ \mathsf{Dec}_{A\leftarrow B}^{-1}\ }
 |0\rangle_{\mathsf S}|M(\alpha)\rangle_{\mathsf T}|0\rangle_{\mathsf R}.
 \end{aligned}
 \label{eq:mf-decoder-uncomputation}
\end{equation}
Together, Steps 2 and 3 implement the phase-and-label transfer
\begin{equation}
 \sum_{\alpha\in S}x_\alpha|\alpha\rangle|0\rangle
 \longmapsto
 |0\rangle\sum_{\alpha\in S}\chi_\alpha x_\alpha|M(\alpha)\rangle,
 \label{eq:mf-label-erasure}
\end{equation}
with the recovery register and all arithmetic workspace reset to zero.
Source Fourier analysis has already combined contributions from
different generator words into each coefficient $x_\alpha$, so the
decoder recovers only the corresponding operator-basis label.

\paragraph{Step 4. Inverse operator--Fourier transform of the target.}
The homomorphism and the balanced basis images give
\begin{equation}
 \frac{q(b)}{\|q(b)\|_{2,B}}
   =\sum_{\alpha\in S}\chi_\alpha x_\alpha E_{M(\alpha)}^B.
\end{equation}
An efficient inverse target operator--Fourier transform therefore acts as
\begin{equation}
 \sum_{\alpha\in S}\chi_\alpha x_\alpha|M(\alpha)\rangle_{\mathsf T}
 \xrightarrow{\ \mathsf F_B^{-1}\ }
 \frac{|q(b)\rangle\!\rangle_B}{\|q(b)\|_{2,B}}
 =|\psi_B^q\rangle.
 \label{eq:mf-target-synthesis-step}
\end{equation}
The source and recovery registers stay in their zero states.

We summarize the construction in the following theorem.
\begin{theorem}[Efficient polynomial transduction by relative decoding]
\label{thm:mf-polynomial-transduction}
Fix an integer $D\geq 0$. Assume the generator--basis compatibility
(C1) and the source--target basis compatibility (C2) hold at degree $D$,
and suppose
\begin{equation}
 2D<d_{\mathrm{rel}}.
 \label{eq:mf-transduction-distance-condition}
\end{equation}
Then, for
\begin{equation}
 S_D:=\{\alpha:w_A(\alpha)\leq D\},
\end{equation}
the source--target basis map satisfies
\begin{equation}
 \pi(E_\alpha^A)
 =
 c_\alpha E_{M(\alpha)}^B,
 \qquad
 |c_\alpha|=1,
 \qquad
 \alpha\in S_D,
 \label{eq:mf-degree-D-basis-map}
\end{equation}
and $M|_{S_D}$ is injective.

Suppose furthermore that

\begin{enumerate}
 \item the normalized source guide
 \begin{equation}
  |\psi_A^q\rangle
  =
  \frac{|q(a_1,\ldots,a_s)\rangle\!\rangle_A}
       {\|q(a_1,\ldots,a_s)\|_{2,A}}
 \end{equation}
 is efficiently preparable for the given polynomial $q$ of marked
 degree at most $D$;

 \item the source operator--Fourier transform $\mathsf F_A$ and the
 inverse target operator--Fourier transform $\mathsf F_B^{-1}$ are
 efficiently implementable on the corresponding degree-$D$ subspaces;

 \item the map $M(\alpha)$ and the phase $c_\alpha$ are efficiently
 computable for $\alpha\in S_D$;

 \item there is an efficient coherent relative decoder
 \begin{equation}
  \operatorname{Dec}_{A\leftarrow B}:
  |M(\alpha)\rangle|0\rangle
  \longmapsto
  |M(\alpha)\rangle|\alpha\rangle,
  \qquad \alpha\in S_D.
 \end{equation}
\end{enumerate}

Then the four-step circuit
\begin{equation}
 \begin{aligned}
 |\psi_A^q\rangle|0\rangle
 &\xrightarrow{\ 1.\ \mathsf F_A\ }
   \sum_{\alpha\in S_D}x_\alpha|\alpha\rangle|0\rangle
 \\[-1mm]
 &\xrightarrow{\ 2.\ \mathsf U_c\mathsf U_M\ }
   \sum_{\alpha\in S_D}
   c_\alpha x_\alpha
   |\alpha\rangle|M(\alpha)\rangle
 \\[-1mm]
 &\xrightarrow{\ 3.\ \mathrm{relative\ decoding}\ }
   |0\rangle
   \sum_{\alpha\in S_D}
   c_\alpha x_\alpha|M(\alpha)\rangle
 \\[-1mm]
 &\xrightarrow{\ 4.\ \mathsf F_B^{-1}\ }
   |0\rangle|\psi_B^q\rangle
 \end{aligned}
 \label{eq:mf-four-step-transduction}
\end{equation}
prepares
\begin{equation}
 |\psi_B^q\rangle
 =
 \frac{|q(b_1,\ldots,b_s)\rangle\!\rangle_B}
      {\|q(b_1,\ldots,b_s)\|_{2,B}}
\end{equation}
deterministically and efficiently.
\end{theorem}

\subsection{Polynomial transduction in spin systems}
\label{subsec:mf-relative-codes}
We now specialize to Pauli Hamiltonians
(Fig.~\ref{intro:fig:relative-decoding-framework}). This case makes the
framework concrete and shows how DQI and HDQI arise as relation-free
instances.

\subsubsection{Formulation: Pauli generators and quotient codes}
\label{subsubsec:spin-formulation}

\paragraph{Pauli Hamiltonians and their relation codes.}
We first instantiate the framework for Pauli Hamiltonians.  Consider a
source Hamiltonian and a target Hamiltonian
\begin{equation}
 H_A=\sum_{i=1}^s \lambda_i a_i,
 \qquad
 H_B=\sum_{i=1}^s \lambda_i b_i,
 \qquad
 \lambda_i>0,
 \label{eq:mf-pauli-hamiltonians}
\end{equation}
whose indexed signed Pauli terms are
\begin{equation}
 a_i=\eta_i^A P_{n_A}(w_i),
 \qquad
 b_i=\eta_i^B P_{n_B}(v_i),
 \qquad
 \eta_i^L\in\{\pm1\}.
 \label{eq:mf-pauli-terms}
\end{equation}
Here
\begin{equation}
 P_n(x,z):=i^{x\cdot z}X^xZ^z
\end{equation}
is the Hermitian Pauli representative, where the exponent
$x\cdot z$ is evaluated as an integer.  The individual Hamiltonian
terms $a_i$ and $b_i$ are the marked generators of the source and
target algebras,
\begin{equation}
 \mathcal A=C^*(I,a_1,\ldots,a_s),
 \qquad
 \mathcal B=C^*(I,b_1,\ldots,b_s).
\end{equation}
Our goal is Hamiltonian polynomial transduction: starting from the
operator state associated with $p(H_A)$, prepare the corresponding
operator state associated with $p(H_B)$.  Algebraically, this is
possible whenever the termwise assignment
\begin{equation}
 a_i\longmapsto b_i
\end{equation}
extends to a unital $*$-homomorphism
$\pi:\mathcal A\to\mathcal B$, because then
\begin{equation}
 H_B=\pi(H_A),
 \qquad
 p(H_B)=\pi\!\left(p(H_A)\right)
 \label{eq:mf-pauli-polynomial-map}
\end{equation}
for every polynomial $p$.

To characterize when such a homomorphism exists, collect the binary
Pauli labels of the Hamiltonian terms into~\cite{DehaeneDeMoor2003}
\begin{equation}
 A=[w_1\ \cdots\ w_s],
 \qquad
 B=[v_1\ \cdots\ v_s],
 \qquad
 K_A:=\ker A,
 \qquad
 K_B:=\ker B.
 \label{eq:mf-pauli-codes}
\end{equation}
All kernels and matrix operations in the label space are over
$\mathbb F_2$.  A vector $c\in K_L$ specifies a product of indexed
Hamiltonian terms whose projective Pauli label vanishes.  Accordingly,
for $c\in\mathbb F_2^s$ define the ordered products
\begin{equation}
 W_A(c):=a_1^{c_1}\cdots a_s^{c_s},
 \qquad
 W_B(c):=b_1^{c_1}\cdots b_s^{c_s}.
\end{equation}
If $c\in K_L$, then this product is scalar, and we write
\begin{equation}
 W_L(c)=\zeta_L(c)I.
\end{equation}
Thus $K_L$ records the multiplicative relations among the indexed
Hamiltonian terms, while $\zeta_L(c)$ records the corresponding scalar
phase.  We call the source Hamiltonian relation free when
\begin{equation}
 K_A=\{0\},
\end{equation}
that is, when no nontrivial reduced ordered product of its indexed Pauli terms is scalar.
\begin{theorem}\label{theorem:pauli}
The termwise assignment $a_i\mapsto b_i$ extends to a unital
$*$-homomorphism $\pi:\mathcal A\to\mathcal B$ exactly when
\begin{equation}
 A^{\mathsf T}J_{n_A}A
 =
 B^{\mathsf T}J_{n_B}B,
 \qquad
 K_A\subseteq K_B,
 \qquad
 \zeta_A(c)=\zeta_B(c)
 \quad(c\in K_A),
 \label{eq:mf-pauli-compatibility}
\end{equation}
where
\begin{equation}
 J_n=
 \begin{pmatrix}
 0&I_n\\
 I_n&0
 \end{pmatrix}.
\end{equation}
\end{theorem}
The first condition says that corresponding Hamiltonian terms have the
same pairwise commutation and anticommutation relations.  The second
says that every scalar relation already present among the source terms
is also a relation among the target terms.  The third requires the
scalar phase of every such source relation to be preserved. Proofs can be found in Appendix~\ref{app:homomorphism-criteria}. Also, in Appendix~\ref{subsec:source-relations}, we prove that for every $K_A\subseteq K_B$, one can construct a source Hamiltonian $H_A$ with the label matrix kernel equal to $K_A$.

The quotient of label spaces is
\begin{equation}
 M:\mathbb F_2^s/K_A\longrightarrow\mathbb F_2^s/K_B,
 \qquad [c]_{K_A}\longmapsto[c]_{K_B},
 \label{eq:mf-code-quotient}
\end{equation}
or, in physical Pauli labels, $M(Ac)=Bc$. The source quotient weight and the coordinate expression for the
intrinsic relative distance are
\begin{equation}
 w_A(Ac):=\min_{k\in K_A}|c+k|,\qquad
 d_{A\to B}:=\min\{|c|:c\in K_B\setminus K_A\},
 \label{eq:mf-relative-distance}
\end{equation}
with $d_{A\to B}=+\infty$ if the set is empty. Here $|c|$ is Hamming
weight. Let $S_D(A):=\{u\in\operatorname{im}A:w_A(u)\le D\}$.
Products of at most $D$ marked generators have labels in $S_D(A)$, so
\begin{equation}
 \mathcal F_D^A
 =\operatorname{span}\{P_{n_A}(u):u\in S_D(A)\}.
 \label{eq:mf-pauli-filtration}
\end{equation}
The Pauli bases are orthonormal for the traces in Eq.~\eqref{eq:mf-traces},
and the homomorphism acts by
\begin{equation}
 \pi(P_{n_A}(u))=\chi(u)P_{n_B}(M(u)),\qquad \chi(u)\in\{\pm1\}.
 \label{eq:mf-pauli-monomial}
\end{equation}
The sign is obtained by choosing any $c$ with $Ac=u$ and comparing
the ordered source and target word phases. Eq.~\eqref{eq:mf-pauli-compatibility}
makes it independent of that choice. The operator--Fourier transform
is the phase-corrected Bell transform given below.

At the selected cutoff $D$, Eq.~\eqref{eq:mf-pauli-filtration}
verifies (C1), and
Eq.~\eqref{eq:mf-pauli-monomial} verifies (C2). Moreover,
\begin{equation}
 d_{\mathrm{rel}}
 =\min_{\substack{u\ne0\\M(u)=0}}w_A(u)
 =\min_{c\in K_B\setminus K_A}|c|
 =d_{A\to B}.
\end{equation}
Indeed, each word has zero normalized trace unless its reduced Pauli
label vanishes. The first word whose traces differ is a minimum-weight
element of $K_B\setminus K_A$; inherited scalar relations have the
same phase on both sides by
Eq.~\eqref{eq:mf-pauli-compatibility}. This proves the displayed
formula directly from the intrinsic definition. Proposition~\ref{thm:mf-degree-distance} therefore
gives injectivity on $S_D(A)$ and preservation of inner products
whenever $2D<d_{A\to B}$.

The uncomputation step therefore requires an efficient coherent decoder
satisfying
\begin{equation}
 \mathsf{Dec}_{A\leftarrow B}:|Bc\rangle|0\rangle
 \longmapsto |Bc\rangle|Ac\rangle
 \qquad(|c|\le D).
 \label{eq:mf-relative-decoder}
\end{equation}
The distance bound ensures that the decoder's output is unique on
the degree-$D$ support, as in Section~\ref{subsec:mf-four-step-circuit}.

\paragraph{Relative versus relation-free ordinary distance.}
The ordinary distance of the target relation code for this same marked
term list is
\begin{equation}
 d_B:=\min\{|c|:0\ne c\in K_B\}.
 \label{eq:mf-ordinary-target-distance}
\end{equation}
It governs recovery
of the complete selection vector $c$ from $Bc$, rather than recovery
only of its source class $Ac$.

\begin{proposition}[Relative distance vs ordinary distance]
\label{prop:mf-relative-vs-ordinary}
For a compatible Pauli source and target with the same indexed marking,
\begin{equation}
 d_{\mathrm{rel}}=d_{A\to B}\ge d_B,
 \qquad K_A=\{0\}\ \Longrightarrow\ d_{A\to B}=d_B.
 \label{eq:mf-ordinary-relative-comparison}
\end{equation}
If $d_B<\infty$, strict improvement holds exactly when every
minimum-weight target relation is already a source relation:
\begin{equation}
 d_{A\to B}>d_B
 \quad\Longleftrightarrow\quad
 \{c\in K_B:|c|=d_B\}\subseteq K_A.
 \label{eq:mf-strict-distance-improvement}
\end{equation}
More generally, for compatible source choices with kernels
$K_A\subseteq K_{A'}\subseteq K_B$ and the same target marking,
\begin{equation}
 d_{A\to B}\le d_{A'\to B}.
 \label{eq:mf-source-kernel-monotonicity}
\end{equation}
These statements use the extended ordering including $+\infty$;
if $K_A=K_B$, then $d_{A\to B}=+\infty$.
\end{proposition}
\begin{proof}
Directly apply Theorem~\ref{the:mf-general-distance-comparison} here and replace with the Pauli notation.
\end{proof}

The relation-free lifts used for DQI and HDQI below have $K_A=\{0\}$,
so their relative distance is $d_B$. The ordinary
full-ball uniqueness condition is $2D<d_B$, while the relative condition
is $2D<d_{A\to B}$. For finite distances their guaranteed integer
cutoffs satisfy
\begin{equation}
 \left\lfloor\frac{d_B-1}{2}\right\rfloor
 \le
 \left\lfloor\frac{d_{A\to B}-1}{2}\right\rfloor.
 \label{eq:mf-compared-degree-windows}
\end{equation}
A strict increase of the integer cutoffs therefore allows a higher polynomial degree obeying $d_B\le2D<d_{A\to B}$. Preparing the new source and decoding
efficiently remain separate requirements.

\paragraph{The operator--Fourier circuit.}
For $u=(x,z)$,
\begin{equation}
 |P_n(x,z)\rangle\!\rangle
 =2^{-n/2}i^{x\cdot z}\sum_r(-1)^{z\cdot r}|r+x,r\rangle.
\end{equation}
The operator--Fourier transform
$|P_n(x,z)\rangle\!\rangle\mapsto|x,z\rangle$ is the $n$-fold
tensor product, over physical--reference qubit pairs, of the
two-qubit circuit below.
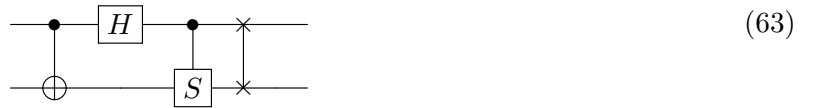
\begin{equation}
\Qcircuit @C=1.1em @R=0.9em {
\lstick{} & \ctrl{1} & \gate{H} & \ctrl{1} & \qswap            & \qw & \rstick{}\qw \\
\lstick{} & \targ    & \qw      & \gate{S} & \qswap\qwx[-1]   & \qw & \rstick{}\qw \\
}
\end{equation}
It is a Bell transform followed by a controlled-$S$ phase correction
and a register swap, as required by our Hermitian Pauli convention.

\subsubsection{Decoded quantum interferometry (DQI) as special relation-free source-to-target transduction}
\label{subsubsec:spin-dqi}
For the binary instance (max-XORSAT) of decoded quantum interferometry
(DQI)~\cite{JordanEtAlDQI}, take $a_i=Z_i$ on $s$ source qubits and
$b_i=(-1)^{v_i}Z^{\beta_i}$ on $n_B$ target qubits, where
$\beta_i\in\mathbb F_2^{n_B}$. Then $K_A=\{0\}$ and all source
generators commute. The target Hamiltonian is diagonal, with
\begin{equation}
 H_B|x\rangle=f(x)|x\rangle,\qquad
 f(x):=\sum_{i=1}^s(-1)^{v_i+\beta_i\cdot x}.
 \label{eq:mf-dqi-objective}
\end{equation}
For $H_A=\sum_iZ_i$, the source Fourier coefficients of $p(H_A)$
depend only on the Hamming weight of the label and vanish above
$\deg p$. After the operator--Fourier transform they give the DQI
pilot state, a weighted superposition of Dicke states. Since $A=I_s$ in reduced \(Z\)-coordinates,
relative decoding reduces to ordinary syndrome decoding
$\mathsf{Dec}_{A\leftarrow B}:|Bc\rangle|0\rangle
 \longmapsto |Bc\rangle|Ac\rangle=|Bc\rangle|c\rangle$, and the phase
in Eq.~\eqref{eq:mf-pauli-monomial} is $(-1)^{v\cdot c}$.
After inverse Fourier transformation, the output is
\begin{equation}
 |\psi_B^p\rangle
 =\frac{\sum_{x\in\mathbb F_2^n}p(f(x))|x\rangle|x\rangle}
 {\bigl(\sum_x|p(f(x))|^2\bigr)^{1/2}}.
 \label{eq:mf-dqi-output}
\end{equation}
Bitwise controlled additions erase the second copy of $x$, leaving
the DQI state. On this diagonal algebra, the operator--Fourier transform
reduces to the ordinary binary Fourier transform (the Hadamard transform) after the same
duplicate-register removal.

\subsubsection{Relative decoding can go beyond the DQI Shannon bound}
\label{sec:relative-shannon}
Relative distance being larger than ordinary distance is not the only benefit. Relative decoding can also help go beyond the fundamental Shannon bound in original DQI, providing a new direction beyond using quantum decoding~\cite{JordanEtAlDQI,ShuttyEtAl2026LocallyQuantum,ChaillouxTillich2023QuantumDecoding}

Let \(B\in\mathbb F_2^{n_B\times s}\) have full row rank, with \(n_B/s\to R\) as \(s\to\infty\). For binary max-XORSAT, the DQI Shannon bound with classical decoding using the semi-circle law is~\cite{JordanEtAlDQI}
\begin{equation}
 \mu_{\mathrm{DQI}}
 \le \overline\mu_{\mathrm{DQI}}(R)
 :=\frac12+
 \sqrt{H_2^{-1}(R)\bigl[1-H_2^{-1}(R)\bigr]},
 \label{eq:dqi-shannon}
\end{equation}
where \(H_2\) is binary entropy and \(H_2^{-1}\) takes values in \([0,1/2]\).

Relative decoding recovers \(Ac\), rather than the full selection vector \(c\). For a compatible source with \(K_A\subseteq K_B\), let
\[
 S_D(A):=\{Ac:|c|\le D\},
 \qquad V_D:=\sum_{j=0}^{D}\binom{s}{j}.
\]
Exact relative decoding on \(S_D(A)\) requires distinct source labels to have distinct syndromes. Thus
\begin{equation}
 \log_2|S_D(A)|\le n_B,
 \qquad |S_D(A)|\le V_D.
 \label{eq:relative-count}
\end{equation}
Here \(V_D\) counts all selection vectors and source relations can make \(|S_D(A)|\) much smaller. Define the largest degree allowed by this counting condition in the range \(D\le s/2\):
\begin{equation}
 \begin{aligned}
 D_{\max}&:=\max\left\{D\in\{0,\ldots,\lfloor s/2\rfloor\}:
 \log_2|S_D(A)|\le n_B\right\},\\
 \delta_{\mathrm{rel}}&:=\limsup_{s\to\infty}\frac{D_{\max}}{s}.
 \end{aligned}
 \label{eq:relative-max-degree}
\end{equation}
Within this range, exact relative decoding on \(S_D(A)\) requires \(D\le D_{\max}\), so monotonicity and the semi-circle law gives
\begin{equation}
 \mu_{\mathrm{rel}}\le\overline\mu_{\mathrm{rel}}
 :=\frac12+\sqrt{\delta_{\mathrm{rel}}(1-\delta_{\mathrm{rel}})}.
 \label{eq:relative-shannon-score}
\end{equation}
The inequality \(|S_D(A)|\le V_D\) implies \(\delta_{\mathrm{rel}}\ge H_2^{-1}(R)\), hence \(\overline\mu_{\mathrm{rel}}\ge\overline\mu_{\mathrm{DQI}}(R)\), strictly when \(\delta_{\mathrm{rel}}>H_2^{-1}(R)\). Relation-free sources recover equality.

\subsubsection{Hamiltonian decoded quantum interferometry (HDQI) as special relation-free source-to-target transduction}
\label{subsubsec:spin-hdqi}

The usual Hamiltonian DQI (HDQI) construction is recovered from a
relation-free source with independent binary generator labels and the same
commutation matrix as the target~\cite{HamiltonianDQI}. For
\begin{equation}
 H_B=\sum_{i=1}^s \eta_i^B P_{n_B}(v_i),
 \qquad
 \eta_i^B\in\{\pm1\},
\end{equation}
set
\begin{equation}
 \Sigma_{ij}:=v_i^{\mathsf T}J_{n_B}v_j
\end{equation}
and choose, on $s$ source qubits,
\begin{equation}
 a_i:=X_i\prod_{j<i}Z_j^{\Sigma_{ij}},
 \qquad
 H_A:=\sum_{i=1}^s a_i,
 \qquad
 \pi(a_i):=\eta_i^B P_{n_B}(v_i).
 \label{eq:mf-hdqi-lift}
\end{equation}
The coefficients of $p(H_A)$ in the ordered-word basis $\{a_1^{c_1}\cdots a_s^{c_s}:c\in\mathbb F_2^s\}$ give the HDQI pilot state, up to the computable phases that
convert these ordered products to the chosen Hermitian Pauli basis.

The binary Pauli-label matrix of the source is
\begin{equation}
 A=
 \begin{pmatrix}
 I_s\\
 T
 \end{pmatrix},
 \qquad
 T_{ji}:=
 \begin{cases}
  \Sigma_{ij}, & j<i,\\
  0, & j\ge i.
 \end{cases}
 \label{eq:mf-hdqi-source-label-matrix}
\end{equation}
The $X$-block of $A$ is the identity, so $A$ is injective and
\begin{equation}
 K_A=\ker A=\{0\}.
\end{equation}
Thus every source Pauli label $Ac$ uniquely determines the
term-selection vector $c\in\mathbb F_2^s$.  Equivalently, on the source
label space we may canonically identify
\begin{equation}
 Ac\longleftrightarrow c.
\end{equation}
Therefore, the relative-decoding task
\begin{equation}
 Bc\longmapsto Ac
\end{equation}
becomes,
\begin{equation}
 Bc\longmapsto c.
 \label{eq:mf-hdqi-ordinary-decoding}
\end{equation}
Hence HDQI is the relation-free case of the present
transduction framework: relative decoding reduces to ordinary decoding.

\clearpage
\section{Applications: optimization and Gibbs states}
\label{sec:performance}
This section develops optimization and Gibbs-state preparation. We first
choose a guiding state to maximize an objective, such as the target
energy or a constraint satisfaction score. When the required source
moments are accessible classically, as in DQI, the optimal polynomial
guide and its expected score can be computed classically. We use this
calculation for the nonlinear-chain comparison in
Section~\ref{sec:applications}. We then turn to Gibbs-state preparation.
Both applications can be combined with the efficient realizations in
Sections~\ref{sec:main-framework} and~\ref{sec:generalizations}.

\subsection{Bounded-degree approximate optimization}
\label{subsec:value-general}
\label{sec:achieved-value}

Fix a degree bound $D$ and a self-adjoint target observable
$O_B$. We optimize its expectation over guides $X_A\in\mathcal F_D^A$
whose images $X_B:=\pi(X_A)$ are nonzero. The target marginal and its
objective value are
\begin{equation}
 \rho_B[X_A]:=\frac{X_BX_B^*}{\operatorname{Tr}(X_BX_B^*)},
 \qquad
 \mathcal E_B(O_B;X_A):=
 \frac{\tau_B(X_B^*O_BX_B)}{\tau_B(X_B^*X_B)}.
 \label{eq:value-born}
\end{equation}
We first formulate the finite variational problem, then identify when
relative distance makes it computable from the source, and compare full
word and polynomial guides. Proofs are in
Appendix~\ref{app:optimization-details}.

\subsubsection{Variational optimization and achievable values}
\label{subsec:value-relative}

Choose $\mathcal W_D=\operatorname{span}\{F_0,\ldots,F_{R-1}\}
\subseteq\mathcal F_D^A$, write $X_A(c)=\sum_jc_jF_j$, and set
$Y_j=\pi(F_j)$. Define
\begin{equation}
 (G_\pi)_{ij}:=\tau_B(Y_i^*Y_j),\qquad
 (K_\pi)_{ij}:=\tau_B(Y_i^*O_BY_j).
 \label{eq:value-general-matrices}
\end{equation}
These matrices use the target trace, or equivalently the pulled-back
trace $\tau_\pi:=\tau_B\circ\pi$. Even if $O_B=\pi(O_A)$, one cannot
in general replace $\tau_\pi$ by $\tau_A$.

\begin{proposition}[Variational objective matrix; proof in Appendix~\ref{app:optimization-variational}]
\label{prop:value-ritz}
If the guide space has a nonzero target image, then
\begin{equation}
 \mathcal E_B(O_B;X_A(c))=\frac{c^*K_\pi c}{c^*G_\pi c},\qquad
 c^*G_\pi c>0,
 \label{eq:value-general-rayleigh}
\end{equation}
and its maximum is
\begin{equation}
 \max_{c:\,c^*G_\pi c>0}\mathcal E_B(O_B;X_A(c))
 =\lambda_{\max}\!\left(G_+^{-1/2}K_+G_+^{-1/2}\right).
 \label{eq:value-general-optimum}
\end{equation}
Here $G_+$ and $K_+$ are the restrictions to $\operatorname{ran}G_\pi$,
where $G_+>0$. For minimization use the smallest eigenvalue.
\end{proposition}
Indeed, $\ker G_\pi\subseteq\ker K_\pi$: a zero-norm coefficient
vector represents the zero target operator. Removing these directions
reduces the problem to the Rayleigh--Ritz principle~\cite{HornJohnson2012}, or equivalently
$K_\pi c=\lambda G_\pi c$ on the positive-norm subspace.

Relative distance identifies when the target variational matrices can be
evaluated entirely from source data.

\begin{theorem}[Expected-value transfer under relative decoding; proof in Appendix~\ref{subsec:value-relative:details}]
\label{thm:value-relative-expectation}
Let $0\ne X_A\in\mathcal F_D^A$ and
$O_A=O_A^*\in\mathcal F_r^A$. If
\begin{equation}
 2D+r<d_{\mathrm{rel}},
 \label{eq:value-observable-window}
\end{equation}
then $\pi(X_A)\ne0$ and
\begin{equation}
 \mathcal E_B(\pi(O_A);X_A)
 =\frac{\tau_A(X_A^*O_AX_A)}{\tau_A(X_A^*X_A)}.
 \label{eq:value-relative-mean}
\end{equation}
Thus the source Gram and objective matrices determine the target
variational problem on every $\mathcal W_D\subseteq\mathcal F_D^A$.
\end{theorem}
The products in the denominator and numerator have degrees at most
$2D$ and $2D+r$, so Eq.~\eqref{eq:mf-trace-distance-characterization} identifies their
traces. This also explains why transferring a degree-$D$ guide requires
less than transferring its observable value. For
$H_A=\sum_i h_i a_i$ and $H_B=\pi(H_A)$, the energy window is
\begin{equation}
 2D+1<d_{\mathrm{rel}},
 \label{eq:value-energy-window}
\end{equation}
whereas guide transduction requires only $2D<d_{\mathrm{rel}}$,
together with the circuit assumptions of the main framework.

Thus accessible source moments determine the expected target energy.
If the stronger window $2D+2<d_{\mathrm{rel}}$ holds, applying the same
result to $H_A^2$ also determines the second moment, and hence the
variance, from the source. In the nonlinear-chain example these
calculations reduce to a small tridiagonal matrix, without simulating
the quantum circuit. We now compare optimization over the full word
space with optimization over polynomial guides.

\paragraph{The full degree-$D$ word space.}
In the Pauli case, choose one ordered word
$W_A(c)=a_1^{c_1}\cdots a_s^{c_s}$ of weight at most $D$ for each
source coset in the radius-$D$ quotient ball. These words form an
orthonormal basis indexed by $\mathcal R_D$. With
$W_B(c)=\pi(W_A(c))$, define for $L\in\{A,B\}$
\begin{equation}
 (G_D^L)_{cd}=\tau_L(W_L(c)^*W_L(d)),\qquad
 (M_D^L)_{cd}=\tau_L(W_L(c)^*H_LW_L(d)).
 \label{eq:value-relative-matrix-meaning}
\end{equation}
For $X_A(z)=\sum_c z_cW_A(c)$, the exact target quotient is
$z^*M_D^Bz/(z^*G_D^Bz)$. Relative distance gives
\begin{equation}
 2D<d_{\mathrm{rel}}\Longrightarrow G_D^B=I,\qquad
 2D+1<d_{\mathrm{rel}}\Longrightarrow M_D^B=M_D^A,
 \label{eq:value-relative-matrix-windows}
\end{equation}
and hence
\begin{equation}
 \max_{0\ne X_A\in\mathcal F_D^A}\mathcal E_B(H_B;X_A)
 =\lambda_{\max}(M_D^A),\qquad 2D+1<d_{\mathrm{rel}}.
 \label{eq:value-full-relative-optimum}
\end{equation}
Its eigenvector specifies the optimal word coefficients. Since the dimension
$|\mathcal R_D|$ can grow as $\sum_{j=0}^D\binom{s}{j}$, direct
optimization over the full word space can require exponentially large
matrices.
Appendix~\ref{app:optimization-spaces} gives the explicit Pauli phase
and scalar-relation formulas for these matrix entries.

\paragraph{Polynomial guides and the Jacobi matrix.}
Restricting the guide to a polynomial in the Hamiltonian leaves at most
$D+1$ coefficients to choose. The resulting optimization is an
eigenvalue problem: multiplication by energy, expressed in an
orthonormal polynomial basis, gives a tridiagonal matrix whose largest
eigenvalue is the best attainable mean energy.

For every represented family, the smaller choice
$X_A=p(H_A)$, $\deg p\le D$, uses
$\mathcal K_D(H_A)=\operatorname{span}\{I,H_A,\ldots,H_A^D\}$.
Writing $H_A=\sum_e eP_e^A$, let
$\mu_A(e)=\tau_A(P_e^A)$ and the moment $m_k=\sum_e\mu_A(e)e^k=\tau_A(H_A^k)$.
In the energy window,
\begin{equation}
 \mathcal E_B(H_B;p(H_A))
 =\frac{\sum_e\mu_A(e)e|p(e)|^2}{\sum_e\mu_A(e)|p(e)|^2}
 =\frac{c^*K_D^{\mathrm{poly}}c}{c^*G_D^{\mathrm{poly}}c},
 \label{eq:value-polynomial-rayleigh}
\end{equation}
where $p(x)=\sum_{j=0}^Dc_jx^j$,
$(G_D^{\mathrm{poly}})_{ij}=m_{i+j}$, and
$(K_D^{\mathrm{poly}})_{ij}=m_{i+j+1}$. Outside that window, the
exact spectral formula instead uses $\nu_\pi(e)=\tau_B(\pi(P_e^A))$.

Orthonormalizing $1,x,\ldots,x^D$ in $L^2(\mu_A)$ gives
$\varphi_0,\ldots,\varphi_{r_D-1}$, where
$r_D=\min\{D+1,|\operatorname{spec}(H_A)|\}$ and zero-norm
directions are removed. The compressed multiplication operator is the
real symmetric tridiagonal Jacobi matrix
\begin{equation}
 (J_D)_{ij}:=\sum_e\mu_A(e)\varphi_i(e)e\varphi_j(e).
 \label{eq:value-jacobi-definition}
\end{equation}
\begin{proposition}[Optimal polynomial-guide value; proof in Appendix~\ref{app:optimization-spaces}]
\label{prop:value-polynomial-optimum}
If $H_A\in\mathcal F_1^A$ and $2D+1<d_{\mathrm{rel}}$, then
\begin{equation}
 \max_{\substack{\deg p\le D\\\tau_A(|p(H_A)|^2)>0}}
 \mathcal E_B(H_B;p(H_A))=\lambda_{\max}(J_D).
 \label{eq:value-jacobi-optimum}
\end{equation}
A unit top eigenvector $z$ gives
$p_\star=\sum_{j=0}^{r_D-1}z_j\varphi_j$; use the bottom eigenpair
for minimization.
\end{proposition}
Polynomial
guides optimize a smaller family, with a matrix of size at most $D+1$. It is easy to see $\lambda_{\max}(J_D)\le\lambda_{\max}(M_D^A)
\le\lambda_{\max}(H_B)$.

\paragraph{Computability.}
Building $J_D$ requires the source moments $\tau_A(H_A^k)$ for
$k\le 2D+1$. When these moments can be computed classically with
polynomial cost and precision sufficient for stable orthogonalization,
$J_D$ and its optimal value can be obtained in polynomial time. Independent-block sources give such moment access when
their block moments are efficient, by
Eq.~\eqref{eq:value-independent-block-moments} in
Appendix~\ref{app:optimization-spaces}.
For ordinary DQI, $H_A=\sum_{i=1}^sZ_i$ gives the explicit matrix
\begin{equation}
 (J_D)_{j,j}=0,\qquad
 (J_D)_{j,j+1}=(J_D)_{j+1,j}=\sqrt{(j+1)(s-j)}.
 \label{eq:value-dqi-jacobi}
\end{equation}
Its extremal eigenvector selects the polynomial, and hence the Dicke
guide of Section~\ref{subsubsec:spin-dqi}.

\subsubsection{Limitations of expectation-value advantages and sampling}
\label{subsec:value-observable-limitation}

Theorem~\ref{thm:value-relative-expectation} identifies the target
expectation of $O_B=\pi(O_A)$ with its source expectation whenever the
source lift $O_A$ has marked degree $r$ and $2D+r<d_{\mathrm{rel}}$.
Given an efficiently accessible lift, a source estimator therefore needs
no transduction circuit (Corollary~\ref{cor:value-source-only-observables}).
This conclusion concerns marked degree; physical locality alone does
not guarantee such a lift. The estimate is stable under source-preparation
errors: a source error $\eta_{\mathrm{src}}$ in trace distance adds at most
$2\|O_A\|\eta_{\mathrm{src}}$ to the estimation error
(Proposition~\ref{prop:value-source-estimation-stability}). Within this
window, transduction provides no additional advantage for expectation
estimation, although the source computation may itself require a quantum
computer. 

Producing target samples is a different task. DQI prepares feasible
assignments, while its mean score is already source-computable. For a
normalized sampled score $S\in[0,1]$ of mean $\mu$,
\begin{equation}
 \Pr\{S>c\}\ge\frac{\mu-c}{1-c},\qquad 0\le c<\mu.
 \label{eq:value-threshold-from-mean}
\end{equation}
Eq.~\eqref{eq:value-threshold-from-mean} converts the mean into
a probability of obtaining a high-scoring target sample. Potential
advantages of transduction therefore lie in producing target samples or
accessing quantities outside the trace-matching window, such as sharp
score indicators and ground-space overlap. Using the target guiding state
for quantum phase estimation or ground-state
preparation~\cite{AbramsLloyd1999,LinTong2020} depends on this
full energy distribution; a mean bound alone guarantees ground-space
overlap only with additional spectral information.

\subsection{Thermofield-double (TFD) transduction for Gibbs sampling}
\label{sec:purified-gibbs-guides}
\label{subsec:gibbs-interface}
The polynomial guiding states of Section~\ref{sec:main-framework} can
also be used to transfer thermal states. The Gibbs-state application was
considered in HDQI~\cite{HamiltonianDQI}. Here we take a
coherently prepared source thermofield double as input, separating source
preparation from the transduction step. Existing preparation algorithms
supply this input in the regimes discussed below.
A common low-degree approximation to the source and target Gibbs
amplitudes controls the transfer error even when the exponential filter
lies outside the degree-$D$ subspace. The approximating polynomial $Q$
enters only the error analysis; the algorithm's input is the source TFD
itself. Full derivations appear in Appendix~\ref{app:gibbs-details}.

\paragraph{When is the source TFD efficiently available?}
The preparation of thermofield double states has received significant attention recently. Several preparation methods exist by now, including exponential
filtering~\cite{PoulinWocjan2009ThermalGibbs,ChowdhurySomma2017Gibbs,AnChildsLin2026Nonunitary, GilyenEtAl2019QSVT,ShangEtAl2025LindbladianODE,ShangEtAl2025FastForwarding},
quantum walks with coherent-access and energy-resolution
assumptions~\cite{YungAspuruGuzik2012QuantumMetropolis,WocjanTemme2023Szegedy}, and parent-Hamiltonian
constructions~\cite{CottrellEtAl2019,ChenKastoryanoBrandaoGilyen2025ThermalSimulation,ChenKastoryanoGilyen2023ExactSampler}.
The latter give polynomial-time adiabatic preparation when the parent
path is efficiently simulable, has an inverse-polynomial minimum gap,
and has polynomially bounded length and smoothness parameters.
Concrete guarantees include local Hamiltonians at
sufficiently high temperature~\cite{RouzeFrancaAlhambra2026}, $1D$ systems~\cite{BergamaschiChen1D} and local
commuting Hamiltonians with uniformly gapped preparation
paths~\cite{GeMolnarCirac2016Adiabatic}. For local commuting sources,
detectability-lemma methods also give coherent annealing with improved
gap dependence under their parent-access and path-gap
assumptions~\cite{FangLuTongZhao2026Detectability}.

At any fixed positive temperature, sufficiently weak quasi-local
fermionic interactions admit a system-size-independent parent
gap~\cite{SmidMeisterBertaBondesan2025FermiHubbard}. Combined with coherent
parent simulation, this allows adiabatic preparation along an interaction
path starting from a free-fermion TFD, within the temperature-dependent
perturbative regime. Modular-annihilator constructions provide further
parent Hamiltonians and controlled
approximations~\cite{YiTakahashiZhou2026ModularAnnihilator}; their use
requires efficient access to the dressed operators and control of the gap.

Thus, all of these methods and regimes can be explored in our framework as efficient sources.
Factorized-parent filtering improves gap dependence given coherent
factor access and a preparable state with constant overlap with the
canonical TFD~\cite{LengJiangLin2026WalkFree}. Fluctuation-theorem methods
prepare a target TFD from a simpler source TFD, with costs controlled by
work and free-energy quantities~\cite{HolmesEtAl2022Fluctuation}.
For a shifted Hamiltonian $H_A\ge0$, direct filtering of Bell pairs
retains an amplification factor
$\sqrt{\dim(\mathcal H_A)/Z_A(\beta)}$~\cite{AnChildsLin2026Nonunitary}.
It is efficient when this factor, the inverse temperature, the
Hamiltonian normalization, and the costs of coherent access are polynomially
bounded. Thus, all of these methods and regimes can be explored in our framework as efficient sources.

\paragraph{TFD transduction.} Let $H_A=H_A^*\in\mathcal A$ have marked degree at most one, and set
$H_B:=\pi(H_A)$. A unital $*$-homomorphism satisfies
$\operatorname{spec}(H_B)\subseteq\operatorname{spec}(H_A)$ and
$\pi(f(H_A))=f(H_B)$. Appendix~\ref{subsec:gibbs-spectral-cover} states these facts and
describes which spectral projections survive; their proof is given in
Appendix~\ref{app:spectral-properties}. Thus
target energies can disappear or change multiplicity, but no new energy
can appear. In particular, for every real polynomial $Q$ and $\beta\ge0$,
\begin{equation}
 \begin{aligned}
 \bigl\|e^{\beta H_B/2}Q(H_B)-I_B\bigr\|
 &=\max_{E\in\operatorname{spec}(H_B)}
       \bigl|e^{\beta E/2}Q(E)-1\bigr|\\
 &\le\bigl\|e^{\beta H_A/2}Q(H_A)-I_A\bigr\|.
 \end{aligned}
 \label{eq:gibbs-inherited-relative-error}
\end{equation}
A relative-amplitude approximation on a source spectral interval
therefore controls both systems.

For $L\in\{A,B\}$, define $Z_L(\beta):=\Tr(e^{-\beta H_L})$ and the
canonical purification
\begin{equation}
 \begin{aligned}
 |\mathrm{TFD}_L(\beta)\rangle
 &:=\frac{(e^{-\beta H_L/2}\otimes I)|\Phi_L\rangle}
        {\sqrt{\tau_L(e^{-\beta H_L})}}\\
 &=\frac1{\sqrt{Z_L(\beta)}}\sum_j e^{-\beta E_j^L/2}
       |E_j^L\rangle|\overline{E_j^L}\rangle.
 \end{aligned}
 \label{eq:gibbs-tfd}
\end{equation}
Its reduced density matrix is the Gibbs state $e^{-\beta H_L}/Z_L(\beta)$.
We assume that this canonical source purification can be prepared
coherently and efficiently to the required accuracy. The circuit acts on
both registers, so the Gibbs marginal alone is insufficient; any reference
transformation must be known and corrected.

Fix $D$ with $2D<d_{\mathrm{rel}}$ and the compatible bases, efficient
Fourier transforms, label maps, and relative decoder of
Theorem~\ref{thm:mf-polynomial-transduction}. Let $U_D$ be a fixed
unitary realization on a common enlarged input/output space, with $J$
collecting the workspace. For every real polynomial $Q$ of degree at
most $D$ with $Q(H_B)\ne0$, it obeys
\begin{equation}
 \begin{aligned}
 U_D\bigl(|\psi_A^Q\rangle|0\rangle_J\bigr)
 =|\psi_B^Q\rangle|0\rangle_J,~
 |\psi_L^Q\rangle
 :=\frac{(Q(H_L)\otimes I)|\Phi_L\rangle}
          {\sqrt{\tau_L(Q(H_L)^2)}}.
 \end{aligned}
 \label{eq:gibbs-polynomial-map}
\end{equation}

\begin{proposition}[Gibbs transduction from a polynomial certificate; proof in Appendix~\ref{subsec:gibbs-interface:details}]
\label{prop:gibbs-certificate}
Suppose a real polynomial $Q$ of degree at most $D$ satisfies
\begin{equation}
 \sup_{x\in I}\bigl|e^{\beta x/2}Q(x)-1\bigr|\le\eta<1
 \label{eq:gibbs-relative-approximation}
\end{equation}
on an interval $I\supseteq\operatorname{spec}(H_A)$, and the prepared
source state $\widehat\rho_A$ is within trace distance
$\delta_{\mathrm{src}}$ of
$|\mathrm{TFD}_A(\beta)\rangle\langle\mathrm{TFD}_A(\beta)|$.
Writing
$\widehat\omega_B:=U_D(\widehat\rho_A\otimes|0\rangle\langle0|_J)U_D^*$,
we have
\begin{equation}
 \frac12\bigl\|\widehat\omega_B-
 |\mathrm{TFD}_B(\beta),0\rangle
 \langle\mathrm{TFD}_B(\beta),0|\bigr\|_1
 \le\delta_{\mathrm{src}}+2\eta.
 \label{eq:gibbs-transfer-error}
\end{equation}
The same bound holds for the Gibbs marginal after discarding the
reference and workspace.
\end{proposition}
For a concrete approximation bound, choose a source spectral interval
$[E_-,E_+]$ of
width $W:=E_+-E_->0$. The degree-$D$ Taylor polynomial about $E_+$ has
relative error
$\Pr\{\operatorname{Pois}(\beta W/2)\ge D+1\}$.
Appendix~\ref{app:gibbs-degree} gives the explicit Taylor polynomial
and the exact Poisson-tail identity. It also states and proves
Proposition~\ref{prop:gibbs-degree}, which bounds this tail, gives the
sufficient degree as a function of bandwidth, temperature, and accuracy,
and yields the following window.

\begin{corollary}[Inverse-temperature window from relative decoding; proof in Appendix~\ref{subsec:gibbs-temperature}]
\label{thm:gibbs-temperature}
Let $D_{\mathrm{rel}}\le\lfloor(d_{\mathrm{rel}}-1)/2\rfloor$ be a degree
at which the preceding efficient transduction assumptions hold. Fix
$0<\varepsilon<1$ and $0\le\delta_{\mathrm{src}}<\varepsilon$, and put
$k:=D_{\mathrm{rel}}+1$ and
$\ell:=\ln[2/(\varepsilon-\delta_{\mathrm{src}})]$.
Assume the source TFD is efficiently preparable to trace distance
$\delta_{\mathrm{src}}$ at the chosen $\beta$. If $k\ge\ell/3$, the
target TFD, and hence its Gibbs marginal, can be prepared to trace
distance at most $\varepsilon$ whenever
\begin{equation}
 0\le\beta\le\frac{2}{W}
 \left(\sqrt{k+\frac{\ell}{6}}-\sqrt{\frac{\ell}{2}}\right)^2.
 \label{eq:gibbs-temperature-window}
\end{equation}
In particular, suppose $H_L=\sum_{i=1}^{m}h_i a_i^L$ with
$\|a_i^L\|\le1$ and $|h_i|\le J$, where $J>0$, so one may take
$W\le2Jm$. If $d_{\mathrm{rel}}\ge\delta_{\mathrm{rel}}m$, efficient
transduction reaches the full radius
$D_{\mathrm{rel}}=\lfloor(d_{\mathrm{rel}}-1)/2\rfloor$, and
$\ell=o(m)$, then every fixed
\begin{equation}
 0\le\beta<\frac{\delta_{\mathrm{rel}}}{2J}
 \label{eq:gibbs-constant-temperature}
\end{equation}
is attainable for sufficiently large $m$, subject to the same source
preparation assumption.
\end{corollary}

For $k\gg\ell$, the window is $(1-o(1))\,2k/W$, so linear relative
distance suffices for constant inverse temperature when $W=O(m)$ and
the other assumptions hold. This bound is sufficient but not necessary:
a smaller bandwidth, a tailored polynomial, or additional structure can
give a wider window.

\clearpage
\section{Generalizations beyond qubit spins}
\label{sec:generalizations}
Optimization and Gibbs-state preparation in Section~\ref{sec:performance} are not confined to qubit
spin systems. In quantum chemistry and many-body physics, the same
tasks arise for fermionic and bosonic Hamiltonians, as well as systems
with more than two local levels. The operator-algebraic formulation
allows us to address these settings without developing a separate
transduction principle for each: the construction relies on a
homomorphism between operator algebras, rather than on a particular
physical realization. We therefore extend the framework to fermions
using Majorana operators and to qudits and occupation-truncated
bosonic systems using Weyl operators, adapting the source preparation,
operator Fourier transforms, and relative decoding to each setting.

\subsection{Fermions and Majorana coordinates}
\label{subsec:families-fermion}
\subsubsection{Formulation: Majorana generators, relative codes, and Fourier coordinates}
\label{subsubsec:fermion-formulation}

\paragraph{Physical fermionic Hamiltonians.}

We begin with number-conserving fermionic Hamiltonians, including
molecular electronic and Fermi--Hubbard Hamiltonians~\cite{McArdleEtAl2020,ArovasEtAl2022}. For
$L\in\{A,B\}$, let $f_{L,j},f_{L,j}^*$ act on the full Fock space of
$n_L$ modes and satisfy
\begin{equation}
 \{f_{L,i},f_{L,j}^*\}=\delta_{ij}I_L,
 \qquad
 \{f_{L,i},f_{L,j}\}=0.
\end{equation}
A physical Hamiltonian with one- and two-particle terms has the form
\begin{equation}
 H_L=E_{c,L}I_L
 +\sum_{i,j}t_{ij}^L f_{L,i}^*f_{L,j}
 +\sum_{i,j,k,\ell}V_{ij;k\ell}^L
     f_{L,i}^*f_{L,j}^*f_{L,k}f_{L,\ell},
 \qquad H_L=H_L^*.
 \label{eq:fermion-general-interacting-hamiltonian}
\end{equation}
The coefficients include the chosen antisymmetry and counting
conventions. Each term creates and annihilates the same number of
fermions, so
\begin{equation}
 [H_L,\widehat N_L]=0,
 \qquad
 \widehat N_L:=\sum_j f_{L,j}^*f_{L,j}.
\end{equation}
Consequently $H_L$ also preserves fermion parity:
\begin{equation}
 [H_L,\Pi_L]=0,
 \qquad
 \Pi_L:=(-1)^{\widehat N_L}.
\end{equation}
Thus the Hamiltonian does not mix even- and odd-particle sectors.
The construction below requires only parity preservation, so it also
allows pairing terms and higher even interactions that need not
conserve particle number. Our goal is to prepare the target polynomial state of $p(H_B)$ from that of source $p(H_A)$.

\paragraph{Majorana expansion and product phases.}
Introduce the Majorana operators~\cite{BravyiKitaev2002,Bravyi2005}
\begin{equation}
 \gamma_{L,2j-1}:=f_{L,j}+f_{L,j}^*,
 \qquad
 \gamma_{L,2j}:=i(f_{L,j}^*-f_{L,j}),
 \qquad
 f_{L,j}=\frac{\gamma_{L,2j-1}+i\gamma_{L,2j}}2.
 \label{eq:fermion-ladder-majorana-conversion}
\end{equation}
They satisfy $\gamma_{L,j}^*=\gamma_{L,j}$ and
$\{\gamma_{L,j},\gamma_{L,k}\}=2\delta_{jk}I_L$.
For $u\in\mathbb F_2^{2n_L}$, define the Hermitian monomial
\begin{equation}
 \Gamma_L(u):=
 i^{|u|(|u|-1)/2}
 \prod_{j=1}^{2n_L}\gamma_{L,j}^{u_j},
 \label{eq:families-majorana-convention}
\end{equation}
with increasing product order. These monomials satisfy
$\Gamma_L(u)^*=\Gamma_L(u)$ and $\Gamma_L(u)^2=I_L$.
They form an orthonormal operator basis for the full-Fock normalized
trace $\tau_L=2^{-n_L}\operatorname{Tr}$. Since
$\Pi_L\Gamma_L(u)\Pi_L=(-1)^{|u|}\Gamma_L(u)$,
parity preservation gives
\begin{equation}
 H_L=\sum_{u:\,|u|\ {\rm even}}h_u^L\Gamma_L(u),
 \qquad
 h_u^L=2^{-n_L}\operatorname{Tr}\bigl(\Gamma_L(u)H_L\bigr)
 \in\mathbb R.
 \label{eq:fermion-general-majorana-expansion}
\end{equation}

The one- and two-particle terms can be written directly in this
Hermitian Majorana basis. Let $e_a$ denote the $a$th standard basis
vector of $\mathbb F_2^{2n_L}$. For an ordered list of Majorana indices
$\mathbf q=(q_1,\ldots,q_r)$, define
\begin{equation}
 u(\mathbf q):=\sum_{\nu=1}^r e_{q_\nu},
 \qquad
 \xi(\mathbf q):=
 (-1)^{\sum_{\nu<\mu}\mathbf1_{\{q_\nu>q_\mu\}}}
 i^{-|u(\mathbf q)|(|u(\mathbf q)|-1)/2},
\end{equation}
where label addition is binary. The list $\mathbf q$ records the
ordered factors, including repetitions, whereas $u(\mathbf q)$ records
which Majorana indices occur an odd number of times. Thus the list
length $r=|\mathbf q|$ need not equal the reduced weight
$|u(\mathbf q)|$. Reordering the factors and cancelling repeated
Majoranas gives
\begin{equation}
 \gamma_{L,q_1}\cdots\gamma_{L,q_r}
 =\xi(\mathbf q)\Gamma_L\bigl(u(\mathbf q)\bigr).
\end{equation}
The sign in $\xi$ records the reordering, while its power of $i$
converts the ordered product to the Hermitian convention. Thus
Eq.~\eqref{eq:fermion-ladder-majorana-conversion} gives
\begin{equation}
 \begin{aligned}
 f_{L,i}^*f_{L,j}
 &=
 \frac14\sum_{a,b=0}^1(-i)^a i^b\,
 \xi(2i-1+a,2j-1+b)\,
 \Gamma_L\bigl(e_{2i-1+a}+e_{2j-1+b}\bigr),\\[1mm]
 f_{L,i}^*f_{L,j}^*f_{L,k}f_{L,\ell}
 &=
 \frac1{16}\sum_{a,b,c,d=0}^1
 (-i)^{a+b}i^{c+d}\,
 \xi(\mathbf q_{abcd})\Gamma_L\bigl(u(\mathbf q_{abcd})\bigr),
 \end{aligned}
 \label{eq:fermion-density-example}
\end{equation}
where
\begin{equation}
 \mathbf q_{abcd}:=
 (2i-1+a,\;2j-1+b,\;2k-1+c,\;2\ell-1+d).
\end{equation}
These formulas hold for arbitrary mode indices, including repeated
ones, and $\Gamma_L(0)=I_L$. They contain at most four and sixteen
Majorana monomials, respectively. Their labels have weight zero or
two for a one-particle term, and weight zero, two, or four for a
two-particle term. The coefficients of an individual ladder-operator
term may be complex, but combining the terms in the self-adjoint
Hamiltonian gives the real coefficients in
Eq.~\eqref{eq:fermion-general-majorana-expansion}. More generally, a
product of $2k$ ladder operators expands into at most $2^{2k}$ Majorana
monomials before collecting terms, giving polynomial overhead for
fixed interaction order.

\paragraph{From product phases to signed Hamiltonian generators.}
After expanding the ladder operators, write the Hamiltonian as a finite
sum of ordered Majorana products,
\begin{equation}
 H_L=\sum_{\mathbf q}g_{\mathbf q}^L
       \gamma_{L,q_1}\cdots\gamma_{L,q_{|\mathbf q|}},
\end{equation}
where the empty product is $I_L$. The coefficient $g_{\mathbf q}^L$
includes the original Hamiltonian coefficient and the factors of $1/2$
and $\pm i$ from the ladder-operator expansion. Reducing each product
and collecting all contributions with the same label gives
\begin{equation}
 H_L=\sum_u h_u^L\Gamma_L(u),
 \qquad
 h_u^L=\sum_{\mathbf q:\,u(\mathbf q)=u}
             g_{\mathbf q}^L\,\xi(\mathbf q).
\end{equation}
Although $g_{\mathbf q}^L\xi(\mathbf q)$ may be complex, the collected
coefficient $h_u^L$ is real because $H_L$ is self-adjoint and the
$\Gamma_L(u)$ form a Hermitian operator basis. Terms with zero
collected coefficient are omitted.

For transduction, choose a source--target pair with matched nonzero
coefficient magnitudes and the same collected scalar term,
\begin{equation}
 h_0^A=h_0^B=:E_c,
 \qquad
 |h_{w_i}^A|=|h_{v_i}^B|=:\lambda_i>0,
 \qquad 1\le i\le s,
\end{equation}
and define the generator signs by
\begin{equation}
 \eta_i^A:=\frac{h_{w_i}^A}{|h_{w_i}^A|},
 \qquad
 \eta_i^B:=\frac{h_{v_i}^B}{|h_{v_i}^B|},
 \qquad
 \eta_i^L\in\{\pm1\}.
\end{equation}
Thus $\eta_i^L$ is the sign of the complete collected coefficient,
which incorporates the original couplings, the ladder-expansion
factors, and the product phases $\xi(\mathbf q)$. With these definitions, the Hamiltonians
and marked generators are
\begin{equation}
 \begin{aligned}
 H_A&=E_cI_A+\sum_{i=1}^s\lambda_i a_i,
 &a_i&:=\eta_i^A\Gamma_A(w_i),\\
 H_B&=E_cI_B+\sum_{i=1}^s\lambda_i b_i,
 &b_i&:=\eta_i^B\Gamma_B(v_i),
 \end{aligned}
 \qquad
 \lambda_i>0,\quad \eta_i^L\in\{\pm1\}.
 \label{eq:fermion-matched-hamiltonians}
\end{equation}
Here $w_i$ and $v_i$ are nonzero even-weight Majorana labels.
The signed Majorana terms $a_i,b_i$ are the marked generators of
\begin{equation}
 \mathcal A:=C^*(I_A,a_1,\ldots,a_s),
 \qquad
 \mathcal B:=C^*(I_B,b_1,\ldots,b_s).
\end{equation}
Each marked term has degree one, even when it contains four or more
Majoranas; marked degree is not Majorana support size. Individual
Majorana terms need not conserve particle number separately, even
when their sum does.

Define the binary label matrices and relation codes by
\begin{equation}
 A:=[w_1\ \cdots\ w_s],
 \qquad
 B:=[v_1\ \cdots\ v_s],
 \qquad
 K_A:=\ker A,\quad K_B:=\ker B.
 \label{eq:fermion-label-matrices}
\end{equation}
All label arithmetic is over $\mathbb F_2$. For
$c\in\mathbb F_2^s$, use the ordered products
\begin{equation}
 W_A(c):=a_1^{c_1}\cdots a_s^{c_s},
 \qquad
 W_B(c):=b_1^{c_1}\cdots b_s^{c_s}.
\end{equation}
For $c\in K_L$, write $W_L(c)=\zeta_L(c)I_L$, with
$\zeta_L(c)\in\{\pm1,\pm i\}$. Thus $K_L$ records which products
become scalar, and $\zeta_L$ records their scalar phases.

\begin{theorem}[Homomorphisms between Majorana-generated algebras]
\label{thm:fermion-homomorphism}
The assignment $a_i\mapsto b_i$ extends to a surjective unital
$*$-homomorphism $\pi:\mathcal A\to\mathcal B$ if and only if
\begin{equation}
 A^{\mathsf T}A=B^{\mathsf T}B,
 \qquad
 K_A\subseteq K_B,
 \qquad
 \zeta_A(c)=\zeta_B(c)\quad(c\in K_A).
 \label{eq:fermion-homomorphism-conditions}
\end{equation}
For the matched Hamiltonians in
Eq.~\eqref{eq:fermion-matched-hamiltonians}, such a map satisfies
\begin{equation}
 \pi(H_A)=H_B,
 \qquad
 \pi\bigl(p(H_A)\bigr)=p(H_B)
\end{equation}
for every polynomial $p$.
\end{theorem}
\noindent The origin of these conditions is similar to those in Theorem~\ref{theorem:pauli}. Proofs can be found in Appendix~\ref{app:homomorphism-criteria}.

\paragraph{Relative decoding and distance.}
Under Eq.~\eqref{eq:fermion-homomorphism-conditions}, the forward
label map is well defined by
\begin{equation}
 M:\operatorname{im}A\longrightarrow\operatorname{im}B,
 \qquad
 M(Ac):=Bc.
\end{equation}
For $u\in\operatorname{im}A$, define
\begin{equation}
 w_A(u):=\min_{c:\,Ac=u}|c|.
\end{equation}
The degree-$D$ promise and relative distance are
\begin{equation}
 S_D:=\{u\in\operatorname{im}A:w_A(u)\le D\},
 \qquad
 d_{\mathrm{rel}}
 =\min_{c\in K_B\setminus K_A}|c|,
 \label{eq:families-fermion-distance}
\end{equation}
with $d_{\mathrm{rel}}=+\infty$ when $K_A=K_B$.
This agrees with Definition~\ref{def:mf-intrinsic-distance}: an ordered word
has zero full-Fock trace unless its label is zero, and inherited
scalar relations have the same phase on both sides. The required coherent relative decoder performs
\begin{equation}
 \operatorname{Dec}_{A\leftarrow B}:
 |Bc\rangle|0\rangle\longmapsto|Bc\rangle|Ac\rangle,
 \qquad |c|\le D.
 \label{eq:fermion-relative-decoder}
\end{equation}
When $K_A=\{0\}$, the source is relation free and recovering $Ac$
is equivalent to recovering $c$, going back to the ordinary decoding.

\paragraph{Operator basis and Fourier coordinates.}
Choose the generated Majorana bases
\begin{equation}
 E_u^A:=\Gamma_A(u)\quad(u\in\operatorname{im}A),
 \qquad
 E_v^B:=\Gamma_B(v)\quad(v\in\operatorname{im}B).
\end{equation}
Their orthonormality follows from
\begin{equation}
 \tau_L\bigl(\Gamma_L(u)^*\Gamma_L(v)\bigr)=\delta_{u,v}.
\end{equation}
Word reduction gives
\begin{equation}
 \mathcal F_D^A=\operatorname{span}\{E_u^A:u\in S_D\},
\end{equation}
which verifies generator--basis compatibility (C1) at the chosen
cutoff $D$. To describe the basis images, write
\begin{equation}
 W_A(c)=\nu_A(c)\Gamma_A(Ac),
 \qquad
 W_B(c)=\nu_B(c)\Gamma_B(Bc).
\end{equation}
Here $\nu_L(c)$ includes both the generator signs
$\prod_i(\eta_i^L)^{c_i}$ and the phase from multiplying their
Majorana basis operators. Then
\begin{equation}
 \pi(E_u^A)=\chi(u)E_{M(u)}^B,
 \qquad
 \chi(Ac):=\frac{\nu_B(c)}{\nu_A(c)}\in\{\pm1\}.
 \label{eq:fermion-basis-map}
\end{equation}
Compatibility makes this phase independent of the representative
$c$. It is a real sign because both basis operators are Hermitian
and $\pi$ preserves adjoints and unitarity. It can be computed by
choosing any solution of $Ac=u$ and multiplying the corresponding
Majorana terms in the fixed order; no minimum-weight search is needed.
This verifies (C2).

In the vectorization convention of the main framework, the
operator--Fourier transforms are
\begin{equation}
 \mathsf F_L|E_u^L\rangle\!\rangle_L=|u\rangle.
 \label{eq:fermion-operator-fourier}
\end{equation}

\paragraph{Qubit realization.}
Encode the occupation basis of $n_L$ modes using $n_L$ qubits.
Jordan--Wigner~\cite{JordanWigner1928,BravyiKitaev2002} represents the Majoranas as
\begin{equation}
 \gamma_{L,2j-1}=Z_1\cdots Z_{j-1}X_j,
 \qquad
 \gamma_{L,2j}=Z_1\cdots Z_{j-1}Y_j.
\end{equation}
Consequently,
\begin{equation}
 \Gamma_L(u)=\varepsilon_L(u)P_{n_L}(T_Lu),
 \qquad \varepsilon_L(u)\in\{\pm1\},
\end{equation}
where $T_L$ is an invertible binary map. Explicitly, if
$T_Lu=(x,z)$, then
\begin{equation}
 x_j=u_{2j-1}+u_{2j},
 \qquad
 z_j=u_{2j}+\sum_{k>j}(u_{2k-1}+u_{2k})
 \pmod2.
\end{equation}
Thus phase-corrected Pauli Bell transform, inverse label conversion,
and the Majorana sign correction implement
\begin{equation}
 |\Gamma_L(u)\rangle\!\rangle_L
 \longmapsto
 \varepsilon_L(u)|T_Lu\rangle
 \longmapsto
 \varepsilon_L(u)|u\rangle
 \longmapsto
 |u\rangle.
\end{equation}
The sign is computed by multiplying the Jordan--Wigner strings in the
fixed order, including the phase in
Eq.~\eqref{eq:families-majorana-convention}. Reversing these operations
gives the inverse target transform.

\paragraph{Postselection onto a fixed particle number.}
The construction above prepares a polynomial guide on the full Fock
space. To obtain a state with particle number $N$, we
measure only the total target particle number
$\widehat N_B=\sum_j f_{B,j}^*f_{B,j}$ and retain outcome $N$.
Let $P_{B,N}$ denote the projector onto this sector. For $n_B$
fermionic modes,
\begin{equation}
 \operatorname{rank}P_{B,N}=\binom{n_B}{N},
 \qquad
 \dim\mathcal H_B=2^{n_B}.
\end{equation}
In particular, at half filling, with even $n_B$ and $N=n_B/2$,
\begin{equation}
 \frac{\operatorname{rank}P_{B,n_B/2}}{\dim\mathcal H_B}
 =
 \frac{\binom{n_B}{n_B/2}}{2^{n_B}}
 \sim\sqrt{\frac{2}{\pi n_B}}.
\end{equation}
Thus the half-filled sector occupies an inverse-polynomial fraction
of the full Fock space. Half filling is particularly important for the repulsive spinful
single-band Hubbard model, where it corresponds to one fermion per
site on average and provides a central setting for studying Mott
insulating behavior and antiferromagnetic correlations~\cite{ArovasEtAl2022}. Half filling also occurs in electronic-structure
calculations in quantum chemistry~\cite{MottaEtAl2017}.

The above dimension ratio is the postselection
probability for the unfiltered guide $p=1$, whose physical marginal
is maximally mixed; a nonconstant filter changes the sector weights.

For a number-conserving target Hamiltonian,
$[H_B,\widehat N_B]=0$, the actual success probability for the
polynomial guide is
\begin{equation}
 s_N(p)
 =
 \frac{\operatorname{Tr}\!\left(P_{B,N}|p(H_B)|^2\right)}
      {\operatorname{Tr}|p(H_B)|^2}.
\end{equation}
Provided $s_N(p)>0$, the conditional physical state is
\begin{equation}
 \rho_{B,N}^{p}
 =
 \frac{P_{B,N}|p(H_B)|^2}
      {\operatorname{Tr}\!\left(P_{B,N}|p(H_B)|^2\right)}.
\end{equation}
Measuring only the total number preserves coherence within the selected
sector. Repeating the preparation until acceptance requires
$1/s_N(p)$ attempts on average.

\subsubsection{Illustrative examples for Fermions}
\label{subsubsec:fermion-examples}

We give two examples with nontrivial source relations. The first has a quadratic, free-Fermion~\cite{SuraceTagliacozzo2022} source, an interacting noncommuting target, and a relative distance
larger than the ordinary distance. The second starts from an
interacting density chain and imposes an extensive number of code
constraints. Source preparation follows
Appendix~\ref{subsec:source-block-guides}: the first example directly
uses independent blocks, while the second extends the same polynomial
accumulation tensors by one binary memory index.
The examples here only aim to show properties of the Fermion transduction and we don't claim advantages over classical and other quantum methods.

\paragraph{Example 1: a quadratic source with a larger relative distance.}
Take $m\ge3$ source blocks of two modes each, so $n_A=2m$.
The target has $m$ active blocks of two modes each, followed by
$m-1$ control modes, so $n_B=3m-1$. Let $e_r^L$ denote the
standard basis vectors of $\mathbb F_2^{2n_L}$, and write
$e_{j,a}^L:=e_{4(j-1)+a}^L$, $1\le a\le4$, for active block $j$.
Define
\begin{equation}
 \begin{aligned}
 w_j^x&:=e_{j,2}^A+e_{j,3}^A,&
 w_j^y&:=e_{j,1}^A+e_{j,3}^A,&
 w_j^z&:=e_{j,1}^A+e_{j,2}^A,\\
 v_j^x&:=e_{j,2}^B+e_{j,3}^B,&
 v_j^y&:=e_{j,1}^B+e_{j,3}^B,&
 v_j^z&:=e_{j,1}^B+e_{j,2}^B,
 \end{aligned}
\end{equation}
and the source block-parity label
\begin{equation}
 w_{A,j}:=e_{j,1}^A+e_{j,2}^A+e_{j,3}^A+e_{j,4}^A.
\end{equation}
With the Hermitian Majorana convention,
\begin{equation}
 \Gamma_A(w_{A,j})
 =-\gamma_{j1}\gamma_{j2}\gamma_{j3}\gamma_{j4},
 \qquad
 \Gamma_A(w_{A,j})\Gamma_A(w_j^\mu)
 =\epsilon_\mu\Gamma_A(w_{A,j}+w_j^\mu),
\end{equation}
where $\epsilon_x=\epsilon_z=1$ and $\epsilon_y=-1$.
The three operators $\Gamma_A(w_j^\mu)$ satisfy the Pauli
multiplication rules and commute with $\Gamma_A(w_{A,j})$.
Choose the six signed bilinear generators in each block as
\begin{equation}
 a_{j,0}^\mu:=\Gamma_A(w_j^\mu),
 \qquad
 a_{j,1}^\mu:=\epsilon_\mu\Gamma_A(w_{A,j}+w_j^\mu),
 \qquad \mu\in\{x,y,z\}.
\end{equation}
Both labels have weight two. For positive coefficients
$h_{j\mu},g_{j\mu}$, the source is therefore quadratic:
\begin{equation}
 H_A=\sum_{j=1}^m\sum_{\mu=x,y,z}
 \left[
 h_{j\mu}\Gamma_A(w_j^\mu)
 +g_{j\mu}\epsilon_\mu
       \Gamma_A(w_{A,j}+w_j^\mu)
 \right].
 \label{eq:fermion-example-quadratic-source}
\end{equation}
Nevertheless, it is not relation free, since
\begin{equation}
 \Gamma_A(w_j^x)\Gamma_A(w_j^y)\Gamma_A(w_j^z)=iI_A.
\end{equation}
The six labels in each block have binary rank three, so the source
label matrix has rank $3m$, with $s=6m$ and $\dim K_A=3m$.

\emph{Source preparation.}
Write $H_A=\sum_jK_j$, where $K_j$ is the $j$th block sum in
Eq.~\eqref{eq:fermion-example-quadratic-source}. Each $K_j$ is an
explicit $4\times4$ Hermitian matrix, and distinct blocks act on
disjoint tensor factors. Thus
Proposition~\ref{prop:source-block-mps} in
Appendix~\ref{subsec:source-block-guides} applies directly: every
given polynomial $p$ of degree at most $D$ with $p(H_A)\ne0$
has a normalized operator state
\begin{equation}
 |\psi_A^p\rangle
 =\frac{|p(H_A)\rangle\!\rangle_A}{\|p(H_A)\|_{2,A}}
\end{equation}
with an explicit matrix product representation of bond dimension at
most $D+1$.

\emph{Target generators and homomorphism.}
For control mode $a$, let
\begin{equation}
 d_a^B:=e_{4m+2a-1}^B+e_{4m+2a}^B,
 \qquad 1\le a\le m-1.
\end{equation}
Then $\Gamma_B(d_a^B)=2n_a^{\mathrm c}-I_B$, where
$n_a^{\mathrm c}$ is its occupation operator. Let
$R\in\mathbb F_2^{(m-1)\times m}$ satisfy
$(Rz)_a=z_a+z_{a+1}$, and set
\begin{equation}
 w_{B,j}:=\sum_{a=1}^{m-1}R_{aj}d_a^B,
 \qquad
 \kappa_j:=(-1)^{|Re_j|}.
\end{equation}
Here $e_j\in\mathbb F_2^m$ and $|Re_j|$ is its column's Hamming
weight. Thus $\kappa_1=\kappa_m=-1$ and $\kappa_j=1$ for
$1<j<m$. The signs retain the occupation-parity convention:
\begin{equation}
 \kappa_j\Gamma_B(w_{B,j})
 =\prod_{a=1}^{m-1}(I_B-2n_a^{\mathrm c})^{R_{aj}}.
\end{equation}
These operators commute with the active blocks and their product
over $j$ is $I_B$. Choose target marked generators
\begin{equation}
 b_{j,0}^\mu:=\Gamma_B(v_j^\mu),
 \qquad
 b_{j,1}^\mu:=\kappa_j\Gamma_B(w_{B,j}+v_j^\mu),
\end{equation}
and define the assignment directly on the Majorana basis by
\begin{equation}
 \begin{aligned}
 \pi\bigl(\Gamma_A(w_j^\mu)\bigr)
   &:=\Gamma_B(v_j^\mu),\\
 \pi\bigl(\Gamma_A(w_{A,j}+w_j^\mu)\bigr)
   &:=\epsilon_\mu\kappa_j
      \Gamma_B(w_{B,j}+v_j^\mu),\\
 H_B&:=\sum_{j=1}^m\sum_{\mu=x,y,z}
 \left[
 h_{j\mu}\Gamma_B(v_j^\mu)
 +g_{j\mu}\kappa_j\Gamma_B(w_{B,j}+v_j^\mu)
 \right].
 \end{aligned}
 \label{eq:fermion-example-quadratic-target}
\end{equation}
In particular, $\pi(a_{j,t}^\mu)=b_{j,t}^\mu$ for $t=0,1$,
and multiplication gives
\begin{equation}
 \pi\bigl(\Gamma_A(w_{A,j})\bigr)
 =\kappa_j\Gamma_B(w_{B,j}).
\end{equation}
Each source block generates $\mathbb C^2\otimes M_2(\mathbb C)$,
with $\Gamma_A(w_{A,j})$ central. The assignment preserves the local
multiplication rules and imposes the single additional central
relation
\begin{equation}
 \pi\!\left(\prod_j\Gamma_A(w_{A,j})\right)=I_B.
\end{equation}
Hence it extends to a surjective unital $*$-homomorphism with
$\pi(H_A)=H_B$. The control and active labels are disjoint.
Since $|w_{B,j}|$ is two at the endpoints and four in the interior,
the second family of target terms is quartic or sextic in Majoranas.
The target has noncommuting marked terms. Both Hamiltonians preserve
parity, but need not conserve particle number.

\emph{Relative distance and decoding.}
Set $w_j^0=v_j^0=0$ and $\epsilon_0=1$. The source Majorana basis
can be indexed by
\begin{equation}
 u(z,\boldsymbol\mu)
 :=\sum_{j=1}^m(z_jw_{A,j}+w_j^{\mu_j}),
 \qquad z\in\mathbb F_2^m,\quad \mu_j\in\{0,x,y,z\}.
\end{equation}
For target coordinates $(y,\boldsymbol\mu)$, write
\begin{equation}
 v(y,\boldsymbol\mu)
 :=\sum_{a=1}^{m-1}y_ad_a^B+\sum_{j=1}^mv_j^{\mu_j}.
\end{equation}
The preceding generator images give the basis map
\begin{equation}
 \pi\bigl(\Gamma_A(u(z,\boldsymbol\mu))\bigr)
 =\chi(z,\boldsymbol\mu)
       \Gamma_B(v(Rz,\boldsymbol\mu)),
 \qquad
 \chi(z,\boldsymbol\mu)
 :=\prod_{j=1}^m(\epsilon_{\mu_j}\kappa_j)^{z_j}.
\end{equation}
Thus both the forward label map and its sign are explicit. The
minimum marked cost of the source label is
\begin{equation}
 w_A(u(z,\boldsymbol\mu))
 =\#\{j:\mu_j\ne0\}
 +2\#\{j:\mu_j=0,\ z_j=1\}.
 \label{eq:fermion-example-quadratic-word-cost}
\end{equation}
A nonidentity active letter costs one, whether or not its label
contains $w_{A,j}$, because both bilinear families are marked.
A pure block parity costs two, since
$a_{j,0}^\mu a_{j,1}^\mu=\Gamma_A(w_{A,j})$.
Disjoint blocks make these local costs additive.

Since $\ker R=\{0,\mathbf1\}$, the only nonzero source basis label
mapped to the zero target label is $\sum_jw_{A,j}$, of cost $2m$.
The target retains the local three-term scalar relations, and its
$6m$ marked labels are distinct and nonzero. Therefore
\begin{equation}
 K_B/K_A\cong\mathbb F_2,
 \qquad
 d_{\mathrm{rel}}=2m,
 \qquad
 d_{\mathrm{ord}}=3.
 \label{eq:fermion-example-quadratic-distances}
\end{equation}
Given a target label $v(y,\boldsymbol\mu)$, extract
$(y,\boldsymbol\mu)$ and solve $Rz_0=y$ by successive binary
additions. The two source candidates have coordinates
$(z_0,\boldsymbol\mu)$ and $(z_0+\mathbf1,\boldsymbol\mu)$, with
\begin{equation}
 w_A(u(z_0,\boldsymbol\mu))
 +w_A(u(z_0+\mathbf1,\boldsymbol\mu))=2m.
\end{equation}
On the degree-$D$ promise with $D<m$, exactly one candidate has cost
at most $D$. This gives an $O(m)$-operation relative decoder;
coordinate conversion and the sign $\chi$ are also efficient.
Together with the Majorana operator--Fourier transforms, the circuit
therefore transduces every nonzero polynomial state with
$D\le m-1$, whereas ordinary unique decoding on the same marking
is guaranteed only for $D\le1$.

\paragraph{Example 2: an interacting source with extensive code constraints.}
\label{subsec:code-constrained-interacting-targets}
For $n_L$ modes, define the pair-label embedding
\begin{equation}
 d_{L,j}:=e_{2j-1}^L+e_{2j}^L,
 \qquad
 \iota_L u:=\sum_{j=1}^{n_L}u_jd_{L,j},
 \qquad u\in\mathbb F_2^{n_L}.
\end{equation}
The Hermitian convention implies
\begin{equation}
 \Gamma_L(\iota_Lu)
 =\prod_j(2f_{L,j}^*f_{L,j}-I_L)^{u_j}
 =(-1)^{|u|}\prod_j(I_L-2f_{L,j}^*f_{L,j})^{u_j}.
\end{equation}
In particular, occupation parity is $-\Gamma_L(d_{L,j})$.
On $m$ source modes, choose the marked generators
\begin{equation}
 a_j:=-\Gamma_A(d_{A,j}),
 \qquad
 a_{m+j}:=\Gamma_A(d_{A,j}+d_{A,j+1}),
\end{equation}
where the first family has $1\le j\le m$ and the second has
$1\le j<m$. With nonzero real coefficients, take
\begin{equation}
 H_A=-\sum_{j=1}^m h_j\Gamma_A(d_{A,j})
     +\sum_{j=1}^{m-1}J_j
                \Gamma_A(d_{A,j}+d_{A,j+1}).
 \label{eq:code-chain-source}
\end{equation}
This is the same number-conserving density chain written in Majorana
coordinates: its marked terms are bilinear and quartic, respectively.
The relations $a_ja_{j+1}a_{m+j}=I_A$ make the source relation rich.

\emph{Source preparation.}
In occupation coordinates the energy is
\begin{equation}
 E_A(x)=\sum_jh_j(-1)^{x_j}
        +\sum_jJ_j(-1)^{x_j+x_{j+1}}.
\end{equation}
Unlike Example 1, neighboring interactions overlap. We use the
polynomial-accumulation construction in the proof of
Proposition~\ref{prop:source-block-mps}, with one additional binary
virtual index for the preceding occupation. Specifically, write
\begin{equation}
 E_A(x)=\sum_{j=1}^m e_j(x_{j-1},x_j),
 \qquad
 e_j(\sigma,x):=h_j(-1)^x+J_{j-1}(-1)^{\sigma+x},
 \qquad J_0:=0,\quad x_0:=0.
\end{equation}
In Eq.~\eqref{eq:source-block-tensors}, replace the block energy by
$e_j(\sigma,x_j)$ and carry $x_j$ to the next tensor. A transition
from polynomial index $a$ to $b\ge a$ contributes
$e_j(\sigma,x_j)^{b-a}/(b-a)!$, with the same polynomial boundary
coefficients. The multinomial identity of
Eq.~\eqref{eq:source-block-tensor-product} then gives amplitudes
$p(E_A(x))$, with bond dimension at most $2(D+1)$.

The norm contractions, canonicalization, and sequential preparation~\cite{Schollwock2011,SchonEtAl2005}
of Appendix~\ref{subsec:source-block-guides} produce
\begin{equation}
 |\psi_A^p\rangle
 =\frac{\sum_xp(E_A(x))|x\rangle|x\rangle}
        {\left(\sum_x|p(E_A(x))|^2\right)^{1/2}}.
\end{equation}

\emph{Target generators and homomorphism.}
Choose a full-row-rank matrix
$R\in\mathbb F_2^{n_B\times m}$, with $n_B<m$, and put
$C:=\ker R$, with minimum nonzero Hamming weight $d(C)$.
Assume $C$ has a polynomial-time syndrome decoder that recovers $u$ from
$Ru$ whenever $|u|\le t_{\mathrm{dec}}$.
For a given shift $\mathbf b\in\mathbb F_2^m$, the Majorana
basis map is
\begin{equation}
 \pi\bigl(\Gamma_A(\iota_Au)\bigr)
 :=(-1)^{|u|+\mathbf b\cdot u+|Ru|}
       \Gamma_B(\iota_BRu).
\end{equation}
The additional signs convert between the Hermitian Majorana
convention and the occupation-parity convention. Setting $r_j:=Re_j$,
the marked images are
\begin{equation}
 \begin{aligned}
 \pi(a_j)
 &=(-1)^{b_j+|r_j|}\Gamma_B(\iota_Br_j),\\
 \pi(a_{m+j})
 &=(-1)^{b_j+b_{j+1}+|r_j+r_{j+1}|}
       \Gamma_B\bigl(\iota_B(r_j+r_{j+1})\bigr).
 \end{aligned}
\end{equation}
The target Hamiltonian is therefore
\begin{equation}
 \begin{aligned}
 H_B={}&\sum_{j=1}^m
 h_j(-1)^{b_j+|r_j|}\Gamma_B(\iota_Br_j)\\
 &+\sum_{j=1}^{m-1}
 J_j(-1)^{b_j+b_{j+1}+|r_j+r_{j+1}|}
       \Gamma_B\bigl(\iota_B(r_j+r_{j+1})\bigr).
 \end{aligned}
 \label{eq:code-chain-target}
\end{equation}
In occupation coordinates this map is the affine restriction
\begin{equation}
 (\pi F)(y)=F(\mathbf b+R^{\mathsf T}y),
 \qquad y\in\mathbb F_2^{n_B}.
 \label{eq:code-chain-restriction}
\end{equation}
It is a surjective unital $*$-homomorphism and satisfies
$\pi(H_A)=H_B$. Both Hamiltonians conserve particle number.
If every column of $R$ has weight at most $w$, the target terms
act on at most $2w$ modes. Constant interaction order requires
this additional sparsity property.

\emph{Relative distance and decoding.}
Let $T\in\mathbb F_2^{m\times(m-1)}$ have columns
$e_j+e_{j+1}$. Viewing the embeddings $\iota_L$ as binary matrices,
the full Majorana label matrices of the marked terms are
\begin{equation}
 A=\iota_A[I_m\ \ T],
 \qquad
 B=\iota_B R[I_m\ \ T].
\end{equation}
Since the embeddings are injective, the induced basis-label map is
$\iota_Au\mapsto\iota_BRu$. Its source word cost satisfies
\begin{equation}
 w_A(\iota_Au)=\min_{x+Te=u}(|x|+|e|).
\end{equation}
Each marked term involves at most two mode pairs, while every
single pair is available. Consequently,
\begin{equation}
 \begin{gathered}
 \left\lceil\frac{|u|}{2}\right\rceil
 \le w_A(\iota_Au)\le|u|,
 \qquad K_B/K_A\cong C,\\
 d_{\mathrm{rel}}
 =\min_{0\ne u\in C}w_A(\iota_Au)
 \ge\left\lceil\frac{d(C)}2\right\rceil.
 \end{gathered}
 \label{eq:code-chain-relative-distance}
\end{equation}
If $d(C)>4$, the marked target labels are nonzero and pairwise
distinct. The inherited three-term scalar relations then give
$d_{\mathrm{ord}}=3$.

For a degree-$D$ guide, $w_A(\iota_Au)\le D$ implies $|u|\le2D$.
Given the target Majorana label $\iota_BRu$, undo the pair embedding,
apply the assumed syndrome decoder, and restore the source embedding:
\begin{equation}
 \iota_BRu\longmapsto Ru\longmapsto u\longmapsto\iota_Au.
\end{equation}
This is efficient on the full degree-$D$ promise whenever
\begin{equation}
 2D\le t_{\mathrm{dec}},\qquad 4D<d(C).
 \label{eq:code-chain-degree-window}
\end{equation}
These conditions imply $2D<d_{\mathrm{rel}}$. The sign
$(-1)^{|u|+\mathbf b\cdot u+|Ru|}$ is also efficiently computable.
The Majorana operator--Fourier transforms and this decoder therefore
prepare
\begin{equation}
 |\psi_B^p\rangle
 \propto\sum_{y\in\mathbb F_2^{n_B}}
 p\!\left(E_A(\mathbf b+R^{\mathsf T}y)\right)|y\rangle|y\rangle.
\end{equation}
For a code family with the assumed decoder and
$\dim C=\Theta(m)$, $d(C)=\Theta(m)$, and
$t_{\mathrm{dec}}=\Theta(m)$, the available polynomial degree is
$D=\Theta(m)$ despite ordinary distance three.

\emph{Comparison with direct source-sector projection.}
The retained source configurations form
$\mathcal S:=\mathbf b+\operatorname{im}R^{\mathsf T}$.
In Majorana coordinates, their projector is
\begin{equation}
 P_{\mathcal S}
 =\frac1{|C|}\sum_{v\in C}
 (-1)^{\mathbf b\cdot v+|v|}\Gamma_A(\iota_Av).
\end{equation}
Each marked source term contains at most two mode pairs, so
$|p(H_A)|^2$ has no Majorana coefficient on $\Gamma_A(\iota_Av)$
with $|v|>4D$. Under Eq.~\eqref{eq:code-chain-degree-window}, only
the zero codeword contributes to its trace against $P_{\mathcal S}$.
The direct-projection probability is therefore
\begin{equation}
 s_{\mathrm{proj}}(p)
 :=\frac{\tau_A(P_{\mathcal S}|p(H_A)|^2)}
          {\tau_A(|p(H_A)|^2)}
 =\frac1{|C|}=2^{-\dim C}.
 \label{eq:code-chain-projection-probability}
\end{equation}
For $\dim C=\Theta(m)$, repeated source preparation and projection
has exponential expected cost, whereas relative decoding prepares
the normalized target guide without this postselection. Note that
projection leaves the state on the source registers.
However, since $R$ has full row rank, the map
$y\mapsto\mathbf b+R^{\mathsf T}y$ is a bijection from
$\mathbb F_2^{n_B}$ onto $\mathcal S$. Hence
$\operatorname{rank}P_{\mathcal S}=2^{n_B}$, equal to the target
Hilbert-space dimension. Binary Gaussian elimination gives an
efficient affine unitary satisfying
\begin{equation}
 U_{\mathrm{aff}}|\mathbf b+R^{\mathsf T}y\rangle
 =
 |y\rangle|0^{m-n_B}\rangle.
\end{equation}
Conditioned on projecting the physical source register onto
$\mathcal S$, applying $U_{\mathrm{aff}}$ to both the physical and
reference registers produces $|\psi_B^p\rangle$ together with
$2(m-n_B)$ zero qubits, because
$E_B(y)=E_A(\mathbf b+R^{\mathsf T}y)$.
Thus the cost is dominated by the probability of accepting the projection.

\paragraph{Scope.}
These examples illustrate the fermionic construction on full-Fock
polynomial states. The first shows
that a free-fermion source can require relative decoding;
the second uses a tractable interacting source to absorb local
relations while leaving an extensive and decodable quotient.
The cost comparison in Example 2 is specific to direct source-sector
projection; a computational advantage over other classical or quantum
preparation methods remains open.
For the number-conserving second target, fixed-number preparation
additionally requires the particle-number postselection discussed
above, with its own success probability. That test is distinct from
the source code-sector projection in
Eq.~\eqref{eq:code-chain-projection-probability}.

\subsection{Bosons and prime-dimensional qudits: Weyl coordinates}
\label{subsec:families-boson}
\label{subsec:families-qudit}

We first encode a fixed-number bosonic problem in finite occupation
registers and expand its Hamiltonian exactly in Weyl operators. We then
formulate source-to-target polynomial transduction for these finite
Weyl algebras (qudits).

\subsubsection{Bosonic Hamiltonians and exact finite Weyl reduction}
\label{subsubsec:boson-formulation}
\label{subsubsec:boson-finite-weyl}

\paragraph{Physical Hamiltonians and fixed particle number.}
For $m$ bosonic modes, let $b_j,b_j^*$ act on
$\mathcal H_\infty=\ell^2(\mathbb N_0)^{\otimes m}$ and satisfy~\cite{GerryKnight2023}
\begin{equation}
 [b_i,b_j^*]=\delta_{ij}I,\qquad [b_i,b_j]=0.
\end{equation}
Consider a number-conserving polynomial Hamiltonian
\begin{equation}
 \begin{aligned}
 H_\infty&=E_cI+\sum_\nu
       \bigl(h_\nu G_\nu+h_\nu^*G_\nu^*\bigr),\\
 G_\nu&:=\prod_{j=1}^m(b_j^*)^{r_{\nu j}}b_j^{s_{\nu j}},
 \qquad
 \sum_jr_{\nu j}=\sum_js_{\nu j}.
 \end{aligned}
 \label{eq:boson-general-interacting-hamiltonian}
\end{equation}
This includes hopping, two-particle interactions, and on-site terms
$n_j(n_j-1)$, where $n_j=b_j^*b_j$~\cite{JakschEtAl1998}. We assume a self-adjoint
realization for which the finite-particle sectors are invariant.
In particular,
\begin{equation}
 [H_\infty,\widehat N]=0,\qquad
 \widehat N:=\sum_j b_j^*b_j.
\end{equation}
For a prescribed total particle number $N$, define
\begin{equation}
 \mathcal H_N:=\operatorname{span}
 \{|\mathbf n\rangle:n_j\ge0,\ \textstyle\sum_jn_j=N\},
 \qquad
 \Pi_N:=\mathbf1_{\{\widehat N=N\}}.
\end{equation}
The sector dimension is
\begin{equation}
 d_N:=\dim\mathcal H_N=\binom{N+m-1}{N}.
 \label{eq:boson-fixed-number-dimension}
\end{equation}
Write $H_N:=H_\infty|_{\mathcal H_N}$ and
$\tau_N:=d_N^{-1}\operatorname{Tr}_{\mathcal H_N}$. Our physical
preparation goal is
\begin{equation}
 |\psi_N^{p_0}\rangle
 :=\frac{(p_0(H_N)\otimes I)|\Phi_N\rangle}
         {\|p_0(H_N)\|_{2,N}},
 \qquad
 |\Phi_N\rangle:=d_N^{-1/2}\sum_{\sum_jn_j=N}
                   |\mathbf n\rangle|\mathbf n\rangle,
\end{equation}
where $\deg p_0\le D$, $p_0(H_N)\ne0$, and
$\|X\|_{2,N}^2:=\tau_N(X^*X)$.
Its physical marginal is
\begin{equation}
 \rho_N^{p_0}
 =\frac{|p_0(H_N)|^2}
       {\operatorname{Tr}_{\mathcal H_N}|p_0(H_N)|^2}.
\end{equation}

Full bosonic Fock space has no normalized maximally entangled
reference, and its canonical polynomial operator vector is generally
not normalizable. Fixing total particle number removes this problem:
$\mathcal H_N$ is finite dimensional.

\paragraph{Finite occupation registers.}
Choose an odd prime $p>N$ and define the occupation isometry
\begin{equation}
 V_p:\mathcal K_p:=(\mathbb C^p)^{\otimes m}
       \longrightarrow\mathcal H_\infty,\qquad
 V_p|\mathbf r\rangle:=|\mathbf r\rangle_{\rm Fock},
 \quad 0\le r_j<p.
\end{equation}
Set
\begin{equation}
 \begin{aligned}
 H^{(p)}&:=V_p^*H_\infty V_p,
 &\widehat n_p&:=\sum_{r=0}^{p-1}r|r\rangle\langle r|,\\
 \widehat N^{(p)}&:=\sum_j\widehat n_{p,j},
 &\Pi_N^{(p)}&:=\mathbf1_{\{\widehat N^{(p)}=N\}}.
 \end{aligned}
\end{equation}
Every configuration with total occupation $N$ has local occupation
at most $N<p$. Therefore
\begin{equation}
 V_p\Pi_N^{(p)}V_p^*=\Pi_N,\qquad
 \operatorname{rank}\Pi_N^{(p)}=d_N,\qquad
 [H^{(p)},\widehat N^{(p)}]=0.
 \label{eq:boson-exact-sector-encoding}
\end{equation}
The total occupation is computed by integer addition.
A local dimension greater than $N$ is sufficient for this exact
sector encoding; primality is used for the finite-field
transduction formulation.

\paragraph{Exact expansion in Weyl operators.}
Let $\omega=e^{2\pi i/p}$ and define the shift and clock matrices
\begin{equation}
 X_p|r\rangle=|r+1\bmod p\rangle,\qquad
 Z_p|r\rangle=\omega^r|r\rangle.
\end{equation}
They satisfy $Z_pX_p=\omega X_pZ_p$.
For $u=(x,z)\in\mathbb F_p^{2m}$, define the Weyl operator~\cite{Gross2006}
\begin{equation}
 W_m(x,z):=\omega^{x\cdot z/2}
              \bigotimes_{j=1}^mX_p^{x_j}Z_p^{z_j},
 \qquad
 [u,v]:=z\cdot x'-x\cdot z',
 \quad v=(x',z').
 \label{eq:families-weyl}
\end{equation}
All exponents are evaluated in $\mathbb F_p$. These matrices form an orthonormal
basis under $p^{-m}\operatorname{Tr}$ and satisfy
$W_m(u)^*=W_m(-u)$.

The compressed ladder matrices have the exact expressions
\begin{equation}
 b_p=\sum_{r=1}^{p-1}\sqrt r\,|r-1\rangle\langle r|
     =X_p^{-1}\sqrt{\widehat n_p},
 \qquad
 b_p^*=\sqrt{\widehat n_p}X_p.
 \label{eq:boson-exact-finite-ladder}
\end{equation}
The diagonal factor gives the amplitudes $\sqrt r$ and removes
cyclic wraparound. Expanding it in clock powers gives
\begin{equation}
 \begin{aligned}
 \sqrt{\widehat n_p}&=\sum_{t=0}^{p-1}s_tZ_p^t,
 &s_t&:=\frac1p\sum_{r=0}^{p-1}\sqrt r\,\omega^{-tr},\\
 b_p&=\sum_{t=0}^{p-1}s_tX_p^{-1}Z_p^t.
 \end{aligned}
 \label{eq:boson-exact-ladder-weyl}
\end{equation}
Thus a compressed ladder operator is an exact linear combination
of Weyl strings.

Equivalently, expand the complete finite Hamiltonian as
\begin{equation}
 H^{(p)}=\sum_{u\in\mathbb F_p^{2m}}h_p(u)W_m(u),
 \qquad
 h_p(u):=p^{-m}\operatorname{Tr}
                  \bigl(W_m(u)^*H^{(p)}\bigr).
 \label{eq:boson-finite-weyl-expansion}
\end{equation}
This is an exact matrix identity on all of $\mathcal K_p$, with
$h_p(-u)=h_p(u)^*$. Its nonzero Weyl terms and their adjoints,
with the identity term separated, determine the finite marking.
A term supported on $k$ modes has at most $p^{2k}$ Weyl coefficients.
Hence a polynomial-size list of bounded-support interactions has
an explicit polynomial-size expansion when $p$ is polynomially
bounded.

The truncated ladder operators obey the modified commutation relation
\begin{equation}
 [b_p,b_p^*]=I_p-p|p-1\rangle\langle p-1|.
\end{equation}
The finite Hamiltonian $H^{(p)}$ agrees with the physical Hamiltonian
on the invariant $N$-particle sector:

\begin{lemma}[Exact finite Weyl realization at fixed particle number]
\label{lem:boson-fixed-number-exact}
For the number-conserving Hamiltonian above and $p>N$,
\begin{equation}
 \left\|(H_\infty V_p-V_pH^{(p)})\Pi_N^{(p)}\right\|=0.
 \label{eq:boson-fixed-number-zero-error}
\end{equation}
For every polynomial $p_0$, of any degree,
\begin{equation}
 p_0(H_\infty)V_p\Pi_N^{(p)}
 =V_pp_0(H^{(p)})\Pi_N^{(p)}.
 \label{eq:boson-fixed-number-polynomial-map}
\end{equation}
Let $\mathcal K_{p,N}:=\Pi_N^{(p)}\mathcal K_p$,
$V_{p,N}:=V_p|_{\mathcal K_{p,N}}$, and
$H_N^{(p)}:=H^{(p)}|_{\mathcal K_{p,N}}$.
Their normalized polynomial operator states agree under
$V_{p,N}\otimes\bigl(V_{p,N}^{\mathsf T}\bigr)^*$, whenever the polynomial
operator is nonzero.
\end{lemma}

\begin{proof}
The sector $\mathcal H_N$ lies inside $\operatorname{ran}V_p$
and is invariant under $H_\infty$. Hence
\[
 V_pH^{(p)}\Pi_N^{(p)}
 =V_pV_p^*H_\infty V_p\Pi_N^{(p)}
 =H_\infty V_p\Pi_N^{(p)}.
\]
Since $H^{(p)}$ also preserves $\mathcal K_{p,N}$, iteration and
linearity prove the polynomial identity. Equivalently,
\[
 H_N=V_{p,N}H_N^{(p)}V_{p,N}^*.
\]
The map $V_{p,N}$ is a unitary identification of the two sector
spaces, preserving spectra, multiplicities, and normalized traces.
Their maximally entangled references also correspond under
$V_{p,N}\otimes\bigl(V_{p,N}^{\mathsf T}\bigr)^*$. This proves the
polynomial-state assertion.
\end{proof}

\paragraph{Full finite transduction followed by number postselection.}
The finite source and target models must satisfy the Weyl
homomorphism and decoding conditions formulated below. Assume that
the resulting circuit prepares the full finite target polynomial
operator state, and put
\begin{equation}
 F_B:=p_0(H_B^{(p)}),\qquad
 \rho_B^{p_0,(p)}:=\frac{|F_B|^2}{\operatorname{Tr}|F_B|^2}.
\end{equation}
Here $H_B^{(p)}$ is the occupation compression of a physical
number-conserving target, with local dimension $p>N$.
Neither the source guide nor this full finite target state is
required to have fixed particle number.

Measure only the total target occupation $\widehat N_B^{(p)}$,
without resolving the individual occupations, and retain outcome
$N$. The probability is
\begin{equation}
 s_N^{(p)}(p_0)
 :=\frac{\operatorname{Tr}
        \bigl(\Pi_{B,N}^{(p)}|p_0(H_B^{(p)})|^2\bigr)}
          {\operatorname{Tr}|p_0(H_B^{(p)})|^2}.
 \label{eq:boson-number-postselection-probability}
\end{equation}
If this probability is positive, the conditional marginal is
\begin{equation}
 \rho_{B,N}^{p_0,(p)}
 =\frac{\Pi_{B,N}^{(p)}|F_B|^2\Pi_{B,N}^{(p)}}
        {\operatorname{Tr}(\Pi_{B,N}^{(p)}|F_B|^2)}.
\end{equation}
The full-space normalization cancels against the postselection
probability. Moreover, number conservation implies
\begin{equation}
 \left.p_0(H_B^{(p)})\right|_{\mathcal K_{p,N}}
 =p_0(H_{B,N}^{(p)}).
\end{equation}
Therefore Lemma~\ref{lem:boson-fixed-number-exact} gives
\begin{equation}
 V_{p,N}\rho_{B,N}^{p_0,(p)}V_{p,N}^*
 =
 \frac{|p_0(H_{B,N})|^2}
      {\operatorname{Tr}_{\mathcal H_{B,N}}|p_0(H_{B,N})|^2}.
\end{equation}
Here $V_{p,N}$ and $\mathcal K_{p,N}$ refer to the target sector.
The corresponding pure operator states agree as well, since
\begin{equation}
 (\Pi_{B,N}^{(p)}\otimes I)|\psi_B^{p_0,(p)}\rangle
 =
 \sqrt{s_N^{(p)}(p_0)}\,|\psi_{B,N}^{p_0,(p)}\rangle.
\end{equation}
Thus full finite preparation followed by postselection gives the
exact encoded fixed-number polynomial guide.

\paragraph{Approximation with a smaller local occupation cutoff.}
Even when $N\ge p$, the local cutoff can retain almost all of the
$N$-photon sector. Let $\Pi_{N,<p}$ project onto configurations
$|\mathbf n\rangle$ satisfying
\begin{equation}
 \sum_{j=1}^m n_j=N,
 \qquad
 0\le n_j<p,
\end{equation}
and define
\begin{equation}
 d_{N,p}:=\operatorname{rank}\Pi_{N,<p},
 \qquad
 d_N:=\operatorname{rank}\Pi_N
     =\binom{N+m-1}{N}.
\end{equation}
Then
\begin{equation}
 \frac{d_{N,p}}{d_N}
 =
 \frac{[z^N](1+z+\cdots+z^{p-1})^m}
      {\binom{N+m-1}{N}}.
 \label{eq:boson-retained-sector-fraction}
\end{equation}
The discarded fraction is the probability, under the uniform
distribution on the $N$-photon sector, that at least one mode contains
$p$ or more photons. A union bound gives
\begin{equation}
 1-\frac{d_{N,p}}{d_N}
 \le
 m\frac{\binom{N-p+m-1}{m-1}}
        {\binom{N+m-1}{m-1}}
 \le
 m\left(\frac{N}{N+m-1}\right)^p.
 \label{eq:boson-retained-sector-bound}
\end{equation}
Therefore, for any constant $c>0$, choosing
\begin{equation}
 p\ge
 \frac{(c+1)\ln m}
      {\ln\!\left(1+\frac{m-1}{N}\right)}
 \label{eq:boson-cutoff-logarithmic}
\end{equation}
is sufficient to guarantee
\begin{equation}
 \frac{d_{N,p}}{d_N}\ge 1-m^{-c}.
 \label{eq:boson-retained-sector-polynomial}
\end{equation}
Hence the cutoff retains all but an inverse-polynomial fraction of the
$N$-photon sector. In particular, when $N=\Theta(m)$, the denominator
in Eq.~\eqref{eq:boson-cutoff-logarithmic} is a positive constant, so
a local cutoff $p=O(\log m)$ already gives
\begin{equation}
 \frac{d_{N,p}}{d_N}
 =1-\frac{1}{\operatorname{poly}(m)},
\end{equation}
even though $p\ll N$.

\subsubsection{Qudit transduction: Weyl generators, relative codes, and Fourier coordinates}
\label{subsubsec:qudit-formulation}

The finite Weyl formulation applies both to the encoded bosonic
models above and to Hamiltonians defined directly on prime-dimensional
qudits. We use an odd prime $p$ and the convention in
Eq.~\eqref{eq:families-weyl}~\cite{HostensEtAl2005,Gross2006}; binary systems use the Pauli
formulation. For $L\in\{A,B\}$, let
\begin{equation}
 \mathcal H_{L,p}:=(\mathbb C^p)^{\otimes m_L},
 \qquad
 \tau_{L,p}:=p^{-m_L}\operatorname{Tr}.
\end{equation}

\paragraph{Marked Weyl directions and field symbols.}
Choose nonzero Weyl directions
\begin{equation}
 w_i\in\mathbb F_p^{2m_A},
 \qquad
 v_i\in\mathbb F_p^{2m_B},
 \qquad i\in[s],
\end{equation}
and phases $r_i^A,r_i^B\in\mathbb F_p$. For every
$\xi\in\mathbb F_p$, define
\begin{equation}
 a_i(\xi):=\omega^{\xi r_i^A}W_{m_A}(\xi w_i),
 \qquad
 b_i(\xi):=\omega^{\xi r_i^B}W_{m_B}(\xi v_i).
 \label{eq:qudit-local-symbols}
\end{equation}
Thus $a_i(0)=I_A$, $b_i(0)=I_B$, and
\begin{equation}
 a_i(\xi)a_i(\eta)=a_i(\xi+\eta),
 \qquad
 b_i(\xi)b_i(\eta)=b_i(\xi+\eta),
 \qquad
 a_i(\xi)^*=a_i(-\xi),
 \qquad
 b_i(\xi)^*=b_i(-\xi).
\end{equation}
The powers along one Weyl direction are therefore treated as the
symbols of a single $\mathbb F_p$-valued coordinate, rather than as
independent generator coordinates.

For each $i$, choose a nonempty symmetric symbol set
$\Lambda_i=-\Lambda_i\subseteq\mathbb F_p^\times$. Each
$a_i(\xi)$ and $b_i(\xi)$ with $\xi\in\Lambda_i$ is a
length-one letter. The induced local word metric is
\begin{equation}
 \rho_i(\xi):=
 \min\left\{t:\xi=\xi_1+\cdots+\xi_t,\
                 \xi_1,\ldots,\xi_t\in\Lambda_i\right\},
 \qquad \rho_i(0):=0,
 \label{eq:qudit-local-symbol-cost}
\end{equation}
where addition is in $\mathbb F_p$. A natural choice is
\begin{equation}
 \Lambda_i=\mathbb F_p^\times,
 \qquad
 \rho_i(\xi)=\mathbf1_{\xi\ne0},
\end{equation}
which is used below in Section~\ref{subsubsec:qudit-opi}.

Consider matched Hermitian Hamiltonians
\begin{equation}
 \begin{aligned}
 H_A^{(p)}&=E_cI_A+
   \sum_{i=1}^s\sum_{\xi\in\Lambda_i}
      \lambda_{i,\xi}a_i(\xi),\\
 H_B^{(p)}&=E_cI_B+
   \sum_{i=1}^s\sum_{\xi\in\Lambda_i}
      \lambda_{i,\xi}b_i(\xi),
 \end{aligned}
 \qquad
 \lambda_{i,-\xi}=\lambda_{i,\xi}^*,
 \label{eq:qudit-matched-hamiltonians}
\end{equation}
where $E_c\in\mathbb R$. Both Hamiltonians have marked degree at
most one. Put
\begin{equation}
 \mathcal A:=C^*(I_A,a_i(\xi):i\in[s],\xi\in\mathbb F_p),
 \qquad
 \mathcal B:=C^*(I_B,b_i(\xi):i\in[s],\xi\in\mathbb F_p).
\end{equation}
The formulation does not assume commutativity between different
Weyl directions.

\paragraph{Commutation phases and relation codes.}
Define
\begin{equation}
 A:=[w_1\ \cdots\ w_s],
 \qquad
 B:=[v_1\ \cdots\ v_s],
 \qquad
 K_A:=\ker A,
 \qquad
 K_B:=\ker B,
 \label{eq:qudit-label-matrices}
\end{equation}
with all arithmetic over $\mathbb F_p$. For
$c=(c_1,\ldots,c_s)\in\mathbb F_p^s$, let
\begin{equation}
 W_A(c):=a_1(c_1)\cdots a_s(c_s),
 \qquad
 W_B(c):=b_1(c_1)\cdots b_s(c_s).
\end{equation}
A reduced symbol vector $c$ belongs to $K_L$ exactly when
$W_L(c)$ is scalar; write
\begin{equation}
 W_L(c)=\zeta_L(c)I_L
 \qquad(c\in K_L).
\end{equation}
Thus $K_L$ records scalar relations remaining after the local
field-symbol laws have already been incorporated.

The Weyl multiplication rules are
\begin{equation}
 W_m(u)W_m(v)=\omega^{[u,v]/2}W_m(u+v),
 \qquad
 W_m(u)W_m(v)=\omega^{[u,v]}W_m(v)W_m(u),
 \label{eq:qudit-weyl-products}
\end{equation}
where
\begin{equation}
 \Omega_m:=\begin{pmatrix}0&-I_m\\ I_m&0\end{pmatrix},
 \qquad
 [u,v]:=u^{\mathsf T}\Omega_m v.
\end{equation}
Hence
\begin{equation}
 a_i(\xi)a_j(\eta)
 =\omega^{\xi\eta\Gamma^A_{ij}}a_j(\eta)a_i(\xi),
 \qquad
 \Gamma_A:=A^{\mathsf T}\Omega_{m_A}A,
 \label{eq:qudit-symbol-commutation}
\end{equation}
with the analogous formula for
$\Gamma_B:=B^{\mathsf T}\Omega_{m_B}B$. The matrices
$\Gamma_A$ and $\Gamma_B$ may be nonzero.

\begin{theorem}[Homomorphisms between Weyl-generated algebras]
\label{thm:qudit-homomorphism}
The assignment
\begin{equation}
 a_i(\xi)\longmapsto b_i(\xi),
 \qquad i\in[s],\quad \xi\in\mathbb F_p,
\end{equation}
extends to a surjective unital $*$-homomorphism
$\pi_p:\mathcal A\to\mathcal B$ if and only if
\begin{equation}
 A^{\mathsf T}\Omega_{m_A}A
 =B^{\mathsf T}\Omega_{m_B}B,
 \qquad
 K_A\subseteq K_B,
 \qquad
 \zeta_A(c)=\zeta_B(c)\quad(c\in K_A).
 \label{eq:qudit-homomorphism-conditions}
\end{equation}
For the matched Hamiltonians in
Eq.~\eqref{eq:qudit-matched-hamiltonians}, this gives
\begin{equation}
 \pi_p(H_A^{(p)})=H_B^{(p)},
 \qquad
 \pi_p\bigl(p_0(H_A^{(p)})\bigr)
 =p_0(H_B^{(p)})
\end{equation}
for every polynomial $p_0$.
\end{theorem}
\noindent
The theorem is proved in Appendix~\ref{app:homomorphism-criteria}.

\paragraph{Marked metric, relative distance, and decoding.}
For $c\in\mathbb F_p^s$, define
\begin{equation}
 \ell_\rho(c):=\sum_{i=1}^s\rho_i(c_i).
 \label{eq:qudit-symbol-weight}
\end{equation}
Under the homomorphism conditions, the forward basis-label map is
\begin{equation}
 M:\operatorname{im}A\longrightarrow\operatorname{im}B,
 \qquad
 M(Ac):=Bc.
\end{equation}
It is well defined because $K_A\subseteq K_B$, and
$\ker M\cong K_B/K_A$. For $u\in\operatorname{im}A$, put
\begin{equation}
 \begin{aligned}
 w_A(u)&:=\min_{Ac=u}\ell_\rho(c),
 &S_D&:=\{u\in\operatorname{im}A:w_A(u)\le D\},\\
 d_{\mathrm{rel}}
 &:=\min_{c\in K_B\setminus K_A}\ell_\rho(c).
 \end{aligned}
 \label{eq:qudit-distance}
\end{equation}
The empty minimum is $+\infty$.

The required relative decoder performs
\begin{equation}
 \operatorname{Dec}_{A\leftarrow B}:
 |Bc\rangle|0\rangle\longmapsto|Bc\rangle|Ac\rangle,
 \qquad \ell_\rho(c)\le D.
 \label{eq:qudit-relative-decoder}
\end{equation}
The metrics $\rho_i$ are symmetric and subadditive, so if
$Bc=Bc'$ and both inputs satisfy the promise, then
\begin{equation}
 c-c'\in K_B,
 \qquad
 \ell_\rho(c-c')
 \le\ell_\rho(c)+\ell_\rho(c')
 \le2D.
\end{equation}
For $2D<d_{\mathrm{rel}}$, this forces $c-c'\in K_A$ and hence
$Ac=Ac'$. For a relation-free source, $K_A=\{0\}$, relative
decoding reduces to ordinary recovery of the reduced field-symbol
vector $c$. 

\paragraph{Operator basis and phase transfer.}
Use the Weyl bases
\begin{equation}
 E_u^A:=W_{m_A}(u)
 \quad(u\in\operatorname{im}A),
 \qquad
 E_v^B:=W_{m_B}(v)
 \quad(v\in\operatorname{im}B).
\end{equation}
They are orthonormal under the normalized full traces, and word
reduction gives
\begin{equation}
 \mathcal F_D^A
 =\operatorname{span}\{E_u^A:u\in S_D\}.
\end{equation}
Write
\begin{equation}
 W_A(c)=\nu_A(c)W_{m_A}(Ac),
 \qquad
 W_B(c)=\nu_B(c)W_{m_B}(Bc).
\end{equation}
For example,
\begin{equation}
 \nu_A(c)=
 \omega^{\sum_i r_i^Ac_i+
       \frac12\sum_{i<j}c_ic_j[w_i,w_j]},
\end{equation}
with the analogous target expression. The homomorphism gives the
global single-image basis map
\begin{equation}
 \pi_p(E_u^A)=\chi(u)E_{M(u)}^B,
 \qquad
 \chi(Ac):=\frac{\nu_B(c)}{\nu_A(c)}
 =\omega^{\sum_i(r_i^B-r_i^A)c_i}.
 \label{eq:qudit-basis-map}
\end{equation}
Agreement of scalar relations makes $\chi(Ac)$ independent of the
chosen solution of $Ac=u$. Both $M(u)$ and $\chi(u)$ can therefore
be computed by solving a linear system over $\mathbb F_p$; no
minimum-weight solution is required for this forward computation.

Let
\begin{equation}
 |\Phi_{L,p}\rangle
 :=p^{-m_L/2}\sum_{\mathbf r\in\mathbb F_p^{m_L}}
       |\mathbf r\rangle|\mathbf r\rangle,
 \qquad
 |X\rangle\!\rangle_{L,p}:=(X\otimes I)|\Phi_{L,p}\rangle,
\end{equation}
then
\begin{equation}
 \mathsf F_L|E_u^L\rangle\!\rangle_{L,p}=|u\rangle
\end{equation}
is the operator--Fourier transform. 

\paragraph{Generalized Bell transform and qubit implementation.}
For one physical--reference qudit pair,
\begin{equation}
 |W_1(x,z)\rangle\!\rangle
 =\frac{\omega^{xz/2}}{\sqrt p}
       \sum_{r\in\mathbb F_p}\omega^{zr}|r+x\rangle|r\rangle.
\end{equation}
Define modular subtraction and the additive Fourier transform by
\begin{equation}
 \mathrm{SUB}|a\rangle|b\rangle:=|a-b\rangle|b\rangle,
 \qquad
 \mathcal Q_p|r\rangle
 :=p^{-1/2}\sum_{t\in\mathbb F_p}\omega^{rt}|t\rangle.
\end{equation}
Then
\begin{equation}
 \begin{aligned}
 |W_1(x,z)\rangle\!\rangle
 &\xrightarrow{\mathrm{SUB}}
   \omega^{xz/2}|x\rangle
      \left(p^{-1/2}\sum_r\omega^{zr}|r\rangle\right)\\
 &\xrightarrow{I\otimes\mathcal Q_p^*}
   \omega^{xz/2}|x\rangle|z\rangle
 \xrightarrow{\text{phase }\omega^{-xz/2}}
   |x\rangle|z\rangle.
 \end{aligned}
 \label{eq:qudit-operator-fourier-circuit}
\end{equation}
Applying this transform to each physical--reference pair gives
$\mathsf F_L$, up to a fixed permutation of the label registers.
On a qubit computer, each field symbol uses
$\lceil\log_2p\rceil$ qubits, and all arithmetic is performed
modulo $p$.

\subsubsection{Prime-field optimal polynomial intersection as a relation-free special case}
\label{subsubsec:qudit-opi}

We now recover prime-field optimal polynomial intersection (OPI)~\cite{JordanEtAlDQI}
as a commuting, relation-free instance of the field-symbol Weyl
formulation. Let $p$ be an odd prime, fix $1\le k<m\le p$, and
choose distinct points
$\alpha_1,\ldots,\alpha_m\in\mathbb F_p$. For
$a=(a_0,\ldots,a_{k-1})\in\mathbb F_p^k$, define
\begin{equation}
 g_a(t):=\sum_{\ell=0}^{k-1}a_\ell t^\ell,
\end{equation}
and let
\begin{equation}
 V\in\mathbb F_p^{m\times k},
 \qquad
 V_{j,\ell}:=\alpha_j^\ell,
 \qquad
 1\le j\le m,
 \quad 0\le\ell<k,
 \label{eq:opi-vandermonde-map}
\end{equation}
so that $(Va)_j=g_a(\alpha_j)$.

For each $j$, let $S_j\subset\mathbb F_p$ be nonempty and proper,
and define
\begin{equation}
 f_j(t):=2\mathbf1_{t\in S_j}-1.
\end{equation}
The independent-variable source score and the OPI target score are
\begin{equation}
 H_A(x):=\sum_{j=1}^m f_j(x_j),
 \qquad
 H_B(a):=\sum_{j=1}^m f_j((Va)_j),
 \label{eq:opi-source-target-scores}
\end{equation}
where $x\in\mathbb F_p^m$ and $a\in\mathbb F_p^k$.
Let
\begin{equation}
 \mathcal A=C(\mathbb F_p^m),
 \qquad
 \mathcal B=C(\mathbb F_p^k)
\end{equation}
with normalized uniform traces. Restriction to the polynomial
evaluation set defines
\begin{equation}
 (\pi F)(a):=F(Va).
 \label{eq:opi-restriction-map}
\end{equation}
Since $V$ has rank $k$, this is a surjective unital
$*$-homomorphism and $\pi(H_A)=H_B$.

Use the diagonal Weyl bases
\begin{equation}
 E_u^A:=W_m(0,u),
 \qquad
 E_v^B:=W_k(0,v).
 \label{eq:opi-diagonal-weyl-bases}
\end{equation}
For $u\in\mathbb F_p^m$,
\begin{equation}
 \pi(E_u^A)=E_{V^{\mathsf T}u}^B.
 \label{eq:opi-weyl-map}
\end{equation}
Equivalently, choose one field-symbol coordinate for each evaluation
point,
\begin{equation}
 a_j(\xi):=W_m(0,\xi e_j)=Z_j^\xi,
 \qquad
 b_j(\xi):=W_k(0,\xi V^{\mathsf T}e_j),
 \qquad
 \xi\in\mathbb F_p.
\end{equation}
The local Fourier expansion
\begin{equation}
 f_j(t)=\sum_{\xi\in\mathbb F_p}
             \widehat f_j(\xi)\omega^{\xi t},
 \qquad
 \widehat f_j(\xi)
 =\frac1p\sum_{t\in\mathbb F_p}f_j(t)\omega^{-\xi t},
\end{equation}
expresses $H_A$ and $H_B$ in the matched form
Eq.~\eqref{eq:qudit-matched-hamiltonians}. For $\xi\ne0$,
\begin{equation}
 \widehat f_j(\xi)
 =\frac2p\sum_{t\in S_j}\omega^{-\xi t}\ne0.
\end{equation}
The nonvanishing follows because a proper subset sum of the $p$th
roots of unity cannot vanish when $p$ is prime. We therefore use
the full symbol alphabet
$\Lambda_j=\mathbb F_p^\times$, for which
\begin{equation}
 \rho_j(\xi)=\mathbf1_{\xi\ne0}.
\end{equation}
Thus the marked weight of $u\in\mathbb F_p^m$ is
$\operatorname{wt}_H(u)$, and both score Hamiltonians have marked
degree at most one.

In the full Weyl phase spaces, the reduced label matrices are
\begin{equation}
 A_{\mathrm{OPI}}
 :=\begin{pmatrix}0\\I_m\end{pmatrix},
 \qquad
 B_{\mathrm{OPI}}
 :=\begin{pmatrix}0\\V^{\mathsf T}\end{pmatrix}.
 \label{eq:opi-weyl-label-matrices}
\end{equation}
Hence
\begin{equation}
 K_A=\{0\},
 \qquad
 K_B=\ker V^{\mathsf T}.
 \label{eq:opi-relation-codes}
\end{equation}
Moreover, all marked directions are
$Z$-type, so both commutation matrices vanish. OPI is therefore a
commuting, relation-free specialization of the general formulation.

The basis-label map is
\begin{equation}
 M(u)=V^{\mathsf T}u.
 \label{eq:opi-basis-label-map}
\end{equation}
Since every $k$ columns of $V^{\mathsf T}$ are independent while
any $k+1$ columns are dependent~\cite{ReedSolomon1960,HuffmanPless2003},
\begin{equation}
 d_{\mathrm{rel}}=d_{\mathrm{ord}}
 =\min_{0\ne u\in\ker V^{\mathsf T}}
       \operatorname{wt}_H(u)
 =k+1.
 \label{eq:opi-relative-distance}
\end{equation}

\subsubsection{An illustrative example for Bosons }
\label{subsubsec:boson-dimer-example}

We construct an explicit bosonic source--target pair with
$d_{\mathrm{ord}}=3$ and $d_{\mathrm{rel}}=2m$.
The source consists of independent four-mode blocks, while the target
contains hopping and density interactions sharing control modes.
The construction uses exact invariant particle-number sectors,
rather than an approximate occupation truncation.

\paragraph{Exact qutrit encoding in a bosonic dimer.}
For two bosonic modes with annihilation operators $\ell,r$, restrict
to the two-particle sector
\begin{equation}
 \mathcal K
 :=\operatorname{span}\{|q,2-q\rangle:q=0,1,2\},
 \qquad
 |q\rangle_{\mathrm{log}}:=|q,2-q\rangle.
 \label{eq:boson-dimer-encoding}
\end{equation}
The occupation imbalance and hopping operators restrict to
\begin{equation}
 Q:=\ell^\dagger\ell-I
 =\begin{pmatrix}-1&0&0\\0&0&0\\0&0&1\end{pmatrix},
 \qquad
 T:=\ell^\dagger r+r^\dagger\ell
 =\sqrt2\begin{pmatrix}0&1&0\\1&0&1\\0&1&0\end{pmatrix}.
 \label{eq:boson-dimer-local-operators}
\end{equation}
Both preserve $\mathcal K$, and $[Q,T]\ne0$.
On this encoded qutrit, define
\begin{equation}
 X|q\rangle_{\mathrm{log}}
 =|q+1\bmod3\rangle_{\mathrm{log}},
 \qquad
 Z|q\rangle_{\mathrm{log}}
 =\omega^q|q\rangle_{\mathrm{log}},
 \qquad \omega=e^{2\pi i/3}.
 \label{eq:boson-dimer-weyl}
\end{equation}
The cyclic shift $X$ is a logical Weyl operator, not a bosonic
creation operator. Physical operators such as $T$ are expanded
in this Weyl basis.

\paragraph{Independent source blocks and the target Hamiltonian.}
Fix $m\ge3$. The source contains $m$ active dimers and $m$ control
dimers, each restricted to $\mathcal K$. Let $Q_j,T_j$ act on
active dimer $j$, and let $\nu_j$ be the left-mode occupation of
source control dimer $j$. Define
\begin{equation}
 H_A=\sum_{j=1}^m h_j,
 \qquad
 h_j:=-T_j+\nu_jQ_j.
 \label{eq:boson-dimer-source}
\end{equation}
Each $h_j$ acts on an independent four-mode block whose selected
sector has dimension nine.

The target retains the $m$ active dimers but uses only $m-1$
control dimers, with left-mode occupations $n_1,\ldots,n_{m-1}$.
On occupations $a,b\in\{0,1,2\}$, define
\begin{equation}
 F(a,b):=[a+b]_3
 =a+b+\frac34ab(3ab-5a-5b+7),
 \label{eq:boson-dimer-modular-density}
\end{equation}
where $[\,\cdot\,]_3$ denotes the representative in $\{0,1,2\}$.
Set
\begin{equation}
 H_B
 =-\sum_{j=1}^mT_j
   +n_1Q_1
   +\sum_{j=2}^{m-1}F(n_{j-1},n_j)Q_j
   +n_{m-1}Q_m.
 \label{eq:boson-dimer-target}
\end{equation}
This is a polynomial bosonic Hamiltonian with hopping and density
interactions. For example,
$[F(n_{j-1},n_j)Q_j,T_j]\ne0$.
Every dimer's particle number is conserved, so the product of the
two-particle sectors is exactly invariant in both systems.
All traces and operator states below refer to these selected sectors.

\paragraph{The homomorphism.}
Write $S_j:=\omega^{\nu_j}$ for the source control clocks and
$C_a:=\omega^{n_a}$ for the target control clocks. The relevant
algebras are
\begin{equation}
 \mathcal A\cong M_3^{\otimes m}\otimes C(\mathbb F_3^m),
 \qquad
 \mathcal B\cong M_3^{\otimes m}\otimes C(\mathbb F_3^{m-1}),
 \label{eq:boson-dimer-algebras}
\end{equation}
with normalized physical traces.
Define $R:\mathbb F_3^m\to\mathbb F_3^{m-1}$ by
\begin{equation}
 (Rs)_a=s_a+s_{a+1},
 \qquad 1\le a<m.
 \label{eq:boson-dimer-control-map}
\end{equation}
Since $R^{\mathsf T}$ is injective, restriction of matrix-valued
functions,
$(\pi G)(t):=G(R^{\mathsf T}t)$,
defines a surjective unital $*$-homomorphism. It fixes the active
operators and maps the control clocks according to
\begin{equation}
 \pi(S_1)=C_1,\qquad
 \pi(S_j)=C_{j-1}C_j\quad(2\le j<m),\qquad
 \pi(S_m)=C_{m-1}.
 \label{eq:boson-dimer-clock-images}
\end{equation}
Consequently,
\begin{equation}
 \pi(H_A)=H_B,
 \qquad
 \pi\bigl(p_0(H_A)\bigr)=p_0(H_B).
 \label{eq:boson-dimer-polynomial-map}
\end{equation}

\paragraph{Weyl marking and distance improvement.}
For each source block, mark the six cyclic Weyl directions
\begin{equation}
 U_{j,t}:=W_j(1,t),
 \qquad
 V_{j,t}:=Z_jS_j^t,
 \qquad t\in\mathbb F_3,
 \label{eq:boson-dimer-marked-directions}
\end{equation}
and mark their images in the target. Each nonzero power of a
direction has length one. A direct Weyl expansion of
Eq.~\eqref{eq:boson-dimer-local-operators} shows that $-T_j$
contains exactly the three $U_{j,t}$ directions and their
adjoints, while $\nu_jQ_j$ contains exactly the three $V_{j,t}$
directions and their adjoints, all with nonzero coefficients.
Thus both Hamiltonians have marked degree one, and no auxiliary
Hamiltonian terms are introduced to obtain the marking.

For $x,z,s\in\mathbb F_3^m$, the source basis operators are
\begin{equation}
 E_{x,z,s}^A
 :=\prod_{j=1}^m W_j(x_j,z_j)S_j^{s_j}.
 \label{eq:boson-dimer-source-basis}
\end{equation}
Their images have basis labels
\begin{equation}
 M(x,z,s)=(x,z,Rs),
 \qquad
 \ker R
 =\operatorname{span}_{\mathbb F_3}
       \{(1,-1,1,-1,\ldots)\}.
 \label{eq:boson-dimer-basis-map}
\end{equation}
The active Weyl operators are unchanged, so their noncommutation
phases are preserved. The homomorphism also preserves every
source scalar relation.

\begin{proposition}[Bosonic distance separation]
\label{prop:boson-dimer-distances}
For the field-symbol marking in
Eq.~\eqref{eq:boson-dimer-marked-directions},
\begin{equation}
 d_{\mathrm{ord}}=3,
 \qquad
 d_{\mathrm{rel}}=2m,
 \qquad
 K_B/K_A\cong\mathbb F_3.
 \label{eq:boson-dimer-distances}
\end{equation}
Consequently, relative decoding certifies polynomial degrees
$D<m$, whereas ordinary unique decoding for the same marking
certifies only $D\le1$.
\end{proposition}

\begin{proof}
The local relation
\[
 V_{j,0}V_{j,1}V_{j,2}=Z_j^3S_j^3=I
\]
is inherited by the target and has symbol weight three.
Every target marked direction is nonzero, and no two distinct
target directions are proportional: their active labels
distinguish the blocks and the two direction types, while
$Re_j\ne0$ distinguishes the three density directions within
a block. Hence there is no scalar relation of weight one or
two, proving $d_{\mathrm{ord}}=3$.

By Eq.~\eqref{eq:boson-dimer-basis-map}, a nonidentity source
basis operator maps to the identity exactly when
\[
 x=z=0,\qquad
 s=\lambda(1,-1,1,-1,\ldots),\qquad
 \lambda\in\mathbb F_3^\times.
\]
A nontrivial pure control clock requires exactly two letters:
\[
 S_j^\xi=V_{j,0}^{-\xi}V_{j,1}^{\xi},
 \qquad \xi\ne0,
\]
and one letter cannot suffice because every marked letter has
a nonzero active label. Source blocks have disjoint support,
so their word costs add. Every new scalar relation therefore
has minimum length $2m$. Finally,
$\ker M\cong\ker R\cong\mathbb F_3$ gives the quotient-code
statement.
\end{proof}

\paragraph{Preparation and explicit relative decoding.}
Each source block is an explicit $9\times9$ matrix.
Proposition~\ref{prop:source-block-mps} therefore gives an MPS
representation of the canonical source polynomial state with
bond dimension at most $D+1$, using the paired-block ordering.
With efficiently specified polynomial coefficients and controlled
numerical precision, these tensors supply source preparation
through standard sequential MPS methods.

Given a target label $(x,z,y)$, the relative decoder considers
the three solutions of $Rs=y$. For each choice
$\lambda\in\mathbb F_3$, obtain one candidate by
\begin{equation}
 s_1=\lambda,\qquad
 s_{j+1}=y_j-s_j\quad(1\le j<m).
 \label{eq:boson-dimer-decoder}
\end{equation}
The local source word costs can be tabulated once on the
$27$ labels $(x_j,z_j,s_j)\in\mathbb F_3^3$.
Summing these costs tests each candidate against the degree
promise. When $D<m$, at most one candidate has total cost at
most $D$, by Proposition~\ref{prop:boson-dimer-distances}.
The decoder thus uses $O(m)$ local weight evaluations and admits
a polynomial-size reversible implementation.

The general Weyl circuit consequently implements
\begin{equation}
 \frac{|p_0(H_A)\rangle\!\rangle_A}
      {\|p_0(H_A)\|_{2,A}}
 \longmapsto
 \frac{|p_0(H_B)\rangle\!\rangle_B}
      {\|p_0(H_B)\|_{2,B}},
 \qquad
 \deg p_0\le D<m,
 \label{eq:boson-dimer-transduction}
\end{equation}
for every prepared nonzero source guide.
Already at $m=3$, this permits quadratic filters while the
ordinary bound permits only linear filters.

The source is tractable because its interacting blocks are
independent, not because it is a free-boson Hamiltonian.
Likewise, the target controls remain conserved, and conditioned
on their occupations the active dimers separate. The example
establishes a larger certified degree window for this fixed
marking, not a classical hardness claim. Direct encoding of
the specified product sectors introduces neither occupation
truncation error nor a particle-number postselection step.

\clearpage
\section{Advantage with relative decoding: Code-constrained optimization with pairwise tests}
\label{sec:applications}
\label{app:chain-example}
 
We now apply the framework to optimization problems in which the
unknowns are the values of a low-degree polynomial and every constraint
tests two neighbouring values. This generalizes the optimal polynomial intersection (OPI) problem
studied in DQI~\cite{JordanEtAlDQI} to tests that couple neighbouring
values. We develop the algorithm for pairwise tests and then specialize
to a nonlinear residual chain. In this example, relative decoding
supports a higher filter degree than ordinary term-selection decoding
with the same Fourier-term marking. Section~\ref{subsec:pairwise-score}
computes the exact quantum mean, and
Section~\ref{subsec:chain-benchmark-summary} compares it with the
scores found by classical heuristics.
 
\subsection{The template and the quantum algorithm}
\label{subsec:pairwise-template}
 We consider the following generalization of OPI:
\paragraph{The problem.}
Let $\F_q$ be a finite field, put $L=q-1$, and let
$\alpha_1,\ldots,\alpha_L$ be its nonzero elements in a fixed order. We
think of these points as the vertices of a path, with \emph{edge} $i$
joining vertices $i$ and $i+1$, so there are $M=L-1$ edges. An instance
assigns to every edge a pass/fail test
$C_i:\F_q\times\F_q\to\{0,1\}$. For a degree cap $2\le n<L$, the task is
\begin{equation}
 \max_{\deg P<n}\operatorname{Sat}(P),
 \qquad
 \operatorname{Sat}(P)=\frac1M\sum_{i=1}^{M}
 C_i\bigl(P(\alpha_i),P(\alpha_{i+1})\bigr).
 \label{eq:numerical-satisfaction}
\end{equation}
 
What makes this interesting is the polynomial constraint. For the tests
below, the problem would be easy if the $L$ values could be chosen
freely: one would satisfy the edges one at a time, from left to right.
A polynomial of degree less than $n$, however, has only $n$ free
coefficients, so fixing any $n$ values determines all the others. 
Optimal polynomial intersection (OPI)~\cite{JordanEtAlDQI,opi-optimized}
is the analogous problem with tests on single values,
$P(\alpha_i)\in F_i$. Here each test couples two neighbouring values.
 
\paragraph{The code, the source, and the homomorphism.}
Let $V\in\F_q^{L\times n}$ with $V_{ij}=\alpha_i^{\,j}$ for
$1\le i\le L$ and $0\le j<n$. The value
vectors of feasible polynomials form the Reed--Solomon code
$C=\operatorname{im}V$. Since any $n$ rows of $V$ form an invertible
Vandermonde matrix, two facts follow:
\begin{enumerate}[leftmargin=*,label=(\roman*),itemsep=0.2em]
 \item if $a\in\F_q^n$ is uniform, then $Va$ is uniform on any $n$
 coordinates;
 \item every nonzero vector in $\ker V^{\mathsf T}$ has at least $n+1$
 nonzero entries.
\end{enumerate}
 
The \emph{source} is the same problem without the polynomial
constraint: assignments $x\in\F_q^L$ with the uniform measure.
Restricting source functions to the code,
\begin{equation}
 (\pi F)(a)=F(Va),\qquad a\in\F_q^n ,
\end{equation}
is a surjective unital $*$-homomorphism of diagonal algebras, and it
maps the source score to the target score. In Fourier coordinates it is
a relabeling. Let the characters $\chi_u(x)=\omega^{\Tr(u\cdot x)}$ for
$u\in\F_q^L$, where $\omega=e^{2\pi i/\ell}$, $\ell$ is the
characteristic of $\F_q$, and $\Tr$ is the trace from $\F_q$ to
$\F_\ell$. Then
\begin{equation}
 \pi(\chi_u)=\chi_{V^{\mathsf T}u}.
\end{equation}
We call $u$ a \emph{vertex label}: it assigns a field element to every
vertex.
 
\paragraph{The guiding state.}
Suppose that under the uniform source the tests are independent, each
passing with probability $\rho\in(0,1)$, as they are for the nonlinear
residual chain below. Let
\begin{equation}
 \tilde f_i=\frac{C_i-\rho}{\sqrt{\rho(1-\rho)}}
\end{equation}
be the centred and normalized tests. The products
$\tilde f_E=\prod_{i\in E}\tilde f_i$ over edge sets $E$ are then
orthonormal, and so are the \emph{shell functions}
\begin{equation}
 \varphi_k=\binom{M}{k}^{-1/2}\sum_{|E|=k}\tilde f_E ,
 \qquad k=0,\ldots,M.
\end{equation}
Each $\varphi_k$ is a polynomial of degree $k$ in the source score
(Section~\ref{subsec:pairwise-score}). A degree-$D$ filter is therefore
a guide
\begin{equation}
 g=\sum_{k=0}^{D}c_k\varphi_k=\sum_{|E|\le D}w_E\,\tilde f_E,
 \qquad
 w_E=\frac{c_{|E|}}{\sqrt{\binom{M}{|E|}}},
\end{equation}
with $\sum_{k=0}^D|c_k|^2=1$, so that $\mathbb E_x|g(x)|^2=1$.
The target state is
$\ket{\psi_B}\propto\sum_{a\in\F_q^n}g(Va)\ket{a}$. Measuring it
returns a feasible polynomial $P_a$ with probability proportional to
$|g(Va)|^2$.
 
\paragraph{The algorithm.}
The four steps of Theorem~\ref{thm:mf-polynomial-transduction} become:
\begin{enumerate}[leftmargin=*,itemsep=0.2em]
 \item \emph{Prepare the Fourier coefficients}
 $\sum_u\widehat g(u)\ket{u}$ of the guide. This is the only step that
 depends on the tests; Lemma~\ref{app:residual-source-lemma}
 supplies it for the nonlinear residual chain.
 \item \emph{Compute the syndrome:} $\ket{u}\mapsto\ket{u}\ket{V^{\mathsf T}u}$.
 \item \emph{Decode and erase:} recover $u$ from $V^{\mathsf T}u$
 (Lemma~\ref{app:chain-decoding-lemma}) and use it to clear the first
 register.
 \item \emph{Inverse Fourier transform} over $\F_q^n$. Since
 $\sum_u\widehat g(u)\chi_{V^{\mathsf T}u}(a)=g(Va)$, this gives
 $\ket{\psi_B}$.
\end{enumerate}
Each test depends on two vertices, so its Fourier labels are nonzero
only on those two vertices. A guide of degree $D$ therefore only has
labels with at most $2D$ nonzero entries, and the next two lemmas hold
for every pairwise test family satisfying the guide assumptions above.
 
\begin{lemma}[Decoding]
\label{app:chain-decoding-lemma}
If $4D\le n$, every vertex label $u$ with $\wt(u)\le2D$ is determined by
its syndrome $V^{\mathsf T}u$, and can be recovered from it by a
polynomial-size reversible circuit.
\end{lemma}
\begin{proof}
Two such labels with the same syndrome differ by a nonzero vector of
$\ker V^{\mathsf T}$ with at most $4D\le n$ nonzero entries, which
contradicts (ii). For recovery, $s_j=\sum_iu_i\alpha_i^{\,j}$ with
$0\le j<n$ is a Reed--Solomon syndrome with $t=\wt(u)\le n/2$ errors.
The error locator $\Lambda(T)=\prod_{i:u_i\ne0}(T-\alpha_i)$ is the
shortest monic recurrence satisfied by the sequence $(s_j)$, namely
$\sum_{\ell=0}^{t}\Lambda_\ell s_{j+\ell}=0$ for $0\le j<n-t$. It is
unique because $n\ge2t$, and can be found by Berlekamp--Massey or by
linear algebra. Its roots among the $\alpha_i$ give the support, and a
Vandermonde solve gives the values. All of this takes polynomially many
field operations and can be made reversible~\cite{Bennett1973}.
\end{proof}
 
\begin{lemma}[Norm and score transfer]
\label{app:chain-transfer-lemma}
If $4D+2\le n$, then, for $a$ uniform on $\F_q^n$ and $x$ uniform on
$\F_q^L$,
\begin{equation}
 \mathbb E_a|g(Va)|^2=\mathbb E_x|g(x)|^2=1,
 \qquad
 \mathbb E_a\bigl[T(Va)|g(Va)|^2\bigr]=\mathbb E_x\bigl[T(x)|g(x)|^2\bigr],
 \label{app:chain-norm-transfer}
\end{equation}
where $T=\sum_iC_i$ is the number of passing edges.
\end{lemma}
\begin{proof}
The function $|g|^2$ is a combination of functions of at most $2D$
tests, and $T|g|^2$ of at most $2D+1$. These depend on at most $4D$,
respectively $4D+2$, vertices. By (i), $Va$ is uniform on such a set,
so each term has the same expectation on the source and on the code.
\end{proof}

\paragraph{Relative decoding and the Fourier-term marking.}
After the Fourier transform, the register holds the vertex label $u$.
Different products of Fourier terms can give the same $u$, and the guide
has already added their amplitudes into $\widehat g(u)$. Relative
decoding recovers this label, which is exactly what erasing the register
requires. Distinct source labels can have the same target syndrome only
if their difference has at least $n+1$ nonzero vertices. Since each
marked Fourier term touches at most two vertices,
$d_{\mathrm{rel}}\ge\lceil(n+1)/2\rceil$, allowing filter degrees
$D\approx n/4$.

\subsection{The nonlinear residual chain}
\label{subsec:residual-chain}
We now specify the nonlinear tests and an efficient guiding-state preparation.
\paragraph{The tests.}
Let $q=p$ be an odd prime, and write $\omega=e^{2\pi i/p}$. For every
edge the instance supplies a nonzero coefficient $a_i\in\F_p$ and a set
$F_i\subset\F_p$ of size $(p-1)/2$. Define
\begin{equation}
 C_i(x_i,x_{i+1})=1\iff r_i:=x_{i+1}+a_ix_i^2\in F_i .
 \label{eq:residual-test}
\end{equation}
This is a pairwise, nonlinear version of OPI: the quantity tested
against a random half-set is the \emph{residual} $r_i$, which mixes two
neighbouring values quadratically. Without the polynomial constraint
the problem is again easy, since one can choose each
$x_{i+1}=r_i-a_ix_i^2$ with $r_i\in F_i$.
 
\paragraph{Short Fourier-term relations.}
For the nonlinear residual chain, the Fourier-term marking also has
short relations already in the source. On any one edge, all characters
$\chi_{(s,t)}$ with $t\ne0$ have nonzero Fourier coefficients: these
coefficients factor into a nonzero quadratic Gauss sum and a nontrivial
Fourier coefficient of the proper, nonempty subset $F_i\subset\F_p$.
The latter is nonzero because the minimal polynomial of a primitive
$p$th root of unity is $1+z+\cdots+z^{p-1}$. In particular,
\begin{equation}
 \chi_{(0,1)}\chi_{(2,1)}=\chi_{(1,1)}^2
\end{equation}
is a length-four relation absent from the reference with independent
Fourier-term generators. Thus $d_{\mathrm{ord}}\le4$, and the ordinary
distance criterion for this marking certifies no filter degree above
$D=1$. Relative decoding retains this source relation while decoding
only the aggregate label $u$.

\paragraph{The source in residual coordinates.}
Define
\begin{equation}
r_i=x_{i+1}+a_i x_i^2,
\qquad
T_{\mathrm{OPI}}(r):=\sum_{i=1}^{M}\mathbf{1}_{F_i}(r_i).
\end{equation}
The map \(x\mapsto(x_1,r_1,\ldots,r_M)\) is bijective, with inverse
\(x_{i+1}=r_i-a_i x_i^2\). In residual coordinates, the source therefore
consists of the same independent single-variable tests as in standard OPI,
together with the free variable \(x_1\). In particular,
\begin{equation}
T(x)=T_{\mathrm{OPI}}(r(x)),
\qquad
g(x)=G(T(x))=g_{\mathrm{OPI}}(r(x)),
\end{equation}
where \(g_{\mathrm{OPI}}(r):=G(T_{\mathrm{OPI}}(r))\) with \(G\) the corresponding univariate polynomial
of degree at most \(D\).
This reduces nonlinear source preparation to standard OPI preparation
followed by a reversible change of coordinates.

\begin{lemma}[Guiding state preparation]
\label{app:residual-source-lemma}
The state
\[
p^{-L/2}\sum_{x\in\mathbb F_p^L}g(x)\lvert x\rangle
\]
can be prepared in polynomial time.
\end{lemma}

\begin{proof}
First apply the standard OPI pilot-preparation routine to the \(M\) tests
\(r_i\in F_i\), using the same filter \(G\), to prepare
\[
\lvert\psi_{\mathrm{OPI}}\rangle
=
p^{-M/2}\sum_{r\in\mathbb F_p^M}
g_{\mathrm{OPI}}(r)\lvert r\rangle.
\]

Next adjoin \(p^{-1/2}\sum_{x_1\in\mathbb F_p}\lvert x_1\rangle\)
and sequentially replace each residual \(r_i\) by
\(x_{i+1}=r_i-a_i x_i^2\). This implements the inverse residual map
coherently and gives
\[
p^{-(M+1)/2}
\sum_{\substack{x_1\in\mathbb F_p\\r\in\mathbb F_p^M}}
g_{\mathrm{OPI}}(r)\lvert x_1,r\rangle
\longmapsto
p^{-L/2}\sum_{x\in\mathbb F_p^L}g(x)\lvert x\rangle,
\]
where \(L=M+1\). Standard OPI preparation is efficient, and the coordinate
conversion requires only \(M\) reversible field updates. The complete
preparation therefore runs in polynomial time.
\end{proof}

\paragraph{The expected quantum score.}
\label{subsec:pairwise-score}
As in DQI~\cite{JordanEtAlDQI}, a tridiagonal matrix determines the
optimal polynomial guide and its expected score.
Multiplying $\tilde f_E$ by $\tilde f_i$ adds edge $i$ if $i\notin E$.
If $i\in E$, it gives $\tilde f_{E\setminus i}+\gamma\tilde f_E$, with
$\gamma=(1-2\rho)/\sqrt{\rho(1-\rho)}$, since
$\tilde f_i^2=1+\gamma\tilde f_i$. Summing over $i$, the number of
passing edges $T=M\rho+\sqrt{\rho(1-\rho)}\sum_i\tilde f_i$ acts on the
shells as
\begin{equation}
 \begin{aligned}
 T\varphi_k={}&\sqrt{\rho(1-\rho)}\Bigl[\sqrt{k(M-k+1)}\,\varphi_{k-1}
 +\sqrt{(k+1)(M-k)}\,\varphi_{k+1}\Bigr]\\
 &+\bigl(M\rho+(1-2\rho)k\bigr)\varphi_k .
 \end{aligned}
 \label{app:chain-recurrence}
\end{equation}
Let $J_{\rho,D}$ be the $(D+1)\times(D+1)$ tridiagonal matrix of this
action, with diagonal $M\rho+(1-2\rho)k$ and off-diagonal
$\sqrt{\rho(1-\rho)(k+1)(M-k)}$. For $\rho=\tfrac12$, it is $M I/2+J_D/2$, where $J_D$ is the
DQI matrix in Eq.~\eqref{eq:value-dqi-jacobi} with $s=M$.
 
\begin{theorem}[Expected quantum score]
\label{app:chain-performance}
For the nonlinear residual chain, let
$D=\lfloor(n-2)/4\rfloor$ and let $c$ be a real unit top eigenvector of
$J_{\rho,D}$. The algorithm of Section~\ref{subsec:pairwise-template} runs in polynomial time and outputs a polynomial $P$ of
degree less than $n$ with
\begin{equation}
 \mathbb E\bigl[\operatorname{Sat}(P)\bigr]=\frac{\lambda_{\max}(J_{\rho,D})}{M},
 \label{app:chain-finite-guarantee}
\end{equation}
the best expected score achievable by any degree-$D$ filter of the
score. As $p\to\infty$ with $n/L\to R\in(0,1)$, we have
$\rho\to\tfrac12$ and $D/M\to R/4$, and the expected score tends to
\begin{equation}
 \frac12+\sqrt{\frac R4\Bigl(1-\frac R4\Bigr)} .
 \label{app:chain-asymptotic}
\end{equation}
\end{theorem}
\begin{proof}
Efficient preparation of the guiding state follows from
Lemma~\ref{app:residual-source-lemma}. Steps
2--4 are correct by Lemmas~\ref{app:chain-decoding-lemma}
and~\ref{app:chain-transfer-lemma}, since $4D+2\le n$. The expected
score is
$\mathbb E_a[T(Va)|g(Va)|^2]/M=\mathbb E_x[T|g|^2]/M=c^{\mathsf T}J_{\rho,D}c/M$,
and the top eigenvector maximizes this. For the limit, write
$\delta=R/4$. As $\rho\to1/2$, the diagonal entries of
$J_{\rho,D}/M$ tend uniformly to $1/2$, and its off-diagonal entries
at $k/M=x$ tend to $\sqrt{x(1-x)}/2$. Since this last expression is
increasing for $0\le x\le\delta<1/4$, the maximum row sum tends to
$1/2+\sqrt{\delta(1-\delta)}$. A trial vector supported on a growing
block of $o(M)$ consecutive indices just below $D$ gives the matching
lower bound.
\end{proof}

\subsection{Classical attacks}
\label{subsec:chain-benchmark-summary}
\label{subsec:additional-numerics}
 
We compare the quantum mean above with the best scores found by
classical algorithms for the nonlinear residual chain.
 
\paragraph{The attacks and why they are reasonable.}
Each attack family targets one natural source of structure.
\begin{itemize}[leftmargin=*,itemsep=0.3em]
 \item \emph{Information-set interpolation.} Any $n$ values determine
 $P$. These attacks satisfy every test inside a window of $n$ values
 by choosing passing residuals, interpolate, and then move the window
 or the information set. The use of an information set follows the
 code-based search principle of Prange~\cite{Prange1962}; here the
 objective is to satisfy tests, rather than minimize error weight.
 The interpolation baseline is analogous to that for
 OPI~\cite{JordanEtAlDQI}. For a fixed interpolation window, its
 mean over the random instance ensemble is approximately $0.55$.

 \item \emph{Local search over feasible polynomials.} Greedy block
 updates~\cite{Besag1986}, simulated annealing~\cite{KirkpatrickGelattVecchi1983},
 and parallel tempering~\cite{HukushimaNemoto1996} run over the
 coefficients of $P$ or over information-set coordinates, so that every
 candidate is feasible. Additional moves keep already passing
 residuals in their sets. These methods combine direct improvement with
 stochastic exploration of the feasible polynomials.

 \item \emph{Inference followed by projection.} Without the polynomial
 constraint the problem is a chain. We use
 forward--backward inference for an anchored auxiliary chain
 distribution, accelerated by FFTs, followed by heuristic projection
 to a Reed--Solomon codeword~\cite{KschischangFreyLoeliger2001,RDQINumerics2026}.
 Only the resulting feasible score is counted. This attack also exploits the
 tractable source; its projection step is related to soft-decision
 Reed--Solomon decoding~\cite{KoetterVardy2003}.
\end{itemize}
 
\paragraph{Protocol.}
Methods and budgets were fixed on separate tuning instances and then
run on ten new residual-coefficient vectors, each with two independent
predicate targets and two solver streams. We tested 13 classical
methods, with additional warm-start variants, including Gibbs and
replica-exchange searches using the actual squared polynomial-guide
weight. Deep runs evaluated approximately $5\times10^6$ candidate
polynomials each; uniform sampling and fixed-window interpolation
used $2\times10^5$ and $5\times10^5$, respectively. Every scored
proposal counts toward the budget, including rejected moves.
After preprocessing, cached coordinate proposals cost $O(L)$ field
operations per candidate, while dense encoding and interpolation
cost $O(Ln)$; inference incurs additional overhead.
A method-by-method summary appears in
Appendix~\ref{app:chain-benchmark}, with implementation details in
Ref.~\cite{RDQINumerics2026}.

\begin{table}[!htbp]\centering\small
\begin{tabular}{cccc}\toprule
Coefficient vectors & Quantum mean & Classical median & Gap (pp)\\\midrule
10 & 0.64313 & 0.60666 & 3.65\\\bottomrule
\end{tabular}
\caption{Ideal quantum expected score and classical search results for
$p=2053$, $n=205$, and $D=50$. For each of the 10 coefficient vectors,
we average the best classical scores for its two predicate choices;
the table reports the median of these 10 averages. The gap is in
percentage points.}
\label{tab:additional-numerics}
\end{table}
 
\paragraph{Results.}
The gap was positive for all 10 coefficient vectors, with median
$3.65$ percentage points. For every coefficient vector, the averaged best classical
score remained more than one percentage point below the quantum mean.
These results compare an ideal quantum expectation with the best
scores found by the specified classical searches at finite budgets;
they do not establish a classical lower bound. Full instance and
protocol details are in Appendix~\ref{app:chain-benchmark}; the frozen
benchmark and scripts reproducing the table are provided in
Ref.~\cite{RDQINumerics2026}.

\bibliographystyle{alphaurl}
\bibliography{qipgenfram}

\clearpage
\appendix

\section{$C^*$-algebras, $*$-homomorphisms, and related facts}
\label{app:cstar-background}

This appendix collects the operator-algebraic facts used throughout the
main framework~\cite{Murphy1990}.  We work only with finite-dimensional represented
$C^*$-algebras, so all operators may be regarded as finite matrices.

\subsection{Finite-dimensional $C^*$-algebras}

A unital $C^*$-algebra $\mathcal A$ is a complex algebra equipped with
an involution $X\mapsto X^*$, an identity $I_{\mathcal A}$, and a norm
satisfying
\begin{equation}
 \|XY\|\le \|X\|\|Y\|,
 \qquad
 \|X^*\|=\|X\|,
 \qquad
 \|X^*X\|=\|X\|^2.
\end{equation}
In this work, $\mathcal A$ is always represented as a unital
$*$-subalgebra of operators on a finite-dimensional Hilbert space,
\begin{equation}
 \mathcal A\subseteq\mathcal B(\mathcal H_A).
\end{equation}
Thus addition, multiplication, and adjoint are the ordinary matrix
operations.

Every finite-dimensional $C^*$-algebra is $*$-isomorphic to a direct
sum of full matrix algebras~\cite{Murphy1990}.  In a concrete representation one may
write
\begin{equation}
 \mathcal H_A
 \cong
 \bigoplus_x
 \left(\mathbb C^{d_x}\otimes\mathbb C^{\mu_x}\right),
 \qquad
 \mathcal A
 \cong
 \bigoplus_x
 \left(M_{d_x}(\mathbb C)\otimes I_{\mu_x}\right).
 \label{eq:app-cstar-wedderburn}
\end{equation}
The factors $\mathbb C^{\mu_x}$ record representation multiplicities.
They do not change the abstract algebra, but they do change the
physical trace used in the main text.

For a represented algebra on $\mathcal H_A$, we use the normalized
physical trace
\begin{equation}
 \tau_A(X):=
 \frac{\operatorname{Tr}_{\mathcal H_A}(X)}
      {\dim\mathcal H_A},
\end{equation}
and the associated operator inner product
\begin{equation}
 \langle X,Y\rangle_A
 :=
 \tau_A(X^*Y),
 \qquad
 \|X\|_{2,A}^2=\tau_A(X^*X).
\end{equation}
The representation multiplicities in
Eq.~\eqref{eq:app-cstar-wedderburn} enter this trace.
Consequently, two isomorphic abstract $C^*$-algebras need not carry
the same normalized physical trace.

\subsection{$*$-homomorphisms and quotients}

A map
\begin{equation}
 \pi:\mathcal A\longrightarrow\mathcal B
\end{equation}
is a unital $*$-homomorphism if, for every $X,Y\in\mathcal A$ and
$\lambda\in\mathbb C$,
\begin{equation}
\begin{aligned}
 \pi(X+\lambda Y)&=\pi(X)+\lambda\pi(Y),\\
 \pi(XY)&=\pi(X)\pi(Y),\\
 \pi(X^*)&=\pi(X)^*,\\
 \pi(I_{\mathcal A})&=I_{\mathcal B}.
\end{aligned}
\end{equation}
Thus a $*$-homomorphism preserves all algebraic expressions built from
addition, multiplication, scalar multiplication, adjoints, and the
identity.

In the main framework, source and target generators satisfy
\begin{equation}
 \pi(a_i)=b_i,
 \qquad i=1,\ldots,s.
\end{equation}
Therefore every noncommutative $*$-polynomial $q$ obeys
\begin{equation}
\pi\!\left(q(a_1,\ldots,a_s,a_1^*,\ldots,a_s^*)\right)
 =
 q(b_1,\ldots,b_s,b_1^*,\ldots,b_s^*).
 \label{eq:app-cstar-polynomial-transfer}
\end{equation}
This elementary identity is the algebraic basis of polynomial
transduction.

The kernel of a $*$-homomorphism,
\begin{equation}
 \ker\pi:=\{X\in\mathcal A:\pi(X)=0\},
\end{equation}
is a two-sided $*$-ideal.  If $\pi$ is surjective, the first
isomorphism theorem gives~\cite{Murphy1990}
\begin{equation}
 \mathcal B\cong \mathcal A/\ker\pi.
 \label{eq:app-cstar-first-isomorphism}
\end{equation}
The target may therefore satisfy additional relations that are absent
from the source.  Relative decoding is designed to detect
the first such new relation inside the marked degree
filtration.

A $*$-homomorphism is automatically positive:
\begin{equation}
 X\ge0
 \quad\Longrightarrow\quad
 \pi(X)\ge0,
\end{equation}
because $X=Y^*Y$ implies
\begin{equation}
 \pi(X)=\pi(Y)^*\pi(Y)\ge0.
\end{equation}
It is also contractive in operator norm,
\begin{equation}
 \|\pi(X)\|\le\|X\|.
\end{equation}
Neither of these properties implies preservation of the normalized
physical trace.  In general,
\begin{equation}
 \tau_B(\pi(X))\ne\tau_A(X).
\end{equation}
The relative-distance condition in the main text identifies a
low-degree subspace on which these traces, and hence the corresponding
operator inner products, do agree.

\subsection{Spectrum cover property}
\label{app:spectral-properties}

For $X\in\mathcal A$, its spectrum is
\begin{equation}
 \operatorname{spec}(X)
 :=
 \{\lambda\in\mathbb C:X-\lambda I_{\mathcal A}
   \text{ is not invertible in }\mathcal A\}.
\end{equation}
For self-adjoint $H=H^*$ in finite dimension,
\begin{equation}
 H=\sum_{e\in\operatorname{spec}(H)}eP_e,
\end{equation}
where the $P_e$ are mutually orthogonal spectral projections satisfying
\begin{equation}
 P_eP_{e'}=\delta_{e,e'}P_e,
 \qquad
 \sum_eP_e=I.
\end{equation}

A unital $*$-homomorphism preserves polynomial and, more generally,
continuous functional calculus~\cite{Murphy1990}.  Thus
\begin{equation}
 \pi(f(H))=f(\pi(H))
 \label{eq:app-cstar-functional-calculus}
\end{equation}
for every function $f$ defined on the finite spectrum of $H$.
In particular,
\begin{equation}
 \pi(H)
 =
 \sum_e e\,\pi(P_e),
\end{equation}
where every $\pi(P_e)$ is again a projection, possibly zero, and
\begin{equation}
 \pi(P_e)\pi(P_{e'})=0
 \qquad(e\ne e').
\end{equation}

\begin{proposition}[Spectral cover]
\label{prop:app-cstar-spectral-cover}
Let $\pi:\mathcal A\to\mathcal B$ be a unital $*$-homomorphism and
let $H=H^*\in\mathcal A$. Then
\begin{equation}
 \operatorname{spec}(\pi(H))
 \subseteq
 \operatorname{spec}(H).
 \label{eq:app-cstar-spectrum-inclusion}
\end{equation}
More precisely, if
\begin{equation}
 H=\sum_e eP_e,
\end{equation}
then
\begin{equation}
 \pi(H)=\sum_e e\,\pi(P_e),
 \label{eq:app-cstar-spectral-projectors}
\end{equation}
and an eigenvalue $e$ occurs in $\pi(H)$ exactly when
$\pi(P_e)\ne0$.
\end{proposition}

\begin{proof}
If $\lambda\notin\operatorname{spec}(H)$, then
$H-\lambda I$ has an inverse $R\in\mathcal A$. Applying $\pi$ gives
\[
 (\pi(H)-\lambda I)\pi(R)
 =
 \pi((H-\lambda I)R)
 =
 I,
\]
and similarly
\[
 \pi(R)(\pi(H)-\lambda I)=I.
\]
Thus $\pi(H)-\lambda I$ is invertible, so
$\lambda\notin\operatorname{spec}(\pi(H))$. This proves
Eq.~\eqref{eq:app-cstar-spectrum-inclusion}.

For the second statement, apply $\pi$ to the spectral decomposition of
$H$. Multiplicativity and preservation of adjoints imply that the
$\pi(P_e)$ remain mutually orthogonal projections and sum to the
identity. Hence the nonzero $\pi(P_e)$ are the spectral
projections of $\pi(H)$.
\end{proof}

We refer to this inclusion as the \emph{spectral cover property}: the
source spectrum covers the target spectrum.  The homomorphism may
remove spectral sectors but cannot create a new eigenvalue.

Although a $*$-homomorphism cannot create new eigenvalues, it need not
preserve their physical multiplicities.  This distinction is important
for polynomial operator states and Gibbs states.

Suppose
\begin{equation}
 H_A=\sum_e eP_e^A,
 \qquad
 H_B:=\pi(H_A).
\end{equation}
Then
\begin{equation}
 P_e^B:=\pi(P_e^A)
\end{equation}
is either the target spectral projector at energy $e$ or zero.
However,
\begin{equation}
 \operatorname{Tr}_{\mathcal H_B}(P_e^B)
\end{equation}
need not equal
\begin{equation}
 \operatorname{Tr}_{\mathcal H_A}(P_e^A).
\end{equation}
The source and target can represent the same abstract matrix block with
different multiplicities, as in
Eq.~\eqref{eq:app-cstar-wedderburn}. Consequently,
\begin{equation}
 \tau_B(P_e^B)
 \ne
 \tau_A(P_e^A)
\end{equation}
in general.

This is why polynomial transport and norm preservation are logically
different statements.  Algebraically,
\begin{equation}
 \pi(p(H_A))=p(H_B)
\end{equation}
holds for every polynomial $p$ solely from the homomorphism property.
But the normalized operator-state norm
\begin{equation}
 \|p(H_L)\|_{2,L}^2
 =
 \tau_L(|p(H_L)|^2)
\end{equation}
depends on the physical trace.  Relative distance determines the degree through which these trace
inner products agree on the
low-degree subspace.

\subsection{Operator states and homomorphic transport}

For a represented algebra
$\mathcal A\subseteq\mathcal B(\mathcal H_A)$, define
\begin{equation}
 |\Phi_A\rangle
 :=
 \frac1{\sqrt{\dim\mathcal H_A}}
 \sum_j|j\rangle|j\rangle,
 \qquad
 |X\rangle\!\rangle_A
 :=
 (X\otimes I)|\Phi_A\rangle.
\end{equation}
Then
\begin{equation}
 \langle\!\langle X|Y\rangle\!\rangle_A
 =
 \tau_A(X^*Y).
\end{equation}
Hence a homomorphism $\pi$ does not automatically induce an isometry
between the physical operator-state spaces, because the normalized
traces may differ.

If, however, a subspace $\mathcal V\subseteq\mathcal A$ satisfies
\begin{equation}
 \tau_B(\pi(X)^*\pi(Y))
 =
 \tau_A(X^*Y)
 \qquad(X,Y\in\mathcal V),
\end{equation}
then $\pi$ is an isometry on $\mathcal V$ in the operator-state inner
product. In particular,
\begin{equation}
 \|\pi(X)\|_{2,B}=\|X\|_{2,A},
 \qquad X\in\mathcal V,
\end{equation}
and therefore $\pi$ is injective on $\mathcal V$.

The degree--distance criterion (Proposition~\ref{thm:mf-degree-distance}) applies this observation to
$\mathcal V=\mathcal F_D^A$.  The condition
\begin{equation}
 2D<d_{\mathrm{rel}}
\end{equation}
guarantees the inner-product preservation needed to transfer
a degree-$D$ polynomial state without changing its normalization.
Together with compatible operator bases, injectivity on
$\mathcal F_D^A$ also makes the target basis label invertible on the
degree-$D$ support, which is the algebraic reason that the source label
can be coherently recovered and uncomputed in the relative-decoding
circuit.

\section{Source design, constructive preparation}
\label{sec:constructing-sources}

\subsection{Designing a source by absorbing short relations}
\label{subsec:source-relations}

Consider the Pauli realization of
Section~\ref{subsec:mf-relative-codes}, with $m$ marked generators,
target label matrix $B$, relation code $K_B=\ker B$, and
anticommutation matrix $\Sigma=B^{\mathsf T}J_{n_B}B$.
For any source kernel $K_A\subseteq K_B$, set
\begin{equation}
 d_{A\to B}:=\min\{|c|:c\in K_B\setminus K_A\},
\end{equation}
with value $+\infty$ when $K_A=K_B$.
For an integer $t\ge1$, define the span of the short target relations
by
\begin{equation}
 K_B^{<t}:=\operatorname{span}_{\mathbb F_2}
             \{c\in K_B:|c|<t\}.
 \label{eq:source-short-span}
\end{equation}

\begin{proposition}[Source relations for a prescribed distance]
\label{thm:source-design}
We use the Pauli language for illustration. For $K_A\subseteq K_B$,
\begin{equation}
 d_{A\to B}\ge t
 \quad\Longleftrightarrow\quad
 K_B^{<t}\subseteq K_A.
 \label{eq:source-design-criterion}
\end{equation}
In particular, a proper source kernel $K_A\subsetneq K_B$ attaining
this distance exists if and only if $K_B^{<t}\subsetneq K_B$.
More generally, for $1\le r\le\dim K_B$, a source satisfying
\begin{equation}
 \dim(K_B/K_A)=r,\qquad d_{A\to B}\ge t
\end{equation}
exists if and only if
\begin{equation}
 \dim K_B^{<t}\le \dim K_B-r.
 \label{eq:source-distance-codimension}
\end{equation}
\end{proposition}

\begin{proof}
The inequality $d_{A\to B}\ge t$ means exactly that every target
relation of weight below $t$ lies in $K_A$. Since $K_A$ is a
subspace, this is equivalent to containing their span.
For the last assertion, such a span is contained in a subspace
of $K_B$ of dimension $\dim K_B-r$ exactly when
Eq.~\eqref{eq:source-distance-codimension} holds.
\end{proof}

For a degree-$D$ promise, the relative-distance condition
$2D<d_{A\to B}$ therefore becomes
\begin{equation}
 K_B^{<2D+1}\subseteq K_A.
 \label{eq:source-degree-design}
\end{equation}
Thus the source must already represent every target relation
involving at most $2D$ marked generators. As in
Proposition~\ref{prop:mf-relative-vs-ordinary}, enlarging $K_A$ can only
increase the relative distance:
\begin{equation}
 K_{A'}\subseteq K_A\subseteq K_B
 \quad\Longrightarrow\quad
 d_{A'\to B}\le d_{A\to B}.
\end{equation}
There is also a useful obstruction. If $K_B$ is generated by
relations of weight at most $w$, every proper source kernel
satisfies $d_{A\to B}\le w$. A large relative distance in the
presence of short target relations requires those short relations
to span a proper subspace of $K_B$.

The next proposition constructs signed Pauli sources realizing these
criteria.

\begin{proposition}[Realization of a prescribed source kernel]
\label{lem:source-realization}
Fix signed Hermitian target Pauli generators $b_1,\ldots,b_m$.
For every given subspace $K_A\subseteq K_B$, one can construct
signed Hermitian Pauli generators $a_1,\ldots,a_m$ with relation
kernel $K_A$ such that $a_i\mapsto b_i$ extends to a unital
$*$-homomorphism. The construction uses polynomial-time binary
linear algebra and Pauli phase arithmetic. It can be realized on
\begin{equation}
 n_A=q+s_A,\qquad
 q:=\tfrac12\operatorname{rank}\Sigma,\qquad
 s_A:=\dim\ker\Sigma-\dim K_A
 \label{eq:source-realization-size}
\end{equation}
qubits.
\end{proposition}

\begin{proof}
Since $K_A\subseteq K_B\subseteq\ker\Sigma$, the alternating form
$c^{\mathsf T}\Sigma d$ descends to $V:=\mathbb F_2^m/K_A$.
It has rank $2q$ and radical dimension $s_A$.
Binary elimination gives a basis consisting of $q$ symplectic
pairs and $s_A$ radical vectors. Map the pairs to the $X$ and $Z$
labels on $q$ qubits, and the radical vectors to independent $Z$
labels on $s_A$ further qubits. This embeds $V$ injectively in a
Pauli label space. Composing with the quotient map gives a label
matrix $A$ satisfying
\[
 \ker A=K_A,\qquad A^{\mathsf T}J_{n_A}A=\Sigma.
\]

Choose Hermitian Pauli representatives $a_i^0$ of its columns.
For $c\in K_A$, let $\zeta_A^0(c)$ and $\zeta_B(c)$ be the scalar
phases of the ordered source and target products.
Their multiplication rules have the same reordering cocycle,
determined by $\Sigma$. Hence
$\chi(c):=\zeta_B(c)/\zeta_A^0(c)$ is a character of $K_A$.
It takes values in $\{\pm1\}$ and extends to a character
$(-1)^{u\cdot c}$ of $\mathbb F_2^m$.
Setting $a_i:=(-1)^{u_i}a_i^0$ matches all source scalar phases
with the target. The homomorphism conditions
Eq.~\eqref{eq:mf-pauli-compatibility} now hold.
Finding the basis, evaluating phases on a basis of $K_A$, and
solving for $u$ all take polynomial time.
\end{proof}

The proposition takes a basis of $K_A$ as input, leaving the tasks of
finding $K_B^{<t}$ and arranging a source with suitable geometry and
efficient state preparation. The block constructions below address
the latter tasks for a concrete family.

\paragraph{An algebraic formulation of the same design principle.}
Let $\mathcal U$ be a common marked finite-dimensional
$C^*$-algebra, let $\mathcal F_D^{\mathcal U}$ be its word
filtration, and write the target as $\mathcal U/I_B$ for a proper
two-sided $*$-ideal $I_B$. Define
\begin{equation}
 J_D:=\left\langle I_B\cap\mathcal F_D^{\mathcal U}
       \right\rangle_{\text{two-sided }*\text{-ideal}}.
 \label{eq:source-generated-ideal}
\end{equation}
An intermediate source $\mathcal U/I_A$, with $I_A\subseteq I_B$,
has an induced map $\pi$ to $\mathcal U/I_B$. If
$q_A:\mathcal U\to\mathcal U/I_A$ is the quotient map, then
\begin{equation}
 \pi|_{q_A(\mathcal F_D^{\mathcal U})}\text{ is injective}
 \quad\Longleftrightarrow\quad J_D\subseteq I_A.
 \label{eq:source-ideal-criterion}
\end{equation}
Indeed, the kernel of the restriction is
$q_A(I_B\cap\mathcal F_D^{\mathcal U})$.
Thus a proper intermediate ideal $I_A\subsetneq I_B$ with this
property exists exactly when $J_D\subsetneq I_B$, and $J_D$
is the smallest possible choice. This establishes linear injectivity;
physical-trace matching and the circuit requirements of the framework
remain separate conditions.

\subsection{Efficient guides from independent blocks}
\label{subsec:source-block-guides}

Independent blocks provide sources that can contain internal
relations without sacrificing tractable guide preparation.
Let
\begin{equation}
 \mathcal H_A=\bigotimes_{j=1}^{L}\mathcal H_j,\qquad
 H_A=\sum_{j=1}^{L}h_j,\qquad d_j:=\dim\mathcal H_j,
 \label{eq:source-block-hamiltonian}
\end{equation}
where $h_j=h_j^*$ acts only on block $j$ and identities on the
other blocks are implicit. Each $h_j$ may itself be a sum of
noncommuting marked generators. Group each physical block with
its reference block in the vectorization.

\begin{proposition}[Polynomial guides of independent blocks]
\label{prop:source-block-mps}
Let $p(x)=\sum_{k=0}^{D}p_kx^k$ and suppose $p(H_A)\ne0$.
The normalized canonical operator state
\begin{equation}
 |\psi_A^p\rangle
 =\frac{(p(H_A)\otimes I)|\Phi_A\rangle}
        {\sqrt{\tau_A(p(H_A)^*p(H_A))}}
 \label{eq:source-block-guide}
\end{equation}
has an exact matrix product state representation of bond
dimension at most $D+1$ in the paired-block ordering.
Its tensors can be constructed explicitly from the local
matrices $h_j$ and the coefficients of $p$.
\end{proposition}

\begin{proof}
Across any cut between blocks, write
$H_A=H_{\mathrm L}\otimes I+I\otimes H_{\mathrm R}$.
The binomial formula gives
\[
 p(H_A)=\sum_{r=0}^{D}H_{\mathrm L}^{r}\otimes
 \left(\sum_{k=r}^{D}p_k\binom{k}{r}H_{\mathrm R}^{k-r}\right).
\]
Vectorization turns this into a sum of at most $D+1$ product
vectors, so every such Schmidt rank is at most $D+1$.

For an explicit construction, use auxiliary indices
$0\le a,b\le D$ and define operator-valued upper-triangular
matrices
\begin{equation}
 (T_j)_{ab}:=
 \begin{cases}
 h_j^{\,b-a}/(b-a)!,&b\ge a,\\
 0,&b<a.
 \end{cases}
 \label{eq:source-block-tensors}
\end{equation}
Multiplication of these matrices uses tensor products between
different physical blocks. The multinomial formula yields
\begin{equation}
 (T_1\cdots T_L)_{0k}
 =\sum_{r_1+\cdots+r_L=k}
       \bigotimes_{j=1}^{L}\frac{h_j^{r_j}}{r_j!}
 =\frac{H_A^k}{k!}.
 \label{eq:source-block-tensor-product}
\end{equation}
Thus the left boundary $e_0^{\mathsf T}$ and right boundary
$(0!p_0,1!p_1,\ldots,D!p_D)^{\mathsf T}$ contract to $p(H_A)$.
Vectorizing each local operator produces the claimed matrix
product state.
\end{proof}

After canonicalization, this matrix product state can be prepared
sequentially with a number of local operations polynomial in the
number of blocks $L$, the bond dimension $D+1$, and the paired-block
dimensions $d_j^2$, with gates synthesized to the desired
precision~\cite{SchonEtAl2005,Schollwock2011}.

\section{Bounded-degree approximate optimization: proofs and details}
\label{app:optimization-details}
\label{subsec:value-general:details}
\label{sec:achieved-value:details}

This appendix gives details for Section~\ref{subsec:value-general}. Fix a marked degree bound $D$ and a self-adjoint target observable
$O_B\in\mathcal B$.  Our goal is to optimize the expectation value of
$O_B$ over target states obtained by transducing source guide operators of
marked degree at most $D$.  Thus we choose
\begin{equation}
 0\ne X_A\in\mathcal F_D^A,
 \qquad
 X_B:=\pi(X_A)\ne0,
\end{equation}
and prepare the normalized target operator state whose marginal is
\begin{equation}
 \rho_B[X_A]
 :=\frac{X_BX_B^*}{\operatorname{Tr}(X_BX_B^*)}.
 \label{eq:value-target-state}
\end{equation}
The corresponding objective value is
\begin{equation}
 \mathcal E_B(O_B;X_A)
 :=\operatorname{Tr}\!\left(O_B\rho_B[X_A]\right)
 =\frac{\tau_B(X_B^*O_BX_B)}{\tau_B(X_B^*X_B)}.
 \label{eq:value-born:details}
\end{equation}

The optimization problem therefore has three layers.  First, one chooses a
bounded-degree guide space in the source.  Second, one writes the exact
target expectation as a finite variational problem.  Third, relative
distance determines when this target variational problem can be evaluated
entirely from source data.  We treat these steps in this order and then
specialize to the full degree-$D$ word space and to the smaller polynomial
subspace $\operatorname{span}\{I,H_A,\ldots,H_A^D\}$.

\subsection{Exact variational formulation on a chosen guide space}
\label{app:optimization-variational}

Suppose first that the target observable has a self-adjoint source lift,
\begin{equation}
 O_B=\pi(O_A),
 \qquad
 O_A=O_A^*\in\mathcal A.
\end{equation}
Introduce the pulled-back target trace
\begin{equation}
 \tau_\pi:=\tau_B\circ\pi.
 \label{eq:value-pullback-trace}
\end{equation}
It is a positive normalized trace on $\mathcal A$, possibly nonfaithful.
Multiplicativity and preservation of adjoints give
\begin{equation}
 \mathcal E_B(O_B;X_A)
 =\frac{\tau_\pi(X_A^*O_AX_A)}{\tau_\pi(X_A^*X_A)}.
 \label{eq:value-pullback-quotient}
\end{equation}
In general $\tau_\pi\ne\tau_A$. Thus the source spectrum alone does not determine the target mean.

Now choose a finite-dimensional guide space
\begin{equation}
 \mathcal W_D:=\operatorname{span}\{F_0,\ldots,F_{R-1}\}
 \subseteq\mathcal F_D^A
\end{equation}
and write
\begin{equation}
 X_A(c):=\sum_{j=0}^{R-1}c_jF_j,
 \qquad
 Y_j:=\pi(F_j).
 \label{eq:value-guide-family}
\end{equation}
Define the target Gram and objective matrices
\begin{equation}
 (G_\pi)_{ij}:=\tau_B(Y_i^*Y_j),
 \qquad
 (K_\pi)_{ij}:=\tau_B(Y_i^*O_BY_j).
 \label{eq:value-general-matrices:details}
\end{equation}
When the lift $O_A$ is used, these entries are
$\tau_\pi(F_i^*F_j)$ and $\tau_\pi(F_i^*O_AF_j)$.

\begin{proposition}[Variational objective matrix]
\label{prop:value-ritz:details}
Assume the guide family has a nonzero target image.  For every $c$ with
$c^*G_\pi c>0$,
\begin{equation}
 \mathcal E_B(O_B;X_A(c))
 =\frac{c^*K_\pi c}{c^*G_\pi c}.
 \label{eq:value-general-rayleigh:details}
\end{equation}
The maximum over $\mathcal W_D$ is
\begin{equation}
 \max_{c:\,c^*G_\pi c>0}\mathcal E_B(O_B;X_A(c))
 =\lambda_{\max}\!\left(G_+^{-1/2}K_+G_+^{-1/2}\right),
 \label{eq:value-general-optimum:details}
\end{equation}
where $G_+$ and $K_+$ are the restrictions to
$\operatorname{ran}G_\pi
 =
 \bigl(\ker G_\pi\bigr)^\perp$, therefore, $G_+$ is strictly positive definite. Equivalently, after removing zero-norm
directions, one solves
\begin{equation}
 K_\pi c=\lambda G_\pi c.
\end{equation}
For minimization, replace the largest eigenvalue by the smallest.
\end{proposition}
\begin{proof}
Define the corresponding target guide operator by
\[
 X_B(c):=\pi(X_A(c))=\sum_{j=0}^{R-1}c_jY_j.
\]
By the definitions of $G_\pi$ and $K_\pi$,
\[
 c^*G_\pi c
 =
 \tau_B\!\left(X_B(c)^*X_B(c)\right)
 =
 \|X_B(c)\|_{2,B}^2,
\]
and
\[
 c^*K_\pi c
 =
 \tau_B\!\left(X_B(c)^*O_BX_B(c)\right).
\]
Substituting these identities into
Eq.~\eqref{eq:value-born:details} proves
Eq.~\eqref{eq:value-general-rayleigh:details}.

Now suppose $c\in\ker G_\pi$. Since $G_\pi\ge0$,
\[
 0=c^*G_\pi c=\|X_B(c)\|_{2,B}^2,
\]
and hence $X_B(c)=0$. Therefore, for every $i$,
\[
 (K_\pi c)_i
 =
 \tau_B\!\left(Y_i^*O_BX_B(c)\right)
 =0,
\]
so
\[
 \ker G_\pi\subseteq\ker K_\pi.
\]
Thus adding a vector in $\ker G_\pi$ changes neither the target guide
nor either quadratic form. The optimization may therefore be restricted
to
\[
 \operatorname{ran}G_\pi
 =
 (\ker G_\pi)^\perp.
\]

Let $V$ have orthonormal columns spanning
$\operatorname{ran}G_\pi$, and define
\[
 G_+:=V^*G_\pi V,
 \qquad
 K_+:=V^*K_\pi V.
\]
Then $G_+>0$. Writing $c=V\widetilde c$, the quotient becomes
\[
 \frac{c^*K_\pi c}{c^*G_\pi c}
 =
 \frac{\widetilde c^*K_+\widetilde c}
      {\widetilde c^*G_+\widetilde c}.
\]
With the change of variables
\[
 z:=G_+^{1/2}\widetilde c,
\]
this is
\[
 \frac{
 z^*G_+^{-1/2}K_+G_+^{-1/2}z
 }{z^*z}.
\]
The Rayleigh--Ritz principle therefore gives
\[
 \max_{c:\,c^*G_\pi c>0}
 \frac{c^*K_\pi c}{c^*G_\pi c}
 =
 \lambda_{\max}\!\left(
 G_+^{-1/2}K_+G_+^{-1/2}
 \right),
\]
which proves Eq.~\eqref{eq:value-general-optimum:details}. Equivalently, on the
positive-norm subspace, the stationary vectors satisfy the generalized
eigenvalue equation
\[
 K_\pi c=\lambda G_\pi c.
\]
\end{proof}

A general sufficient condition for evaluating the same variational problem
from the source only is the following.  If, for one common $\lambda>0$,
\begin{equation}
 \tau_\pi(Y)=\lambda\tau_A(Y)
 \quad\text{on}\quad
 \operatorname{span}\{X^*Z,\,X^*O_AZ:X,Z\in\mathcal W_D\},
 \label{eq:value-trace-matching}
\end{equation}
then $G_\pi=\lambda G_A$ and $K_\pi=\lambda K_A$, so the factor
$\lambda$ cancels from every Rayleigh quotient.  Balance on
$\mathcal W_D$ alone gives the Gram-matrix identity; preserving the
objective requires the additional insertion of $O_A$.  The relative-distance
theorem below gives a uniform degree criterion for
Eq.~\eqref{eq:value-trace-matching}.

\subsection{Relative distance reduces the target problem to the source}
\label{subsec:value-relative:details}

The variational matrices above are defined using the target trace.  The
main optimization consequence of relative decoding is that, in a
low-degree window, these target traces equal source traces.

\begin{theorem}[Expected-value transfer under relative decoding]
\label{thm:value-relative-expectation:details}
Let $0\ne X_A\in\mathcal F_D^A$ and let
$O_A=O_A^*\in\mathcal F_r^A$.  If
\begin{equation}
 2D+r<d_{\mathrm{rel}},
 \label{eq:value-observable-window:details}
\end{equation}
then $\pi(X_A)\ne0$ and
\begin{equation}
 \mathcal E_B(\pi(O_A);X_A)
 =\frac{\tau_A(X_A^*O_AX_A)}{\tau_A(X_A^*X_A)}.
 \label{eq:value-relative-mean:details}
\end{equation}
Consequently, for any guide space
$\mathcal W_D\subseteq\mathcal F_D^A$, the source Gram and objective
matrices determine the target variational problem.
\end{theorem}

\begin{proof}
The products $X_A^*X_A$ and $X_A^*O_AX_A$ have marked degrees at most
$2D$ and $2D+r$.  Eq.~\eqref{eq:mf-trace-distance-characterization} therefore
identifies their source and target traces.  The source denominator is
positive, so its target image is nonzero.  The same argument applied to
$F_i^*F_j$ and $F_i^*O_AF_j$ identifies the corresponding Gram and
objective matrices.
\end{proof}

\paragraph{Hamiltonian-value optimization.}
The central energy-optimization problem is obtained by taking
\begin{equation}
 O_A=H_A,
 \qquad
 O_B=H_B=\pi(H_A).
\end{equation}
Because the indexed Hamiltonian terms are the marked generators,
$H_A\in\mathcal F_1^A$ and therefore the observable degree in
Theorem~\ref{thm:value-relative-expectation:details} is $r=1$.  Hence a
marked-degree-$D$ guide has a source-computable target energy whenever
\begin{equation}
 2D+1<d_{\mathrm{rel}}.
 \label{eq:value-energy-window:details}
\end{equation}
The full-word matrix $M_D^A$ and the Jacobi matrix $J_D$ introduced below
are the two corresponding energy-optimization problems over,
respectively, all of $\mathcal F_D^A$ and the smaller polynomial guide
space.

\subsection{Two bounded-degree optimization spaces}
\label{app:optimization-spaces}

Theorem~\ref{thm:value-relative-expectation:details} applies to any chosen
subspace of $\mathcal F_D^A$.  We now distinguish two natural choices.
The first is the full word space $\mathcal F_D^A$; the second is the
smaller polynomial, or Krylov, space
$\operatorname{span}\{I,H_A,\ldots,H_A^D\}$.

\paragraph{The full degree-$D$ word space.} Here, without loss of generality, we use the spin Pauli case for illustration. For $L\in\{A,B\}$, define the ordered words
\begin{equation}
 W_A(c)=a_1^{c_1}\cdots a_s^{c_s},
 \qquad
 W_B(c)=b_1^{c_1}\cdots b_s^{c_s}.
\end{equation}
Let $\mathcal R_D\subseteq\mathbb F_2^s$ contain one representative of
weight at most $D$ for each source coset in the radius-$D$ quotient
ball.  Then
\begin{equation}
 \{W_A(c):c\in\mathcal R_D\}
\end{equation}
is an orthonormal basis of $\mathcal F_D^A$.  The homomorphism sends $W_A(c)$ to $W_B(c)$.

Let
\begin{equation}
 \Sigma=A^{\mathsf T}J_{n_A}A=B^{\mathsf T}J_{n_B}B
\end{equation}
and define
\begin{equation}
 \kappa(c,d):=(-1)^{\sum_{i>j}c_i d_j\Sigma_{ij}},
 \qquad
 t_L(z):=
 \begin{cases}
  \zeta_L(z),&z\in K_L,\\
  0,&z\notin K_L.
 \end{cases}
 \label{eq:value-word-phases}
\end{equation}
Then
\begin{equation}
 W_L(c)W_L(d)=\kappa(c,d)W_L(c+d),
 \qquad
 W_L(c)^*=\kappa(c,c)W_L(c),
 \qquad
 \tau_L(W_L(z))=t_L(z),
 \label{eq:value-word-products}
\end{equation}
with all additions in $\mathbb F_2^s$.

For the degree-one Hamiltonians
\begin{equation}
 H_A=\sum_{i=1}^s h_i a_i,
 \qquad
 H_B=\sum_{i=1}^s h_i b_i,
 \qquad
 h_i\in\mathbb R,
\end{equation}
define matrices indexed by $c,d\in\mathcal R_D$:
\begin{equation}
 (G_D^L)_{cd}
 :=\kappa(c,c)\kappa(c,d)t_L(c+d),
 \label{eq:value-relative-gram}
\end{equation}
\begin{equation}
 (M_D^L)_{cd}
 :=\sum_{i=1}^s h_i\,
 \kappa(c,c)\kappa(c,e_i)\kappa(c+e_i,d)
 t_L(c+e_i+d).
 \label{eq:value-relative-score-matrix}
\end{equation}
These formulas are valid without a distance assumption and satisfy
\begin{equation}
 (G_D^L)_{cd}=\tau_L(W_L(c)^*W_L(d)),
 \qquad
 (M_D^L)_{cd}=\tau_L(W_L(c)^*H_LW_L(d)).
 \label{eq:value-relative-matrix-meaning:details}
\end{equation}
In particular, both matrices are Hermitian.

For
\begin{equation}
 X_A(z)=\sum_{c\in\mathcal R_D}z_cW_A(c),
\end{equation}
Proposition~\ref{prop:value-ritz:details} gives the exact target value
\begin{equation}
 \mathcal E_B(H_B;X_A(z))
 =\frac{z^*M_D^Bz}{z^*G_D^Bz},
 \label{eq:value-relative-exact-quotient}
\end{equation}
whenever $z^*G_D^Bz>0$.  By construction $G_D^A=I$, and
Theorem~\ref{thm:value-relative-expectation:details} gives
\begin{equation}
 2D<d_{\mathrm{rel}}\Longrightarrow G_D^B=I,
 \qquad
 2D+1<d_{\mathrm{rel}}\Longrightarrow M_D^B=M_D^A.
 \label{eq:value-relative-matrix-windows:details}
\end{equation}
Consequently,
\begin{equation}
 \max_{0\ne X_A\in\mathcal F_D^A}
 \mathcal E_B(H_B;X_A)
 =\lambda_{\max}(M_D^A),
 \qquad
 2D+1<d_{\mathrm{rel}}.
 \label{eq:value-full-relative-optimum:details}
\end{equation}
For minimization use $\lambda_{\min}(M_D^A)$.  The optimizing
eigenvector gives the coefficients of the source guide in the ordered-word
basis.  If only $2D<d_{\mathrm{rel}}$ holds, the guide still transfers,
and the target optimum is determined by $M_D^B$.

\paragraph{Polynomial guides and the Jacobi matrix.} This part applies to any realization. A smaller variational family is obtained by restricting to
\begin{equation}
 X_A=p(H_A),
 \qquad
 \deg p\le D.
\end{equation}
Without a distance assumption, the exact target value is described by the
pulled-back spectral weights
\begin{equation}
 \nu_\pi(e):=\tau_\pi(P_e^A),
 \qquad
 H_A=\sum_e eP_e^A,
\end{equation}
namely
\begin{equation}
 \mathcal E_B(H_B;p(H_A))
 =\frac{\sum_e\nu_\pi(e)e|p(e)|^2}
        {\sum_e\nu_\pi(e)|p(e)|^2}.
 \label{eq:value-general-spectral}
\end{equation}
Equivalently, in the monomial family $1,H_A,\ldots,H_A^D$, the exact
target moment matrices are
\begin{equation}
 (G_\pi)_{ij}=m_{i+j}^\pi,
 \qquad
 (K_\pi)_{ij}=m_{i+j+1}^\pi,
 \qquad
 m_k^\pi:=\tau_\pi(H_A^k)=\tau_B(H_B^k).
 \label{eq:value-target-moment-matrices}
\end{equation}
Relative distance determines when these target moments may be replaced by
source moments.

This polynomial family is the Krylov subspace~\cite{GolubMeurant2010}
\begin{equation}
 \mathcal K_D(H_A)
 :=\operatorname{span}\{I,H_A,\ldots,H_A^D\}
 \subseteq\mathcal F_D^A,
\end{equation}
which has dimension at most $D+1$ even when $\mathcal F_D^A$ is much
larger.  Let
\begin{equation}
 H_A=\sum_e eP_e^A,
 \qquad
 \mu_A(e):=\tau_A(P_e^A),
 \qquad
 m_k:=\tau_A(H_A^k)=\sum_e\mu_A(e)e^k.
 \label{eq:value-source-moments}
\end{equation}
For $p(x)=\sum_{j=0}^D c_jx^j$, define
\begin{equation}
 (G_D^{\mathrm{poly}})_{ij}=m_{i+j},
 \qquad
 (K_D^{\mathrm{poly}})_{ij}=m_{i+j+1},
 \qquad
 0\le i,j\le D.
 \label{eq:value-polynomial-matrices}
\end{equation}
Theorem~\ref{thm:value-relative-expectation:details} gives
\begin{equation}
 \mathcal E_B(H_B;p(H_A))
 =\frac{\sum_e\mu_A(e)e|p(e)|^2}
        {\sum_e\mu_A(e)|p(e)|^2}
 =\frac{c^*K_D^{\mathrm{poly}}c}
        {c^*G_D^{\mathrm{poly}}c},
 \qquad
 2D+1<d_{\mathrm{rel}}.
 \label{eq:value-polynomial-rayleigh:details}
\end{equation}
Thus only source moments through order $2D+1$ enter the polynomial
optimization.

To turn the polynomial-guide optimization into an ordinary Hermitian
eigenvalue problem, consider the space of real polynomials of degree at
most $D$, viewed as functions on the source spectrum:
\begin{equation}
 \mathcal P_D
 :=
 \operatorname{span}\{1,x,\ldots,x^D\}
 \subseteq L^2(\mu_A).
\end{equation}
Let
\begin{equation}
 N_{\mathrm{lev}}
 :=
 |\operatorname{supp}\mu_A|
 =
 |\operatorname{spec}(H_A)|
\end{equation}
be the number of distinct source energy levels.  Since a function on
these $N_{\mathrm{lev}}$ spectral points has at most
$N_{\mathrm{lev}}$ independent values,
\begin{equation}
 r_D:=\dim\mathcal P_D
 =
 \min\{D+1,N_{\mathrm{lev}}\}.
\end{equation}

Apply Gram--Schmidt~\cite{GolubMeurant2010} to
\begin{equation}
 1,x,\ldots,x^D
\end{equation}
in the inner product
\begin{equation}
 \langle f,g\rangle_{L^2(\mu_A)}
 :=
 \int f(x)g(x)\,d\mu_A(x).
\end{equation}
After removing any zero-norm directions, this produces real
orthonormal polynomials
\begin{equation}
 \varphi_0,\varphi_1,\ldots,\varphi_{r_D-1}
\end{equation}
satisfying
\begin{equation}
 \deg\varphi_j=j,
 \qquad
 \int\varphi_i(x)\varphi_j(x)\,d\mu_A(x)=\delta_{ij}.
\end{equation}
The source spectral measure $\mu_A$ determines these polynomials up to
an overall sign at each degree.  Writing
\begin{equation}
 \varphi_j(x)
 =
 \kappa_jx^j+\text{lower-degree terms},
\end{equation}
we fix this sign ambiguity by requiring the leading coefficient
$\kappa_j$ to be positive.  Since $\mu_A$ is a probability measure,
$\varphi_0=1$.

Let $\mathsf M_x$ denote multiplication by the energy variable,
\begin{equation}
 (\mathsf M_xf)(x):=xf(x),
\end{equation}
and let $\Pi_D$ be the orthogonal projection onto $\mathcal P_D$.
The Jacobi matrix $J_D$ is the matrix of the compressed multiplication
operator
\begin{equation}
 \left.\Pi_D\mathsf M_x\Pi_D\right|_{\mathcal P_D}
\end{equation}
in the basis
$\{\varphi_0,\ldots,\varphi_{r_D-1}\}$:
\begin{equation}
 (J_D)_{ij}
 :=
 \int\varphi_i(x)x\varphi_j(x)\,d\mu_A(x),
 \qquad
 0\le i,j<r_D.
 \label{eq:value-jacobi-definition:details}
\end{equation}

To see the structure of this matrix, expand
\begin{equation}
 \Pi_D(x\varphi_j)
 =
 \sum_{i=0}^{r_D-1}
 \langle\varphi_i,x\varphi_j\rangle_{L^2(\mu_A)}
 \varphi_i.
\end{equation}
If $i\ge j+2$, then $x\varphi_j$ has degree at most $i-1$, so
orthogonality gives
\begin{equation}
 \langle\varphi_i,x\varphi_j\rangle_{L^2(\mu_A)}=0.
\end{equation}
If $i\le j-2$, then multiplication by the real variable $x$ is
self-adjoint in $L^2(\mu_A)$, and hence
\begin{equation}
 \langle\varphi_i,x\varphi_j\rangle_{L^2(\mu_A)}
 =
 \langle x\varphi_i,\varphi_j\rangle_{L^2(\mu_A)}
 =
 0,
\end{equation}
because $x\varphi_i$ has degree at most $j-1$.  Therefore only the
components with $i=j-1,j,j+1$ can remain.  Defining
\begin{equation}
 \alpha_j
 :=
 \int x\varphi_j(x)^2\,d\mu_A(x),
 \qquad
 \beta_{j+1}
 :=
 \int x\varphi_j(x)\varphi_{j+1}(x)\,d\mu_A(x),
\end{equation}
we obtain the three-term recurrence
\begin{equation}
 \Pi_D(x\varphi_j)
 =
 \beta_j\varphi_{j-1}
 +
 \alpha_j\varphi_j
 +
 \beta_{j+1}\varphi_{j+1},
\end{equation}
where terms with indices outside
$\{0,\ldots,r_D-1\}$ are omitted.  With the positive-leading-coefficient
convention,
\begin{equation}
 \beta_{j+1}
 =
 \frac{\kappa_j}{\kappa_{j+1}}
 >0
 \qquad
 (j+1<r_D).
\end{equation}

Equivalently,
\begin{equation}
 (J_D)_{ij}
 =
 \begin{cases}
  \alpha_j,&i=j,\\
  \beta_{j+1},&i=j+1,\\
  \beta_{i+1},&j=i+1,\\
  0,&|i-j|>1.
 \end{cases}
 \label{eq:value-jacobi-entries}
\end{equation}
Thus $J_D$ is real, symmetric, and tridiagonal.

The degree cutoff $D$ determines both the admissible polynomial-guide
space and the size of the Jacobi matrix.  If
$D<N_{\mathrm{lev}}$, then $J_D$ has size $D+1$.  Once
$D\ge N_{\mathrm{lev}}-1$, the space saturates at dimension
$N_{\mathrm{lev}}$, because higher-degree polynomials introduce no new
functions on the finite spectrum of $H_A$.

\begin{proposition}[Optimal polynomial-guide value]
\label{prop:value-polynomial-optimum:details}
If $H_A$ has marked degree one and $2D+1<d_{\mathrm{rel}}$, then
\begin{equation}
 \max_{\substack{\deg p\le D\\\tau_A(|p(H_A)|^2)>0}}
 \mathcal E_B(H_B;p(H_A))
 =\lambda_{\max}(J_D).
 \label{eq:value-jacobi-optimum:details}
\end{equation}
If $z$ is a unit top eigenvector, then
$p_\star=\sum_{j=0}^{r_D-1}z_j\varphi_j$ attains the optimum.  For energy
minimization use the bottom eigenpair.
\end{proposition}

\begin{proof}
Since $H_A\in\mathcal F_1^A$ and $\deg p\le D$, one has
$p(H_A)\in\mathcal F_D^A$.  The condition
$2D+1<d_{\mathrm{rel}}$ therefore allows
Theorem~\ref{thm:value-relative-expectation:details} to be applied with
$X_A=p(H_A)$ and $O_A=H_A$.  Hence
\[
 \mathcal E_B(H_B;p(H_A))
 =
 \frac{\int x|p(x)|^2\,d\mu_A(x)}
      {\int |p(x)|^2\,d\mu_A(x)}.
\]

Expand $p$ in the orthonormal polynomial basis:
\[
 p(x)=\sum_{j=0}^{r_D-1}z_j\varphi_j(x).
\]
Since the $\varphi_j$ are orthonormal in $L^2(\mu_A)$,
\[
 \int |p(x)|^2\,d\mu_A(x)
 =
 \sum_{j=0}^{r_D-1}|z_j|^2
 =
 z^*z.
\]
Moreover, by the definition of $J_D$,
\[
\begin{aligned}
 \int x|p(x)|^2\,d\mu_A(x)
 &=
 \sum_{i,j=0}^{r_D-1}
 \overline{z_i}z_j
 \int\varphi_i(x)x\varphi_j(x)\,d\mu_A(x)\\
 &=
 z^*J_Dz.
\end{aligned}
\]
Therefore every nonzero admissible polynomial satisfies
\[
 \mathcal E_B(H_B;p(H_A))
 =
 \frac{z^*J_Dz}{z^*z}.
\]
Conversely, every $z\in\mathbb C^{r_D}$ defines a polynomial of degree at
most $D$ in this basis.  The Rayleigh--Ritz principle now gives
\[
 \max_{z\ne0}\frac{z^*J_Dz}{z^*z}
 =
 \lambda_{\max}(J_D),
\]
which proves Eq.~\eqref{eq:value-jacobi-optimum:details}.

If $z$ is a unit eigenvector of $J_D$ with eigenvalue
$\lambda_{\max}(J_D)$, then
\[
 p_\star(x)=\sum_{j=0}^{r_D-1}z_j\varphi_j(x)
\]
attains the maximum.  Replacing the top eigenpair by the bottom
eigenpair gives the corresponding minimum.
\end{proof}

The two variational problems are different:
\begin{equation}
 \lambda_{\max}(J_D)
 \le\lambda_{\max}(M_D^A)
 \le\lambda_{\max}(H_B),
 \qquad
 2D+1<d_{\mathrm{rel}}.
 \label{eq:value-guide-space-comparison}
\end{equation}
The Jacobi matrix optimizes only polynomial functions of $H_A$; the full
matrix $M_D^A$ optimizes every guide in $\mathcal F_D^A$.

\paragraph{Computability of the Jacobi matrix.}
A sufficient condition for computing $J_D$ efficiently is access, to the
required precision, to the source moments
\begin{equation}
 m_k=\tau_A(H_A^k),
 \qquad
 0\le k\le 2D+1.
\end{equation}
From these moments one can construct the orthonormal polynomials and
their recurrence coefficients by
\begin{equation}
 \alpha_j=\langle\varphi_j,x\varphi_j\rangle_{L^2(\mu_A)},
\end{equation}
\begin{equation}
 r_{j+1}
 =(x-\alpha_j)\varphi_j-\beta_j\varphi_{j-1},
 \qquad
 \beta_{j+1}=\|r_{j+1}\|_{L^2(\mu_A)},
 \qquad
 \varphi_{j+1}=r_{j+1}/\beta_{j+1}.
\end{equation}
Thus, if $D$ is polynomially bounded, the moments through order
$2D+1$ are efficiently computable, and the precision needed for stable
orthogonalization is polynomially bounded, then $J_D$ and its extremal
eigenvalues are computable in polynomial time. The eigenvalue problem
also yields optimizing guides; recovering a particular eigenvector to a
prescribed accuracy requires the corresponding spectral conditioning.
Classical moment access is an additional requirement beyond efficient
source-state preparation.

For an independent-block source Hamiltonian, global moments are efficiently computable whenever the required individual block moments are efficiently computable. Suppose
\begin{equation}
 \mathcal H_A
 =
 \bigotimes_{j=1}^{N_{\mathrm{blk}}}\mathcal H_j,
 \qquad
 H_A
 =
 \sum_{j=1}^{N_{\mathrm{blk}}}h_j,
\end{equation}
where $h_j$ acts nontrivially only on block $j$ with product normalized trace
$\tau_A=\bigotimes_{j=1}^{N_{\mathrm{blk}}}\tau_j$, its moments are
\begin{equation}
\begin{aligned}
 m_k
 :=\tau_A(H_A^k)
 &=
 \tau_A\!\left[
 \left(\sum_{j=1}^{N_{\mathrm{blk}}}h_j\right)^k
 \right]
 \\
 &=
 \sum_{\substack{t_1,\ldots,t_{N_{\mathrm{blk}}}\ge0\\
                  t_1+\cdots+t_{N_{\mathrm{blk}}}=k}}
 \frac{k!}{t_1!\cdots t_{N_{\mathrm{blk}}}!}
 \tau_A\!\left(
 \prod_{j=1}^{N_{\mathrm{blk}}}h_j^{t_j}
 \right)
 \\
 &=
 \sum_{\substack{t_1,\ldots,t_{N_{\mathrm{blk}}}\ge0\\
                  t_1+\cdots+t_{N_{\mathrm{blk}}}=k}}
 \frac{k!}{t_1!\cdots t_{N_{\mathrm{blk}}}!}
 \prod_{j=1}^{N_{\mathrm{blk}}}\tau_j(h_j^{t_j})
 \\
 &=
 k!\,[z^k]
 \sum_{t_1,\ldots,t_{N_{\mathrm{blk}}}=0}^{k}
 \left(
 \prod_{j=1}^{N_{\mathrm{blk}}}
 \frac{\tau_j(h_j^{t_j})}{t_j!}
 \right)
 z^{t_1+\cdots+t_{N_{\mathrm{blk}}}}
 \\
 &=
 k!\,[z^k]
 \prod_{j=1}^{N_{\mathrm{blk}}}
 \left(
  \sum_{t=0}^{k}
  \frac{\tau_j(h_j^t)}{t!}z^t
 \right).
 \label{eq:value-independent-block-moments}
\end{aligned}
\end{equation}
Here $[z^k]F(z)$ denotes the coefficient of $z^k$ in the formal
polynomial $F(z)$.  Expanding the product shows that this operation
selects the tuples $(t_1,\ldots,t_{N_{\mathrm{blk}}})$ with
$t_1+\cdots+t_{N_{\mathrm{blk}}}=k$.

\subsection{Low-degree expectations versus sampling and state preparation}
\label{subsec:value-observable-limitation:details}
We prove the source-only evaluation and stability results used in
Section~\ref{subsec:value-observable-limitation}.

\begin{corollary}[Source-only evaluation of target observables]
\label{cor:value-source-only-observables}
Let $X_A=q(a)\ne0$ have marked degree at most $D$, set
$X_B=\pi(X_A)$, and define
\begin{equation}
 \rho_L^q:=\frac{X_LX_L^*}{\operatorname{Tr}(X_LX_L^*)},
 \qquad
 L=A,B.
\end{equation}
Let $O_A=O_A^*\in\mathcal F_r^A$ be a given lift of
$O_B=\pi(O_A)$, and assume $2D+r<d_{\mathrm{rel}}$.  Then
\begin{equation}
 \operatorname{Tr}(O_B\rho_B^q)
 =\operatorname{Tr}(O_A\rho_A^q)
 =\langle\psi_A^q|(O_A\otimes I)|\psi_A^q\rangle.
 \label{eq:value-source-only-observables}
\end{equation}
Consequently, any source algorithm estimating the normalized source
expectation to additive error $\epsilon$ and failure probability at most
$\delta$ also estimates the target expectation with the same guarantees.
No source-to-target label map, relative decoder, or inverse target
operator--Fourier transform is needed after the lift is known.
\end{corollary}

\begin{proof}
Theorem~\ref{thm:value-relative-expectation:details} identifies the normalized
source and target quotients.  Cyclicity of the physical traces gives the
first equality in Eq.~\eqref{eq:value-source-only-observables}, while
vectorization gives the second.
\end{proof}

For finite distance and $2D<d_{\mathrm{rel}}$, the uniform observable
budget is
\begin{equation}
 0\le r\le d_{\mathrm{rel}}-2D-1.
 \label{eq:value-source-observable-budget}
\end{equation}
In particular, the Hamiltonian expectation corresponds to $r=1$ and
requires $2D+1<d_{\mathrm{rel}}$.

The following proposition also shows the stability of the source-only evaluation.

\begin{proposition}[Stability of source-only expectation estimation]
\label{prop:value-source-estimation-stability}
Assume the hypotheses of
Corollary~\ref{cor:value-source-only-observables}.  Let
\begin{equation}
 \sigma_A^q
 :=
 |\psi_A^q\rangle\langle\psi_A^q|,
 \qquad
 \rho_A^q
 :=
 \operatorname{Tr}_{\mathrm{ref}}\sigma_A^q
\end{equation}
be the ideal full source operator state and its physical marginal.
Suppose the actually prepared full source state $\widehat\sigma_A$
satisfies
\begin{equation}
 \frac12
 \left\|
 \widehat\sigma_A-\sigma_A^q
 \right\|_1
 \le\eta_{\mathrm{src}},
 \label{eq:value-full-source-error}
\end{equation}
and define
\begin{equation}
 \widetilde\rho_A
 :=
 \operatorname{Tr}_{\mathrm{ref}}\widehat\sigma_A.
\end{equation}

Suppose a source-only estimator outputs $\widehat e$ such that,
\begin{equation}
 \left|
 \widehat e-\operatorname{Tr}(O_A\widetilde\rho_A)
 \right|
 \le\epsilon.
\end{equation}
Then,
\begin{equation}
 \left|
 \widehat e-\operatorname{Tr}(O_B\rho_B^q)
 \right|
 \le
 \epsilon+2\|O_A\|\eta_{\mathrm{src}}.
 \label{eq:value-source-estimation-stability}
\end{equation}
Thus estimating the ideal target expectation requires only preparation
and measurement of the source guide; the target state and the
source-to-target transduction circuit need not be implemented.
\end{proposition}

\begin{proof}
Contractivity of trace distance under partial trace~\cite{Watrous2018} gives
\[
 \frac12
 \left\|
 \widetilde\rho_A-\rho_A^q
 \right\|_1
 \le\eta_{\mathrm{src}}.
\]
By Corollary~\ref{cor:value-source-only-observables},
\[
 \operatorname{Tr}(O_A\rho_A^q)
 =
 \operatorname{Tr}(O_B\rho_B^q).
\]
Therefore, on the estimator's success event,
\begin{align*}
 \left|
 \widehat e-\operatorname{Tr}(O_B\rho_B^q)
 \right|
 &\le
 \left|
 \widehat e-\operatorname{Tr}(O_A\widetilde\rho_A)
 \right|\\
 &\quad+
 \left|
 \operatorname{Tr}\!\left[
 O_A(\widetilde\rho_A-\rho_A^q)
 \right]
 \right|\\
 &\quad+
 \left|
 \operatorname{Tr}(O_A\rho_A^q)
 -
 \operatorname{Tr}(O_B\rho_B^q)
 \right|.
\end{align*}
The first term is at most $\epsilon$, and the last term vanishes.
Hölder's inequality gives
\[
 \left|
 \operatorname{Tr}\!\left[
 O_A(\widetilde\rho_A-\rho_A^q)
 \right]
 \right|
 \le
 \|O_A\|
 \left\|
 \widetilde\rho_A-\rho_A^q
 \right\|_1
 \le
 2\|O_A\|\eta_{\mathrm{src}}.
\]
This proves Eq.~\eqref{eq:value-source-estimation-stability}.
\end{proof}

For completeness, the sampling bound in
Eq.~\eqref{eq:value-threshold-from-mean} follows from
$\mathbb E S\le c+(1-c)\Pr\{S>c\}$ for $S\in[0,1]$.

\section{Details of thermofield-double transduction}
\label{app:gibbs-details}
\label{sec:purified-gibbs-guides:details}

This appendix gives details for Section~\ref{sec:purified-gibbs-guides}.
We derive a common source and target spectral error bound, prove the
transfer-error estimate, and obtain an inverse-temperature window. Throughout, $H_A=H_A^*\in\mathcal A$ has
marked degree at most one and $H_B:=\pi(H_A)$.

\subsection{Spectral containment and a common error bound}
\label{subsec:gibbs-spectral-cover}
We first show why a source spectral error bound controls both
Hamiltonians. Let $H_A=H_A^*\in\mathcal A$ and
$H_B=\pi(H_A)$ for the represented homomorphism of
Section~\ref{sec:main-framework}. Spectral containment uses only the
unital $*$-homomorphism.

\begin{proposition}[Spectral containment and inherited filters]
\label{prop:mf-spectral-containment}
For every $X\in\mathcal A$,
\begin{equation}
 \operatorname{spec}(\pi(X))\subseteq\operatorname{spec}(X).
 \label{eq:mf-spectral-containment}
\end{equation}
In particular, if $H_A=H_A^*$ and $H_B=\pi(H_A)$, write
$H_A=\sum_{E\in\operatorname{spec}(H_A)}E P_E^A$. Then
\begin{equation}
 \begin{gathered}
 P_E^B:=\pi(P_E^A),\qquad
 H_B=\sum_{E\in\operatorname{spec}(H_A)}E P_E^B,\\
 \operatorname{spec}(H_B)
   =\{E\in\operatorname{spec}(H_A):P_E^B\ne0\}.
 \end{gathered}
 \label{eq:mf-spectral-survival}
\end{equation}
Here $P_E^B=0$ is allowed; otherwise it is the target spectral projection
at $E$. For every function $f:\operatorname{spec}(H_A)\to\mathbb C$,
\begin{equation}
 \pi(f(H_A))=f(H_B),
 \label{eq:mf-functional-calculus}
\end{equation}
where $f$ is restricted to the target spectrum on the right.
\end{proposition}
\noindent The proof is given in Appendix~\ref{app:spectral-properties}.

Thus the target can lose energies, but cannot acquire an energy absent
from the source. In particular,
\begin{equation}
 \lambda_{\min}(H_A)\le\lambda_{\min}(H_B)
 \le\lambda_{\max}(H_B)\le\lambda_{\max}(H_A).
 \label{eq:mf-spectral-edge-bounds}
\end{equation}

\paragraph{Passing a uniform approximation error to the target.}
For inverse temperature $\beta\ge0$ and a real polynomial $Q$, spectral
containment and functional calculus imply
\begin{equation}
 \begin{aligned}
 \bigl\|e^{\beta H_B/2}Q(H_B)-I_B\bigr\|
 &=\max_{E\in\operatorname{spec}(H_B)}
           \bigl|e^{\beta E/2}Q(E)-1\bigr|\\
 &\le\max_{E\in\operatorname{spec}(H_A)}
           \bigl|e^{\beta E/2}Q(E)-1\bigr|
 =\bigl\|e^{\beta H_A/2}Q(H_A)-I_A\bigr\|.
 \end{aligned}
 \label{eq:gibbs-inherited-relative-error:details}
\end{equation}
Thus a uniform relative-amplitude bound on an interval containing the
source spectrum automatically holds on the target spectrum. This is
why one polynomial approximation controls both errors in
Proposition~\ref{prop:gibbs-certificate:details}.

\subsection{Transduction of the purified Gibbs state}
\label{subsec:gibbs-interface:details}

Let $H_A=H_A^*\in\mathcal A$ and $H_B:=\pi(H_A)$, with $H_A$ of marked degree at most one. For $L\in\{A,B\}$
and inverse temperature $\beta\ge0$, define
\begin{equation}
 Z_L(\beta):=\operatorname{Tr}(e^{-\beta H_L}),\qquad
 |\mathrm{TFD}_L(\beta)\rangle
 :=\frac{(e^{-\beta H_L/2}\otimes I)|\Phi_L\rangle}
 {\sqrt{\tau_L(e^{-\beta H_L})}}.
 \label{eq:gibbs-tfd:details}
\end{equation}
Here $|\Phi_L\rangle$ and $\tau_L$ use the physical representation
and fixed reference basis from the framework. In an eigenbasis,
\begin{equation}
 |\mathrm{TFD}_L(\beta)\rangle
 =\frac1{\sqrt{Z_L(\beta)}}
 \sum_j e^{-\beta E_j^L/2}
 |E_j^L\rangle|\overline{E_j^L}\rangle,
 \qquad
 \rho_{L,\beta}:=
 \operatorname{Tr}_{\mathrm{ref}}\!
 \bigl(|\mathrm{TFD}_L(\beta)\rangle
       \langle\mathrm{TFD}_L(\beta)|\bigr)
 =\frac{e^{-\beta H_L}}{Z_L(\beta)}.
 \label{eq:gibbs-marginal}
\end{equation}
We require efficient coherent preparation of this canonical purification
to the accuracy in Proposition~\ref{prop:gibbs-certificate:details}, using
the reference convention of Eq.~\eqref{eq:gibbs-tfd:details}.

Functional calculus gives
$\pi(e^{-\beta H_A/2})=e^{-\beta H_B/2}$. To implement an approximate
transfer, fix a degree $D$ satisfying the compatibility and efficient
circuit assumptions of Theorem~\ref{thm:mf-polynomial-transduction},
with $2D<d_{\mathrm{rel}}$. Denote its fixed unitary realization on a
common enlarged input/output space by $U_D$. Its action outside the
degree-$D$ subspace is arbitrary but fixed.
With the input and output embeddings understood, it obeys
\begin{equation}
 U_D\bigl(|\psi_A^Q\rangle|0\rangle_J\bigr)
 =|\psi_B^Q\rangle|0\rangle_J,
 \qquad
 |\psi_L^Q\rangle:=
 \frac{(Q(H_L)\otimes I)|\Phi_L\rangle}
 {\sqrt{\tau_L(Q(H_L)^2)}},
 \label{eq:gibbs-polynomial-map:details}
\end{equation}
for real polynomials $Q$ of degree at most $D$ with nonzero target
image.

\begin{proposition}[Gibbs transduction from a polynomial certificate]
\label{prop:gibbs-certificate:details}
Suppose a real polynomial $Q$ of degree at most $D$ satisfies
\begin{equation}
 \sup_{x\in I}
 \bigl|e^{\beta x/2}Q(x)-1\bigr|\le\eta<1,
 \label{eq:gibbs-relative-approximation:details}
\end{equation}
where $I$ contains $\operatorname{spec}(H_A)$ and therefore also
$\operatorname{spec}(H_B)$ by
Proposition~\ref{prop:mf-spectral-containment}. Let the prepared source
state $\widehat\rho_A$ satisfy
\begin{equation}
 \frac12\bigl\|\widehat\rho_A-
 |\mathrm{TFD}_A(\beta)\rangle
 \langle\mathrm{TFD}_A(\beta)|\bigr\|_1
 \le\delta_{\mathrm{src}}.
\end{equation}
Then the output
\begin{equation}
 \widehat\omega_B:=U_D\bigl(\widehat\rho_A\otimes
 |0\rangle\langle0|_J\bigr)U_D^*
\end{equation}
obeys
\begin{equation}
 \frac12\bigl\|\widehat\omega_B-
 |\mathrm{TFD}_B(\beta),0\rangle
 \langle\mathrm{TFD}_B(\beta),0|\bigr\|_1
 \le\delta_{\mathrm{src}}+2\eta.
 \label{eq:gibbs-transfer-error:details}
\end{equation}
\end{proposition}
\begin{proof}

For $L\in\{A,B\}$, let
\[
 p_{L,j}:=\frac{e^{-\beta E_{L,j}}}{Z_L(\beta)},
 \qquad
 r_{L,j}:=e^{\beta E_{L,j}/2}Q(E_{L,j}).
\]
By Eq.~\eqref{eq:gibbs-relative-approximation:details},
$r_{L,j}\in[1-\eta,1+\eta]$. Since
$Q(E_{L,j})=e^{-\beta E_{L,j}/2}r_{L,j}$, the normalized polynomial
state has amplitudes
$\sqrt{p_{L,j}}\,r_{L,j}/\sqrt{\nu_L}$ in the same
energy-purification basis as the TFD, where
\[
 \mu_L:=\sum_j p_{L,j}r_{L,j},
 \qquad
 \nu_L:=\sum_j p_{L,j}r_{L,j}^2.
\]
Hence
\[
 |\langle\mathrm{TFD}_L(\beta)|\psi_L^Q\rangle|^2
 =\frac{\mu_L^2}{\nu_L}.
\]
Moreover,
\[
 (r_{L,j}-(1-\eta))(r_{L,j}-(1+\eta))\le0
\]
Expanding this
gives
$r_{L,j}^2\le2r_{L,j}-(1-\eta^2)$.
Averaging with respect to $p_{L,j}$ then yields
\[
 \nu_L\le2\mu_L-(1-\eta^2).
\]
Since $1-\eta^2>0$, multiplying by $1-\eta^2$ gives
\[
 (1-\eta^2)\nu_L
 \le
 2(1-\eta^2)\mu_L-(1-\eta^2)^2.
\]
Therefore,
\begin{align*}
 \mu_L^2-(1-\eta^2)\nu_L
 &\ge
 \mu_L^2-2(1-\eta^2)\mu_L+(1-\eta^2)^2\\
 &=
 \bigl(\mu_L-(1-\eta^2)\bigr)^2
 \ge0.
\end{align*}
Hence, since $\nu_L>0$,
\[
 \frac{\mu_L^2}{\nu_L}\ge 1-\eta^2.
\]
Thus
\[
 |\langle\mathrm{TFD}_L(\beta)|\psi_L^Q\rangle|^2
 \ge1-\eta^2,
\]
so the trace distance between these two pure states is at most $\eta$~\cite{FuchsVanDeGraaf1999}.

By the source-preparation assumption, $\widehat\rho_A$ is therefore
within $\delta_{\mathrm{src}}+\eta$ of
$|\psi_A^Q\rangle\langle\psi_A^Q|$. Since $Q$ has degree at most $D$,
Eq.~\eqref{eq:gibbs-polynomial-map:details} maps
$|\psi_A^Q\rangle|0\rangle_J$ exactly to
$|\psi_B^Q\rangle|0\rangle_J$. Unitary invariance of trace distance
then gives an output within $\delta_{\mathrm{src}}+\eta$ of the target
polynomial state. Applying the same $\eta$ bound on the target side and
using the triangle inequality yields
\[
 \frac12\bigl\|\widehat\omega_B-
 |\mathrm{TFD}_B(\beta),0\rangle
 \langle\mathrm{TFD}_B(\beta),0|\bigr\|_1
 \le\delta_{\mathrm{src}}+2\eta.
\]
The same bound holds after tracing out the reference and workspace by
contractivity of trace distance~\cite{Watrous2018}.
\end{proof}

\subsection{An explicit degree bound}
\label{subsec:gibbs-degree}
\label{app:gibbs-degree}

Choose an interval $[E_-,E_+]$ containing
$\operatorname{spec}(H_A)$, and write
\begin{equation}
 W:=E_+-E_-,\qquad \lambda:=\frac{\beta W}{2}.
 \label{eq:gibbs-bandwidth}
\end{equation}
By Proposition~\ref{prop:mf-spectral-containment},
$\operatorname{spec}(H_B)\subseteq\operatorname{spec}(H_A)$,
so this interval also contains the target spectrum.

If $W=0$ or $\beta=0$, a constant polynomial gives zero approximation
error. For $W,\beta>0$, construct $Q_D$ by the positive Taylor truncation about the
upper spectral edge,
\begin{equation}
 Q_D(x):=e^{-\beta E_+/2}
 \sum_{j=0}^{D}\frac{\bigl[\beta(E_+-x)/2\bigr]^j}{j!}.
 \label{eq:gibbs-taylor}
\end{equation}
The scalar prefactor cancels in the normalized guiding state. For
$t:=\beta(E_+-x)/2\in[0,\lambda]$,
\begin{equation}
 e^{\beta x/2}Q_D(x)
 =e^{-t}\sum_{j=0}^{D}\frac{t^j}{j!}
 =\Pr\{\operatorname{Pois}(t)\le D\}.
 \label{eq:gibbs-poisson-identity}
\end{equation}
Monotonicity of the Poisson upper tail therefore gives the exact
interval approximation error
\begin{equation}
 \sup_{x\in[E_-,E_+]}
 \bigl|e^{\beta x/2}Q_D(x)-1\bigr|
 =\eta_D(\lambda)
 :=\Pr\{\operatorname{Pois}(\lambda)\ge D+1\}.
 \label{eq:gibbs-exact-tail}
\end{equation}
This tail can be evaluated directly to choose a degree or bound
the admissible inverse temperature.

\begin{proposition}[Sufficient Gibbs approximation degree]
\label{prop:gibbs-degree}
For $0<\eta<1$, let $\ell:=\ln(1/\eta)$. The polynomial
Eq.~\eqref{eq:gibbs-taylor} satisfies
Eq.~\eqref{eq:gibbs-relative-approximation:details} whenever
\begin{equation}
\quad
 D+1\ge\frac{\beta W}{2}
       +\sqrt{\beta W\ell}+\frac{\ell}{3}.
 \quad
 \label{eq:gibbs-degree-condition}
\end{equation}
In particular, a sufficient degree scales as
\begin{equation}
 D=\frac{\beta W}{2}
   +O\!\left(\sqrt{\beta W\ln(1/\eta)}+\ln(1/\eta)+1\right).
 \label{eq:gibbs-degree-scaling}
\end{equation}
\end{proposition}

\begin{proof}
The case $\lambda=0$ is exact. Assume $\lambda>0$ and let
$X\sim\operatorname{Pois}(\lambda)$. For $0<\theta<3$,
\[
 \ln\mathbb E e^{\theta(X-\lambda)}
 =\lambda(e^\theta-1-\theta).
\]
Using
\[
 e^\theta-1-\theta
 =\sum_{j=2}^\infty \frac{\theta^j}{j!}
 \le
 \frac{\theta^2}{2}
 \sum_{j=0}^\infty\left(\frac{\theta}{3}\right)^j
 =
 \frac{\theta^2}{2(1-\theta/3)},
\]
where $j!\ge 2\cdot 3^{j-2}$ for $j\ge2$, we obtain
\[
 \ln\mathbb E e^{\theta(X-\lambda)}
 \le
 \frac{\lambda\theta^2}{2(1-\theta/3)}.
\]

Hence, for any $a>0$, exponential Markov~\cite{BoucheronEtAl2013} gives
\[
 \Pr\{X-\lambda\ge a\}
 \le
 \exp\!\left(
 -\theta a
 +\frac{\lambda\theta^2}{2(1-\theta/3)}
 \right).
\]
Set
\[
 \ell:=\ln(1/\eta),
 \qquad
 y:=\sqrt{\frac{2\ell}{\lambda}},
 \qquad
 \theta:=\frac{y}{1+y/3}.
\]
With
\[
 a:=\sqrt{2\lambda\ell}+\frac{\ell}{3},
\]
one checks that
\[
 -\theta a
 +\frac{\lambda\theta^2}{2(1-\theta/3)}
 =-\ell.
\]
Therefore
\[
 \Pr\!\left\{
 X-\lambda\ge
 \sqrt{2\lambda\ell}+\frac{\ell}{3}
 \right\}
 \le e^{-\ell}
 =\eta.
\]

Thus, if
\[
 D+1
 \ge
 \lambda+\sqrt{2\lambda\ell}+\frac{\ell}{3},
\]
then
\[
 \Pr\{X\ge D+1\}\le\eta.
\]
Using
\[
 \lambda=\frac{\beta W}{2},
 \qquad
 \sqrt{2\lambda\ell}=\sqrt{\beta W\ell},
\]
and Eq.~\eqref{eq:gibbs-exact-tail} gives
Eq.~\eqref{eq:gibbs-degree-condition}. The scaling
in Eq.~\eqref{eq:gibbs-degree-scaling} follows immediately.
\end{proof}

\subsection{Inverse temperature from the relative-decoding degree}
\label{subsec:gibbs-temperature}
We give the calculation behind Corollary~\ref{thm:gibbs-temperature}.

\begin{corollary}[Inverse-temperature window from relative decoding]
\label{thm:gibbs-temperature:details}
Assume the canonical source TFD can be prepared to the accuracy below,
and that the compatibility and efficient circuit assumptions of
Theorem~\ref{thm:mf-polynomial-transduction} hold at degree
$D_{\mathrm{rel}}$, with
\begin{equation}
 D_{\mathrm{rel}}
 \le
 \left\lfloor\frac{d_{\mathrm{rel}}-1}{2}\right\rfloor.
 \label{eq:gibbs-relative-degree}
\end{equation}
Fix a desired target trace-distance error $0<\varepsilon<1$ and a
source error $0\le\delta_{\mathrm{src}}<\varepsilon$, and define
\begin{equation}
 \eta:=\frac{\varepsilon-\delta_{\mathrm{src}}}{2},
 \qquad
 \ell:=\ln\frac{2}{\varepsilon-\delta_{\mathrm{src}}},
 \qquad
 k:=D_{\mathrm{rel}}+1.
 \label{eq:gibbs-error-budget}
\end{equation}
Let $[E_-,E_+]$ be a common spectral interval for $H_A$ and
$H_B$, with width $W:=E_+-E_->0$.

If $k\ge\ell/3$, then the target TFD and its Gibbs marginal can be prepared to trace
distance at most $\varepsilon$ whenever
\begin{equation}
 0\le\beta\le\frac{2}{W}
 \left(
 \sqrt{k+\frac{\ell}{6}}
 -
 \sqrt{\frac{\ell}{2}}
 \right)^2.
 \label{eq:gibbs-temperature-window:details}
\end{equation}

Suppose moreover that
\begin{equation}
 H_L=\sum_{i=1}^{m}h_i a_i^L,
 \qquad
 \|a_i^L\|\le1,
 \qquad
 |h_i|\le J,
\end{equation}
with $J>0$, so that one can choose $W\le2Jm$, and assume a linear relative distance
\begin{equation}
 d_{\mathrm{rel}}\ge\delta_{\mathrm{rel}}m
\end{equation}
with efficient relative decoding reaching the full radius in
Eq.~\eqref{eq:gibbs-relative-degree}. If $\ell=o(m)$, then every fixed
\begin{equation}
 \beta<\frac{\delta_{\mathrm{rel}}}{2J}
 \label{eq:gibbs-constant-temperature:details}
\end{equation}
is attainable for sufficiently large $m$, subject to the same efficient
source-preparation and transduction assumptions.
\end{corollary}
\begin{proof}
By Proposition~\ref{prop:gibbs-certificate:details}, it suffices to choose the
polynomial approximation error as
\[
 \eta=\frac{\varepsilon-\delta_{\mathrm{src}}}{2}.
\]
Then Proposition~\ref{prop:gibbs-degree}, with
$\ell=\ln(1/\eta)$, shows that a polynomial of degree at most
$D_{\mathrm{rel}}$ is sufficient whenever
\[
 k
 \ge
 \frac{\beta W}{2}
 +\sqrt{\beta W\ell}
 +\frac{\ell}{3}.
\]
Set $\lambda:=\beta W/2$. This becomes
\[
 \lambda+\sqrt{2\lambda\ell}+\frac{\ell}{3}\le k.
\]
Writing $y=\sqrt{\lambda}$ and solving the resulting quadratic
inequality gives, for $k\ge\ell/3$,
\[
 y\le
 \sqrt{k+\frac{\ell}{6}}
 -
 \sqrt{\frac{\ell}{2}}.
\]
Since $\lambda=\beta W/2$, this gives
Eq.~\eqref{eq:gibbs-temperature-window:details}.

For the extensive case, full-radius relative decoding gives
\[
 k
 =
 \left\lceil\frac{d_{\mathrm{rel}}}{2}\right\rceil
 \ge
 \frac{\delta_{\mathrm{rel}}m}{2}.
\]
Together with $W\le2Jm$, Eq.~\eqref{eq:gibbs-temperature-window:details}
therefore gives the sufficient window
\[
 \beta
 \le
 \frac{1}{Jm}
 \left(
 \sqrt{\frac{\delta_{\mathrm{rel}}m}{2}
       +\frac{\ell}{6}}
 -
 \sqrt{\frac{\ell}{2}}
 \right)^2.
\]
If $\ell=o(m)$, the right-hand side converges to
$\delta_{\mathrm{rel}}/(2J)$. Hence every fixed
$\beta<\delta_{\mathrm{rel}}/(2J)$ is attainable for all sufficiently
large $m$.
\end{proof}

\section{Proofs of the homomorphism criteria}
\label{app:homomorphism-criteria}

In this appendix, we prove the homomorphism criteria used for the
Pauli, fermionic, and prime-dimensional Weyl realizations. The three
proofs have the same structure. The pairwise commutation data determine
how arbitrary words are reordered, while the relation code determines
which ordered words become scalar. If every source scalar relation,
including its phase, remains valid in the target, then every polynomial
identity satisfied by the source generators is also satisfied by the
target generators. This is the condition required for the
generator assignment to descend to a $*$-homomorphism.

The relation code $K_L$ below is the kernel of the label matrix.
The kernel of the algebra homomorphism instead consists of source
operators annihilated by the map; additional target relations in
$K_B\setminus K_A$ can make this kernel nontrivial.

We first record the elementary algebraic principle used in all three
cases~\cite{Murphy1990}.

\begin{lemma}[Evaluation-kernel criterion]
\label{lem:app-evaluation-kernel}
Let $\mathcal P_s$ be the free unital complex $*$-algebra on
$x_1,\ldots,x_s$, where ``free'' means that the generators satisfy
no relations beyond the axioms of a unital $*$-algebra. Let
\begin{equation}
 \vartheta_A:\mathcal P_s\longrightarrow\mathcal A,
 \qquad
 \vartheta_B:\mathcal P_s\longrightarrow\mathcal B
\end{equation}
be the evaluation maps
\begin{equation}
 \vartheta_A(x_i)=a_i,\qquad
 \vartheta_B(x_i)=b_i,
\end{equation}
where
\begin{equation}
 \mathcal A=C^*(a_1,\ldots,a_s),
 \qquad
 \mathcal B=C^*(b_1,\ldots,b_s).
\end{equation}
If
\begin{equation}
 \ker\vartheta_A\subseteq\ker\vartheta_B,
 \label{eq:app-evaluation-kernel-inclusion}
\end{equation}
then
\begin{equation}
 \pi\bigl(\vartheta_A(q)\bigr):=\vartheta_B(q),
 \qquad q\in\mathcal P_s,
 \label{eq:app-induced-homomorphism}
\end{equation}
defines a unique surjective unital $*$-homomorphism
$\pi:\mathcal A\to\mathcal B$ satisfying $\pi(a_i)=b_i$.
\end{lemma}

\begin{proof}
We first show that Eq.~\eqref{eq:app-induced-homomorphism} is
well defined. Suppose that two polynomials $q,q'\in\mathcal P_s$
represent the same source operator:
\[
 \vartheta_A(q)=\vartheta_A(q').
\]
Then
\[
 \vartheta_A(q-q')=0,
\]
so $q-q'\in\ker\vartheta_A$. By
Eq.~\eqref{eq:app-evaluation-kernel-inclusion},
$q-q'\in\ker\vartheta_B$, and therefore
\[
 \vartheta_B(q)=\vartheta_B(q').
\]
Hence the value assigned by
Eq.~\eqref{eq:app-induced-homomorphism} depends only on the source
operator $\vartheta_A(q)$ and not on the polynomial used to represent
it.

Next, the map preserves the algebraic operations. For
$X=\vartheta_A(q)$ and $Y=\vartheta_A(r)$,
\begin{align*}
 \pi(X+Y)
 &=\pi\!\left(\vartheta_A(q+r)\right)
 =\vartheta_B(q+r)
 =\pi(X)+\pi(Y),\\
 \pi(XY)
 &=\pi\!\left(\vartheta_A(qr)\right)
 =\vartheta_B(qr)
 =\pi(X)\pi(Y),\\
 \pi(X^*)
 &=\pi\!\left(\vartheta_A(q^*)\right)
 =\vartheta_B(q^*)
 =\pi(X)^*.
\end{align*}
Likewise, for $\lambda\in\mathbb C$,
\[
 \pi(\lambda X)=\lambda\pi(X),
 \qquad
 \pi(I_A)=I_B.
\]
Thus $\pi$ is a unital $*$-homomorphism on the algebraic
$*$-algebra generated by $a_1,\ldots,a_s$.

In the present finite-dimensional setting, this algebraic
$*$-algebra is already all of $\mathcal A$. Indeed, it is a
finite-dimensional linear subspace of $\mathcal A$ and hence closed;
by definition its norm closure is
$C^*(a_1,\ldots,a_s)=\mathcal A$. Therefore
Eq.~\eqref{eq:app-induced-homomorphism} defines $\pi$ on all of
$\mathcal A$.

By construction,
\[
 \pi(a_i)
 =\pi\!\left(\vartheta_A(x_i)\right)
 =\vartheta_B(x_i)
 =b_i.
\]
Consequently, the image of $\pi$ contains $I_B$ and every generator
$b_i$. Since the image of a $*$-homomorphism between finite-dimensional
$C^*$-algebras is a $C^*$-subalgebra,
\[
 \operatorname{im}\pi
 \supseteq C^*(b_1,\ldots,b_s)
 =\mathcal B.
\]
Hence $\pi$ is surjective.
\end{proof}

\subsection{Pauli generators}
\label{app:homomorphism-pauli}

We use the signed Pauli generators of the main text,
\begin{equation}
 a_i=\eta_i^A P_{n_A}(w_i),
 \qquad
 b_i=\eta_i^B P_{n_B}(v_i),
 \qquad
 \eta_i^A,\eta_i^B\in\{\pm1\},
\end{equation}
and write
\begin{equation}
 A=[w_1\ \cdots\ w_s],
 \qquad
 B=[v_1\ \cdots\ v_s].
\end{equation}
Let $J_n$ denote the binary symplectic form~\cite{DehaeneDeMoor2003}. Define
\begin{equation}
 \Sigma_A:=A^{\mathsf T}J_{n_A}A,
 \qquad
 \Sigma_B:=B^{\mathsf T}J_{n_B}B,
\end{equation}
and
\begin{equation}
 K_A:=\ker A,
 \qquad
 K_B:=\ker B.
\end{equation}
For $c\in\mathbb F_2^s$, define the ordered words
\begin{equation}
 U_A(c):=a_1^{c_1}\cdots a_s^{c_s},
 \qquad
 U_B(c):=b_1^{c_1}\cdots b_s^{c_s}.
 \label{eq:app-pauli-ordered-words}
\end{equation}
If $c\in K_L$, the Pauli label of $U_L(c)$ vanishes and hence
\begin{equation}
 U_L(c)=\zeta_L(c)I
\end{equation}
for a scalar phase $\zeta_L(c)\in\mathbb T$.

\begin{proposition}[Pauli homomorphism criterion]
\label{prop:app-pauli-homomorphism}
The assignment
\begin{equation}
 a_i\longmapsto b_i
\end{equation}
extends to a surjective unital $*$-homomorphism
\begin{equation}
 \pi:\mathcal A\longrightarrow\mathcal B
\end{equation}
if
\begin{equation}
 A^{\mathsf T}J_{n_A}A
 =
 B^{\mathsf T}J_{n_B}B,
 \qquad
 K_A\subseteq K_B,
 \qquad
 \zeta_A(k)=\zeta_B(k)
 \quad\text{for all }k\in K_A.
 \label{eq:app-pauli-homomorphism-conditions}
\end{equation}
Conversely, these conditions are necessary for a homomorphism
satisfying $\pi(a_i)=b_i$.
\end{proposition}

\begin{proof}
We first prove sufficiency. Put
\[
 \Sigma:=
 A^{\mathsf T}J_{n_A}A
 =
 B^{\mathsf T}J_{n_B}B.
\]
Since every signed Pauli generator is a Hermitian involution,
\[
 a_i^*=a_i,\qquad a_i^2=I,
 \qquad
 b_i^*=b_i,\qquad b_i^2=I.
\]
Moreover,
\[
 a_i a_j=(-1)^{\Sigma_{ij}}a_j a_i,
 \qquad
 b_i b_j=(-1)^{\Sigma_{ij}}b_j b_i.
\]
Thus the source and target obey the same elementary
reordering rules.

For $c,d\in\mathbb F_2^s$, define
\[
 \kappa_\Sigma(c,d)
 :=
 (-1)^{\sum_{i>j}c_i d_j\Sigma_{ij}}.
\]
Moving the second ordered word through the first gives
\begin{equation}
 U_L(c)U_L(d)
 =
 \kappa_\Sigma(c,d)U_L(c+d),
 \qquad L\in\{A,B\}.
 \label{eq:app-pauli-word-product}
\end{equation}
The phase on the right depends only on the common
commutation matrix $\Sigma$.

Let $q\in\mathcal P_s$ be any noncommutative $*$-polynomial.
Using
\[
 x_i^*=x_i,\qquad x_i^2=1,
 \qquad
 x_ix_j=(-1)^{\Sigma_{ij}}x_jx_i
\]
to reduce each monomial to ordered form gives coefficients
$\alpha_c\in\mathbb C$, independent of $L$, such that
\begin{equation}
 \vartheta_L(q)
 =
 \sum_{c\in\mathbb F_2^s}\alpha_cU_L(c),
 \qquad L\in\{A,B\}.
 \label{eq:app-pauli-common-normal-form}
\end{equation}

Choose a set $\mathcal R_A\subseteq\mathbb F_2^s$ containing one
representative of each coset of $K_A$. Every $c$ has a unique
decomposition
\[
 c=r+k,
 \qquad
 r\in\mathcal R_A,\quad k\in K_A.
\]
By Eq.~\eqref{eq:app-pauli-word-product},
\[
 U_L(r)U_L(k)
 =
 \kappa_\Sigma(r,k)U_L(r+k).
\]
Since $K_A\subseteq K_B$, both $U_A(k)$ and $U_B(k)$ are scalar,
and the phase condition gives
\[
 U_A(k)=\zeta_A(k)I,
 \qquad
 U_B(k)=\zeta_B(k)I
       =\zeta_A(k)I.
\]
Hence
\[
 U_L(r+k)
 =
 \kappa_\Sigma(r,k)^{-1}\zeta_A(k)U_L(r),
 \qquad L\in\{A,B\}.
\]
Substituting this into
Eq.~\eqref{eq:app-pauli-common-normal-form} gives the same
coefficients $\beta_r$ in source and target:
\[
 \vartheta_A(q)=\sum_{r\in\mathcal R_A}\beta_rU_A(r),
 \qquad
 \vartheta_B(q)=\sum_{r\in\mathcal R_A}\beta_rU_B(r).
\]

If $r,r'\in\mathcal R_A$ are distinct, then
$Ar\ne Ar'$. Therefore $U_A(r)$ and $U_A(r')$ are nonzero scalar
multiples of distinct Pauli operators. They are orthogonal under
the normalized Hilbert--Schmidt inner product and are in particular
linearly independent. Consequently,
\[
 \vartheta_A(q)=0
 \quad\Longrightarrow\quad
 \beta_r=0\ \text{for every }r
 \quad\Longrightarrow\quad
 \vartheta_B(q)=0.
\]
Thus
\[
 \ker\vartheta_A\subseteq\ker\vartheta_B.
\]
Lemma~\ref{lem:app-evaluation-kernel} now gives the desired
surjective unital $*$-homomorphism.

For necessity, suppose such a homomorphism exists. Applying $\pi$
to
\[
 a_i a_j=(-1)^{(\Sigma_A)_{ij}}a_ja_i
\]
and comparing with the actual target commutation relation gives
\[
 (\Sigma_A)_{ij}=(\Sigma_B)_{ij}
\]
for every $i,j$. If $k\in K_A$, then
\[
 U_A(k)=\zeta_A(k)I.
\]
Applying $\pi$ gives
\[
 U_B(k)=\zeta_A(k)I.
\]
Therefore $Bk=0$, so $k\in K_B$, and
$\zeta_B(k)=\zeta_A(k)$. Hence
$K_A\subseteq K_B$ and the scalar phases agree on $K_A$.
\end{proof}

The proof also shows that there are no additional source relations
that need to be checked separately. Pairwise Pauli commutation,
together with the scalar relations indexed by $K_A$, generates all
polynomial relations among the marked source generators.

\subsection{Even Majorana generators}
\label{app:homomorphism-majorana}

Let
\begin{equation}
 \Gamma_L(u)
 :=
 i^{|u|(|u|-1)/2}
 \prod_{j:\,u_j=1}\gamma_{L,j},
 \qquad
 u\in\mathbb F_2^{2n_L},
\end{equation}
with the Majorana factors ordered by increasing index. We restrict
to even labels $u$, as appropriate for parity-preserving fermionic
Hamiltonians~\cite{BravyiKitaev2002,Bravyi2005}. The marked generators are
\begin{equation}
 a_i=\eta_i^A\Gamma_A(w_i),
 \qquad
 b_i=\eta_i^B\Gamma_B(v_i),
 \qquad
 \eta_i^A,\eta_i^B\in\{\pm1\}.
\end{equation}
Set
\begin{equation}
 A=[w_1\ \cdots\ w_s],
 \qquad
 B=[v_1\ \cdots\ v_s],
 \qquad
 K_A:=\ker A,\quad K_B:=\ker B.
\end{equation}

For even Majorana labels,
\begin{equation}
 \Gamma_L(u)\Gamma_L(v)
 =
 (-1)^{u^{\mathsf T}v}
 \Gamma_L(v)\Gamma_L(u).
 \label{eq:app-majorana-commutation}
\end{equation}
Indeed, the general commutation sign is
$(-1)^{|u||v|-u^{\mathsf T}v}$, and $|u||v|$ is even.
Therefore the marked commutation matrices are
\begin{equation}
 \Sigma_A=A^{\mathsf T}A,
 \qquad
 \Sigma_B=B^{\mathsf T}B.
\end{equation}

For $c\in\mathbb F_2^s$, define
\begin{equation}
 U_A(c):=a_1^{c_1}\cdots a_s^{c_s},
 \qquad
 U_B(c):=b_1^{c_1}\cdots b_s^{c_s}.
\end{equation}
If $c\in K_L$, the resulting Majorana label vanishes, so
\begin{equation}
 U_L(c)=\zeta_L(c)I.
\end{equation}

\begin{proposition}[Majorana homomorphism criterion]
\label{prop:app-majorana-homomorphism}
The assignment $a_i\mapsto b_i$ extends to a surjective unital
$*$-homomorphism if
\begin{equation}
 A^{\mathsf T}A=B^{\mathsf T}B,
 \qquad
 K_A\subseteq K_B,
 \qquad
 \zeta_A(k)=\zeta_B(k)
 \quad\text{for all }k\in K_A.
 \label{eq:app-majorana-homomorphism-conditions}
\end{equation}
These conditions are also necessary.
\end{proposition}

\begin{proof}
Each Hermitian Majorana monomial in our convention satisfies
\[
 \Gamma_L(u)^*=\Gamma_L(u),
 \qquad
 \Gamma_L(u)^2=I.
\]
Hence the signed generators satisfy
\[
 a_i^*=a_i,\quad a_i^2=I,
 \qquad
 b_i^*=b_i,\quad b_i^2=I.
\]
Under
$A^{\mathsf T}A=B^{\mathsf T}B=:\Sigma$, their pairwise
commutation relations are also identical:
\[
 a_i a_j=(-1)^{\Sigma_{ij}}a_ja_i,
 \qquad
 b_i b_j=(-1)^{\Sigma_{ij}}b_jb_i.
\]

As for the Pauli words,
\[
 U_L(c)U_L(d)
 =
 \kappa_\Sigma(c,d)U_L(c+d),
 \qquad
 \kappa_\Sigma(c,d)
 =
 (-1)^{\sum_{i>j}c_i d_j\Sigma_{ij}}.
\]
Every evaluated $*$-polynomial therefore has a common ordered-word
normal form
\[
 \vartheta_L(q)
 =
 \sum_{c\in\mathbb F_2^s}\alpha_cU_L(c),
 \qquad L\in\{A,B\},
\]
with the same coefficients $\alpha_c$ in source and target.

Choose a transversal $\mathcal R_A$ of
$\mathbb F_2^s/K_A$. For $c=r+k$ with
$r\in\mathcal R_A$ and $k\in K_A$,
\[
 U_L(c)
 =
 \kappa_\Sigma(r,k)^{-1}\zeta_L(k)U_L(r).
\]
The assumptions
$K_A\subseteq K_B$ and
$\zeta_A(k)=\zeta_B(k)$ therefore give
\[
 \vartheta_A(q)=\sum_{r\in\mathcal R_A}\beta_rU_A(r),
 \qquad
 \vartheta_B(q)=\sum_{r\in\mathcal R_A}\beta_rU_B(r)
\]
with identical coefficients $\beta_r$.

Distinct representatives satisfy $Ar\ne Ar'$. Thus
$U_A(r)$ and $U_A(r')$ are scalar multiples of distinct Majorana
monomials $\Gamma_A(Ar)$ and $\Gamma_A(Ar')$. Distinct Majorana
monomials are orthogonal under the normalized full-Fock trace.
Hence
\[
 \vartheta_A(q)=0
 \quad\Longrightarrow\quad
 \beta_r=0\ \text{for every }r
 \quad\Longrightarrow\quad
 \vartheta_B(q)=0.
\]
We again obtain
\[
 \ker\vartheta_A\subseteq\ker\vartheta_B,
\]
so Lemma~\ref{lem:app-evaluation-kernel} produces the desired
surjective unital $*$-homomorphism.

Conversely, any homomorphism mapping $a_i$ to $b_i$ must preserve
the pairwise commutation relations, which gives
$A^{\mathsf T}A=B^{\mathsf T}B$. If $k\in K_A$, then
$U_A(k)=\zeta_A(k)I$, and applying the homomorphism yields
\[
 U_B(k)=\zeta_A(k)I.
\]
Therefore $k\in K_B$ and
$\zeta_B(k)=\zeta_A(k)$.
\end{proof}

\subsection{Prime-dimensional Weyl generators in field-symbol coordinates}
\label{app:homomorphism-qudit}

Let $p$ be an odd prime and $\omega=e^{2\pi i/p}$. For
$u=(x,z)\in\mathbb F_p^{2m}$, use
\begin{equation}
 W_m(x,z):=\omega^{x\cdot z/2}X^xZ^z.
\end{equation}
With
\begin{equation}
 \Omega_m:=
 \begin{pmatrix}
  0&-I_m\\
  I_m&0
 \end{pmatrix},
 \qquad
 [u,v]:=u^{\mathsf T}\Omega_m v,
\end{equation}
the Weyl relations are~\cite{HostensEtAl2005,Gross2006}
\begin{equation}
 W_m(u)W_m(v)
 =\omega^{[u,v]/2}W_m(u+v),
 \qquad
 W_m(u)W_m(v)
 =\omega^{[u,v]}W_m(v)W_m(u).
 \label{eq:app-weyl-relations}
\end{equation}

Choose Weyl directions
\begin{equation}
 w_i\in\mathbb F_p^{2m_A},
 \qquad
 v_i\in\mathbb F_p^{2m_B},
 \qquad i\in[s],
\end{equation}
and phases $r_i^A,r_i^B\in\mathbb F_p$. For every field symbol
$\xi\in\mathbb F_p$, define
\begin{equation}
 a_i(\xi):=\omega^{\xi r_i^A}W_{m_A}(\xi w_i),
 \qquad
 b_i(\xi):=\omega^{\xi r_i^B}W_{m_B}(\xi v_i).
 \label{eq:app-qudit-local-symbols}
\end{equation}
Then
\begin{equation}
 \begin{aligned}
 a_i(\xi)a_i(\eta)&=a_i(\xi+\eta),
 &a_i(\xi)^*&=a_i(-\xi),\\
 b_i(\xi)b_i(\eta)&=b_i(\xi+\eta),
 &b_i(\xi)^*&=b_i(-\xi).
 \end{aligned}
 \label{eq:app-qudit-local-symbol-laws}
\end{equation}
Thus the powers along one Weyl direction form a single
$\mathbb F_p$-valued coordinate. These local symbol laws belong to
the reference presentation and are not regarded as additional
source relations.

Define
\begin{equation}
 \mathcal A
 :=C^*\!\left(I_A,a_i(\xi):
                  i\in[s],\ \xi\in\mathbb F_p\right),
 \qquad
 \mathcal B
 :=C^*\!\left(I_B,b_i(\xi):
                  i\in[s],\ \xi\in\mathbb F_p\right).
\end{equation}
Set
\begin{equation}
 A:=[w_1\ \cdots\ w_s],
 \qquad
 B:=[v_1\ \cdots\ v_s],
 \qquad
 K_A:=\ker A,
 \qquad
 K_B:=\ker B,
 \label{eq:app-qudit-label-matrices}
\end{equation}
where all kernels are over $\mathbb F_p$. For
$c=(c_1,\ldots,c_s)\in\mathbb F_p^s$, define the reduced ordered
field-symbol words
\begin{equation}
 U_A(c):=a_1(c_1)\cdots a_s(c_s),
 \qquad
 U_B(c):=b_1(c_1)\cdots b_s(c_s).
 \label{eq:app-qudit-ordered-word}
\end{equation}
Each $U_L(c)$ is a phase times the Weyl operator with label $Ac$
or $Bc$, respectively. Hence, for $c\in K_L$,
\begin{equation}
 U_L(c)=\zeta_L(c)I_L,
 \qquad
 \zeta_L(c)\in\{1,\omega,\ldots,\omega^{p-1}\}.
\end{equation}

The matrices
\begin{equation}
 \Gamma_A:=A^{\mathsf T}\Omega_{m_A}A,
 \qquad
 \Gamma_B:=B^{\mathsf T}\Omega_{m_B}B
\end{equation}
record the complete commutation phases between distinct
field-symbol coordinates and may be nonzero.

\begin{proposition}[Weyl homomorphism criterion]
\label{prop:app-qudit-homomorphism}
The assignment
\begin{equation}
 a_i(\xi)\longmapsto b_i(\xi),
 \qquad
 i\in[s],\quad \xi\in\mathbb F_p,
\end{equation}
extends to a surjective unital $*$-homomorphism
\begin{equation}
 \pi_p:\mathcal A\longrightarrow\mathcal B
\end{equation}
if and only if
\begin{equation}
 A^{\mathsf T}\Omega_{m_A}A
 =
 B^{\mathsf T}\Omega_{m_B}B,
 \qquad
 K_A\subseteq K_B,
 \qquad
 \zeta_A(k)=\zeta_B(k)
 \quad\text{for all }k\in K_A.
 \label{eq:app-qudit-homomorphism-conditions}
\end{equation}
Equivalently, it is enough to prescribe
$a_i(1)\mapsto b_i(1)$, since the local symbol laws determine the
images of all $a_i(\xi)$.
\end{proposition}

\begin{proof}
Put
\[
 \Gamma
 :=
 A^{\mathsf T}\Omega_{m_A}A
 =
 B^{\mathsf T}\Omega_{m_B}B.
\]
Eq.~\eqref{eq:app-weyl-relations} gives, for both
$L=A,B$,
\begin{equation}
 a_{L,i}(\xi)a_{L,j}(\eta)
 =
 \omega^{\xi\eta\Gamma_{ij}}
 a_{L,j}(\eta)a_{L,i}(\xi),
 \label{eq:app-qudit-common-commutation}
\end{equation}
where $a_{A,i}=a_i$ and $a_{B,i}=b_i$. Thus the source and target
obey the same local symbol, adjoint, and reordering rules.

For $c,d\in\mathbb F_p^s$, define
\[
 \kappa_\Gamma(c,d)
 :=
 \omega^{\sum_{i>j}c_i d_j\Gamma_{ij}}.
\]
Reordering the field-symbol words gives
\begin{equation}
 U_L(c)U_L(d)
 =
 \kappa_\Gamma(c,d)U_L(c+d),
 \qquad L\in\{A,B\}.
 \label{eq:app-qudit-word-product}
\end{equation}

Let $q\in\mathcal P_s$ be a noncommutative $*$-polynomial, with
$x_i$ evaluated as $a_i(1)$ or $b_i(1)$. Using the common local
symbol and reordering rules, every monomial reduces to an ordered
word, so
\begin{equation}
 \vartheta_L(q)
 =
 \sum_{c\in\mathbb F_p^s}\alpha_cU_L(c),
 \qquad L\in\{A,B\},
 \label{eq:app-qudit-common-normal-form}
\end{equation}
with the same coefficients $\alpha_c$ in source and target.

Choose a transversal $\mathcal R_A$ of
$\mathbb F_p^s/K_A$. Every $c$ has a unique decomposition
$c=r+k$, with $r\in\mathcal R_A$ and $k\in K_A$. From
Eq.~\eqref{eq:app-qudit-word-product},
\[
 U_L(r+k)
 =
 \kappa_\Gamma(r,k)^{-1}U_L(r)U_L(k).
\]
Since $K_A\subseteq K_B$ and the scalar phases agree on $K_A$,
\[
 U_A(k)=\zeta_A(k)I_A,
 \qquad
 U_B(k)=\zeta_B(k)I_B=\zeta_A(k)I_B.
\]
Therefore
\[
 U_L(r+k)
 =
 \kappa_\Gamma(r,k)^{-1}
 \zeta_A(k)U_L(r),
 \qquad L\in\{A,B\}.
\]
Grouping Eq.~\eqref{eq:app-qudit-common-normal-form} by
$K_A$-cosets gives
\[
 \vartheta_A(q)
 =
 \sum_{r\in\mathcal R_A}\beta_rU_A(r),
 \qquad
 \vartheta_B(q)
 =
 \sum_{r\in\mathcal R_A}\beta_rU_B(r),
\]
with identical coefficients $\beta_r$.

Distinct representatives satisfy $Ar\ne Ar'$, so
$U_A(r)$ and $U_A(r')$ are scalar multiples of distinct Weyl
operators. They are Hilbert--Schmidt orthogonal and therefore
linearly independent. Hence
\[
 \vartheta_A(q)=0
 \quad\Longrightarrow\quad
 \beta_r=0\ \text{for all }r
 \quad\Longrightarrow\quad
 \vartheta_B(q)=0.
\]
Thus
\[
 \ker\vartheta_A\subseteq\ker\vartheta_B,
\]
and Lemma~\ref{lem:app-evaluation-kernel} gives the required
surjective unital $*$-homomorphism.

Conversely, any such homomorphism must preserve
Eq.~\eqref{eq:app-qudit-common-commutation}. Since Weyl operators
are invertible, this forces
\[
 \Gamma_A=\Gamma_B
 \qquad\text{over }\mathbb F_p.
\]
If $k\in K_A$, then
\[
 U_A(k)=\zeta_A(k)I_A.
\]
Applying the homomorphism gives
\[
 U_B(k)=\zeta_A(k)I_B.
\]
A Weyl operator is scalar only when its phase-space label vanishes,
so $Bk=0$ and hence $k\in K_B$. Comparing scalar factors gives
$\zeta_B(k)=\zeta_A(k)$. This proves necessity.
\end{proof}

The local costs $\rho_i$ used to define the marked filtration do
not enter the homomorphism criterion. They determine only the word
metric and the resulting ordinary or relative distance.

\section{Numerical instances, methods, and supplementary results}
\label{app:chain-benchmark}
\label{app:additional-numerics}

This appendix specifies the nonlinear residual-chain ensemble and the
comparison protocol for Section~\ref{subsec:additional-numerics}.

\paragraph{Reproducibility package.}
Release v1.0.0 of \texttt{RDQINumerics}~\cite{RDQINumerics2026}
contains all 20 test targets, their generated arrays and seeds, all
880 final polynomial witnesses, the result tables, and the frozen
120-batch execution plan. After installing the pinned dependencies,
\texttt{python reproduce.py paper --all-runs} independently rescores
the witnesses, checks the quantum prediction using 60-digit arithmetic,
and reconstructs the reported table. The README gives commands for
the tests, a small smoke benchmark, and a resumable full rerun; new
runs are stored separately from the reference results.

\subsection{Nonlinear residual-chain instances and quantum prediction}
\label{app:additional-instances}
Let us recall and discuss in more detail the specifications of our instances.
Write $x_v=P_a(\alpha_v)$ at $L$ distinct nonzero field points,
where $P_a(X)=\sum_{j=0}^{n-1}a_jX^j$. Let $T(x)$ be the number of
satisfied predicates, so that $\operatorname{Sat}(P_a)=T(Va)/M$ in
Eq.~\eqref{eq:numerical-satisfaction}. Each reward depends on two
polynomial evaluations, in contrast to the symbolwise constraints of
list recovery and weighted reconstruction~\cite{GuruswamiSudan1999}.

Over $\F_p$ with $p$ an odd prime, take $L=p-1$ and $M=L-1$.
Independently choose $a_i^{\rm res}$ uniformly from $\F_p\setminus\{0\}$
and $F_i$ uniformly among subsets of size $(p-1)/2$. Define
\begin{equation}
 T(x)=\sum_{i=1}^{M}
 \one\{x_{i+1}+a_i^{\rm res}x_i^2\in F_i\}.
 \label{eq:additional-residual-score}
\end{equation}
The vector $(a_i^{\rm res})_i$ specifies a structure; the sets $F_i$
specify its predicate target. The map
$(x_1,\ldots,x_L)\mapsto(x_1,r_1,\ldots,r_M)$ with
$r_i=x_{i+1}+a_i^{\rm res}x_i^2$ is bijective, with inverse
$x_{i+1}=r_i-a_i^{\rm res}x_i^2$. Uniform source symbols therefore
give $T\sim\operatorname{Bin}(M,\rho)$, where $\rho=(p-1)/(2p)$.

\paragraph{Guides and degree windows.}
For $T\sim\operatorname{Bin}(M,\rho)$, the orthonormal-polynomial
Jacobi matrix for multiplication by $T$ has entries
\begin{equation}
 (J_\rho)_{jj}=M\rho+(1-2\rho)j,\qquad
 (J_\rho)_{j,j+1}=\sqrt{\rho(1-\rho)(j+1)(M-j)}.
 \label{eq:additional-jacobi}
\end{equation}
The largest eigenvalue of the compression to degrees $0,\ldots,D$
is $M\mu_Q$; a unit top eigenvector defines the source-normalized
filter $p_D(T)$. Its guiding state is prepared from signed binomial
shells: prepare the satisfaction pattern coherently, sample each
residual from $F_i$ or its complement according to that pattern,
prepare $x_1$ uniformly, and invert the residual map. This is the
construction in Lemma~\ref{app:residual-source-lemma}.

Uniform degree-below-$n$ polynomials have independent uniform
evaluations at any $n$ distinct points. Since each predicate touches
two vertices, normalization, mean, and second-moment transfer hold,
respectively, when
\begin{equation}
 4D\le n,\qquad 4D+2\le n,\qquad 4D+4\le n.
 \label{eq:additional-transfer-windows}
\end{equation}
The benchmark parameters satisfy all three conditions:
\begin{equation}
 p=2053,\qquad L=2052,\qquad M=2051,\qquad n=205,\qquad D=50.
 \label{eq:additional-cells}
\end{equation}
The measured polynomial coefficients have probability proportional to
$p_D(T(Va))^2$. Source preparation and coherent decoding produce the
score bias while enforcing global code consistency. For comparison, a
uniform feasible polynomial has mean satisfaction
$\rho=1026/2053\simeq0.49976$. The numerical comparison asks how
closely classical search can approach the filtered mean
$\mu_Q\simeq0.643134$ within the reported budgets.
Independent Jacobi reconstruction and recurrence evaluation, including
60-digit precision controls, check the mean in
Table~\ref{tab:additional-numerics}. These moment calculations do not
require a quantum-circuit simulation or identification of the full
target score distribution with the filtered source distribution.

\subsection{Classical comparison protocol}
\label{app:additional-protocol}

The test set consists of ten independently sampled residual-coefficient
vectors, two independently seeded predicate targets per vector, and two
solver streams per target. Evaluation-point lists are fixed. Development
instances used to choose the methods and budgets are separate from
these test instances.

Uniform search samples the polynomial coefficients independently and
uniformly. Interpolation through fixed and moving information sets
adapts information-set search~\cite{Prange1962} and its use in
OPI optimization~\cite{JordanEtAlDQI,opi-optimized}.
Greedy coefficient and information-coordinate updates and bounded block
search follow coordinate and block local-improvement
templates~\cite[Secs.~2.5 and~5.4]{Besag1986}. The methods also include
simulated annealing~\cite{KirkpatrickGelattVecchi1983} and replica
exchange~\cite{HukushimaNemoto1996}. The benchmark-specific
implementations are documented in Ref.~\cite{RDQINumerics2026}.

Structure-aware methods add residual-coordinate moves that preserve the
satisfied predicates on a selected interval~\cite{RDQINumerics2026},
and FFT-accelerated forward--backward (sum--product) inference for an anchored
auxiliary chain distribution~\cite{KschischangFreyLoeliger2001,RDQINumerics2026}. Inference is followed by
heuristic projection to a Reed--Solomon codeword; only the resulting
feasible score is counted.
This inference-to-codeword architecture is related to soft-decision
Reed--Solomon decoding on channels with memory, where symbol reliabilities
from sum--product or forward--backward inference feed an algebraic
decoder~\cite[Sec.~VII]{KoetterVardy2003}.
Runs from fresh initial states are supplemented by refinements of the
best assignments found by information-set methods.

Guide-weight Gibbs and replica-exchange methods use the actual weight
$w(a)=|p_D(T(Va))|^2$, including its zeros, with an interpolating
bridge $w_\eta=\eta+(1-\eta)w$, as implemented in
Ref.~\cite{RDQINumerics2026}. Gil-Fuster et al.~\cite{GilFusterDQIMCMC}
showed that block-Gibbs methods can reach DQI's predicted scores on
their tested instances, motivating their inclusion here. Every scored
proposal counts toward the candidate budget and can improve the best
score found, including rejected MCMC proposals. These runs test
optimization performance; distributional accuracy is a separate sampling
question.

\begin{table}[htbp]
\centering\small
\begin{tabular}{@{}lrr@{}}
\toprule
Method & Cost/candidate & Budget/run\\
\midrule
Uniform feasible-polynomial sampling & $O(Ln)$ & $200{,}000$\\
Fixed-window interpolation & $O(Ln)$ & $500{,}000$\\
Moving-window interpolation & $O(Ln)$ & $B$\\
Residual-preserving greedy search & $O(Ln)$ & $B$\\
Information-coordinate greedy search & $O(L)$ & $B$\\
Information-coordinate annealing & $O(L)$ & $B$\\
Polynomial-coefficient greedy search & $O(L)$ & $B$\\
Two-coordinate block greedy search & $O(L)$ & $B$\\
Two-coordinate block annealing & $O(L)$ & $B$\\
Relaxed-chain repair and projection & $O(Ln)$ & $B$\\
FFT inference and soft RS projection & $O(Ln)^{*}$ & $B$\\
Guide-weight coordinate Gibbs & $O(L)$ & $B$\\
Guide-weight replica exchange & $O(L)$ & $B$\\
\bottomrule
\end{tabular}
\caption{Tested classical methods, with $B=4{,}999{,}056$.
Costs count leading arithmetic operations for candidate construction
and scoring after preprocessing, excluding witness audits; they are
not convergence guarantees. The residual-preserving bound is
worst-case. A full coordinate scan evaluates $p$ candidates and
costs $O(pL)$.
${}^{*}$Inference additionally uses four rounds with main FFT work
$O(Lp\log p)$ per round.
Nine methods also have warm-start runs initialized by moving-window
interpolation; configurations are specified in
Ref.~\cite{RDQINumerics2026}.}
\label{tab:classical-method-budgets}
\end{table}

Run budgets are listed in Table~\ref{tab:classical-method-budgets};
both projection methods received the full deep-search budget.
Candidate counts,
final scores, timing, and preprocessing costs were recorded. All 120
batches of nonlinear-chain runs completed. The reported best assignments
were independently rescored from their polynomial coefficients.

For residual-coefficient vector $s$ and predicate target $r$, let $C_{s,r}$
be its best score across the methods and independent solver runs fixed
before testing. The reported classical statistic and paired gap are
\begin{equation}
 \operatorname{med}_s\overline C_s,\qquad
 \operatorname{med}_s\Delta_s,\qquad
 \overline C_s=\tfrac12(C_{s,1}+C_{s,2}),\quad
 \Delta_s=\mu_Q-\overline C_s.
\end{equation}
Each independently sampled coefficient vector contributes one statistical
unit, after averaging over its two predicate targets and combining solver
restarts as above.

\end{document}